\documentclass[11pt]{article}

\usepackage[T1]{fontenc} 
\usepackage[margin=1in]{geometry}
\usepackage{mathtools}
\usepackage{color}
\usepackage{proof}
\usepackage{caption}
\usepackage{bbm}
\usepackage{amssymb}
\usepackage{microtype}
\usepackage[dvipsnames]{xcolor}
\usepackage{enumitem}
\usepackage{amsthm}
\usepackage{thmtools,thm-restate}
\usepackage{mdframed}

\usepackage{amsmath}
\usepackage{amsfonts}
\usepackage{url}

\usepackage[colorlinks,citecolor=blue]{hyperref}
\newtheorem{theorem}{Theorem}[section]
\newtheorem{lemma}[theorem]{Lemma}
\newtheorem{definition}[theorem]{Definition}
\newtheorem{proposition}[theorem]{Proposition}
\newtheorem{corollary}[theorem]{Corollary}

\newtheorem{assumption}[theorem]{Assumption}

\newcounter{prob}
\newtheorem{problem}[prob]{Problem}

\newcommand{\alert}[1]{{#1}}

\newcommand{\eq}[1]{\hyperref[eq:#1]{(\ref*{eq:#1})}}
\renewcommand{\sec}[1]{\hyperref[sec:#1]{Section~\ref*{sec:#1}}}
\newcommand{\thm}[1]{\hyperref[thm:#1]{Theorem~\ref*{thm:#1}}}
\newcommand{\lem}[1]{\hyperref[lem:#1]{Lemma~\ref*{lem:#1}}}
\newcommand{\cor}[1]{\hyperref[cor:#1]{Corollary~\ref*{cor:#1}}}
\newcommand{\prp}[1]{\hyperref[prp:#1]{Proposition~\ref*{prp:#1}}}
\newcommand{\itm}[1]{\hyperref[itm:#1]{\ref*{itm:#1}}}
\newcommand{\app}[1]{\hyperref[app:#1]{Appendix~\ref*{app:#1}}}
\newcommand{\dfn}[1]{\hyperref[dfn:#1]{Definition~\ref*{dfn:#1}}}
\newcommand{\fig}[1]{\hyperref[fig:#1]{Figure~\ref*{fig:#1}}}
\newcommand{\clm}[1]{\hyperref[clm:#1]{Claim~\ref*{clm:#1}}}
\newcommand{\alg}[1]{\hyperref[alg:#1]{Algorithm~\ref*{alg:#1}}}
\newcommand{\stp}[1]{\hyperref[stp:#1]{Step~\ref*{stp:#1}}}
\newcommand{\asm}[1]{\hyperref[asm:#1]{Assumption~\ref*{asm:#1}}}
\newcommand{\prot}[1]{\hyperref[prot:#1]{Protocol~\ref*{prot:#1}}}
\newcommand{\prob}[1]{\hyperref[prob:#1]{Problem~\ref*{prob:#1}}}
\newcommand{\rmk}[1]{\hyperref[rmk:#1]{Remark~\ref*{rmk:#1}}}

\newcommand{\bra}[1]{\langle #1 \vert}
\newcommand{\ket}[1]{\vert #1 \rangle}
\newcommand{\proj}[1]{\vert #1\rangle\!\langle #1\vert}
\newcommand{\tr}[0]{\mathrm{tr}}
\newcommand{\op}[0]{\mathrm{op}}

\let\originalleft\left
\let\originalright\right
\renewcommand{\left}{\mathopen{}\mathclose\bgroup\originalleft}
\renewcommand{\right}{\aftergroup\egroup\originalright}

\renewcommand{\d}[0]{\mathrm{d}}
\newcommand{\A}[0]{\mathcal{A}}
\newcommand{\B}[0]{\mathcal{B}}
\newcommand{\C}[0]{\mathcal{C}}
\newcommand{\D}[0]{\mathcal{D}}

\newcommand{\F}[0]{\mathcal{F}}
\newcommand{\G}[0]{\mathcal{G}}
\renewcommand{\H}[0]{\mathcal{H}}
\newcommand{\K}[0]{\mathcal{K}}
\newcommand{\M}[0]{\mathcal{M}}

\renewcommand{\O}[0]{\mathcal{O}}
\newcommand{\Q}[0]{\mathcal{Q}}
\newcommand{\R}[0]{\mathcal{R}}
\newcommand{\T}[0]{\mathcal{T}}
\newcommand{\U}[0]{\mathcal{U}}
\newcommand{\V}[0]{\mathcal{V}}

\newcommand{\X}[0]{\mathcal{X}}

\DeclareMathOperator*{\Exp}{\mathbb{E}}
\DeclareMathOperator{\poly}{poly}
\DeclareMathOperator{\negl}{negl}

\newcommand{\NP}[0]{\ensuremath{\mathsf{NP}}\xspace}
\newcommand{\QMA}[0]{\mathsf{QMA}}
\newcommand{\coNP}[0]{\mathsf{coNP}}

\usepackage{xspace}
\newcommand{\BQP}[0]{\ensuremath{\mathsf{BQP}}\xspace}
\newcommand{\MA}[0]{\ensuremath{\mathsf{MA}}\xspace}
\newcommand{\AM}[0]{\ensuremath{\mathsf{AM}}\xspace}

\newcommand{\bit}{\{0,1\}}
\DeclareMathOperator*{\Span}{\mathrm{span}}

\newcommand{\QCAM}[0]{\mathsf{QCAM}}
\newcommand{\QCMA}[0]{\mathsf{QCMA}}
\newcommand{\cpoly}[0]{\mathsf{poly}}
\newcommand{\qpoly}[0]{\mathsf{qpoly}}

\renewcommand{\P}[0]{\mathcal{P}}
\renewcommand{\NP}[0]{\mathsf{NP}}

\newcommand{\QSZK}[0]{\mathsf{QSZK}}

\newcommand{\ALL}[0]{\mathsf{ALL}}
\newcommand{\PSPACE}[0]{\mathsf{PSPACE}}
\newcommand{\SZK}[0]{\mathsf{SZK}}
\newcommand{\Id}[0]{\mathbbm{1}}
\newcommand{\Ext}[0]{\mathrm{Ext}}

\newcommand{\XHOG}[0]{\mathrm{XHOG}}
\newcommand{\TIME}[0]{\mathsf{TIME}}
\newcommand{\LLQSV}[0]{\mathrm{LLQSV}}
\newcommand{\BFBD}[0]{\mathrm{BFBD}}
\newcommand{\LLHA}[0]{\textsc{LLHA}}
\newcommand{\LXEB}[0]{\textsc{LXEB}}
\newcommand{\adv}[0]{\mathrm{adv}}

\newcommand{\Dir}[0]{\mathrm{Dir}}

\newcommand{\diag}[0]{\mathrm{diag}}
\newcommand{\Haar}[0]{\mathrm{Haar}}
\newcommand{\CPTP}[0]{\mathrm{CPTP}}

\newcommand{\Bernoulli}[0]{\mathsf{Bernoulli}}
\newcommand{\Binomial}[0]{\mathsf{Binomial}}
\newcommand{\QSD}[0]{\mathrm{QSD}}
\newcommand{\NISZK}[0]{\mathsf{NISZK}}
\newcommand{\QIP}[0]{\mathsf{QIP}}

\renewcommand{\epsilon}[0]{\varepsilon}

\usepackage{environ}

\begin{document}
\title{Certified Randomness from Quantum Supremacy}
\author{Scott~Aaronson\footnote{University of Texas at Austin. \ Email: \textsf{aaronson@cs.utexas.edu}. \ Supported by a Vannevar Bush Fellowship
from the US Department of Defense, the Berkeley NSF-QLCI CIQC Center, a Simons Investigator Award, and the
Simons ``It from Qubit'' collaboration.} 
\and 
Shih-Han~Hung\footnote{University of Texas at Austin. Email: \textsf{shung@cs.utexas.edu}. \ Supported by the U.S. Department of Energy, Office of Science, National Quantum Information Science Research Centers, Quantum Systems Accelerator.}}
\date{}

\maketitle

\begin{abstract}
We propose an application for near-term quantum devices: namely, generating cryptographically certified random bits, to use (for example) in proof-of-stake cryptocurrencies. \ Our protocol repurposes the existing ``quantum supremacy'' experiments, based on random circuit sampling, that Google and USTC have successfully carried out starting in 2019. \ We show that, whenever the outputs of these experiments pass the now-standard Linear Cross-Entropy Benchmark (LXEB), under plausible hardness assumptions they necessarily contain $\Omega(n)$ min-entropy, where $n$ is the number of qubits. \ To achieve a net gain in randomness, we use a small random seed to produce pseudorandom challenge circuits. \ In response to the challenge circuits, the quantum computer generates output strings that, after verification, can then be fed into a randomness extractor to produce certified nearly-uniform bits---thereby ``bootstrapping'' from pseudorandomness to genuine randomness. \ We prove our protocol sound in two senses: (i) under a hardness assumption called Long List Quantum Supremacy Verification, which we justify in the random oracle model, and (ii) unconditionally in the random oracle model against an eavesdropper who could share arbitrary entanglement with the device. \ (Note that our protocol's output is unpredictable even to \emph{a computationally unbounded adversary who can see the random oracle}.) \ Currently, the central drawback of our protocol is the exponential cost of verification, which in practice will limit its implementation to at most $n\sim 60$ qubits, a regime where attacks are expensive but not impossible. \ Modulo that drawback, our protocol appears to be the \emph{only} practical application of quantum computing that both requires a QC and is physically realizable today.
\end{abstract}

\newpage
{
  \hypersetup{linkcolor=black}
  \tableofcontents
}
\newpage

\section{Introduction}

After three decades of quantum computing theory and experiment, the world finally has noisy quantum devices, with $50-60$ qubits or $\sim 100$ photons, that solve special sampling problems in a way that's conjectured to outperform any existing classical computer. \ The devices include Google's $53$-qubit ``Sycamore'' chip \cite{google19}, USTC's ``Jiuzhang'' \cite{ustc20} and ``Zu Chongzhi'' \cite{ustc21}, and most recently Xanadu's ``Borealis'' \cite{xanadu22}. \ The sampling problems, which include Random Circuit Sampling \cite{AC16} and BosonSampling \cite{AA11}, grew directly out of work in quantum complexity theory beginning around 2010.

To be clear, it's still debated in which senses current devices have achieved the milestone of ``quantum supremacy''---a term coined by Preskill \cite{Pre12} in 2012, to refer to an orders-of-magnitude speedup over all known classical approaches for some well-defined (but not necessarily useful) computational task. \ On the one hand, since Google's original 2019 announcement \cite{google19}, the sampling experiments have continued to improve, for example in number of qubits and circuit depth (for RCS) \cite{ustc21}, and in number of photons and measurement fidelity (for BosonSampling) \cite{xanadu22}. \ One expects further improvements. \ On the other hand, classical spoofing attacks against the experiments have \emph{also} improved---with some attacks based on tensor-network contraction (e.g., \cite{PCZ21}), and others taking advantage of noise in the devices (e.g., \cite{BCG20}). \ Notably, however, the attacks that fully replicate the Google device's observed performance, such as that of \cite{PCZ21}, still have inherently exponential scaling, and still seem to require an ExaFLOPS supercomputer to match or beat the Google device's running time of $\sim 3$ minutes. \ As a rough estimate, the Summit supercomputer uses $13$ megawatts, while Google \cite[Appendix H]{google19} estimated that the dilution refrigerator for its $53$-qubit QC uses $\sim 20$ kilowatts. \ Thus, \emph{despite} the QC's extreme need for refrigeration, it still wins by a factor of hundreds as measured by electricity cost.

For some, the recent quantum supremacy demonstrations were important mostly because they showcased many of the key ingredients of a future \emph{fault-tolerant, scalable} quantum computer---and just as importantly, did not detect any correlated errors of the sort that would render fault-tolerant quantum computing impossible. \ For others, however, these experiments have done more: namely, they've inaugurated the era of ``NISQ'' or Noisy Intermediate Scale Quantum computation, another term coined by Preskill \cite{Pre18}. \ The hope of NISQ is that, even \emph{before} fault-tolerance is achieved, noisy QCs with up to (say) $1000$ qubits might already prove useful for certain practical problems, just like various analog computing devices were useful even before the invention of the transistor inaugurated the digital era.

Unfortunately, despite the billions that have by now been invested into NISQ hopes, the lack of any obvious ``killer app'' for NISQ devices has emerged as a defining fact of the field.

Perhaps NISQ devices will be useful for simulation of condensed-matter physics or even quantum chemistry. \ Alas, while there are exciting proposals for quantum simulations that would need only a few hundred qubits (e.g., \cite{RWSWT17}), these proposals invariably have the drawback of requiring thousands or millions of layers of gates. \ Unless it can be remedied, this would put them completely out of reach for NISQ devices. \ Or perhaps NISQ devices will yield speedups for optimization problems, via quantum annealing or QAOA \cite{FGG14}. \ Alas, despite years of intense theoretical and empirical work, researchers have struggled to show any clear advantage for quantum annealing or QAOA over classical computing, for any practical optimization problem---\emph{let alone} an advantage that would be achievable on a NISQ device.

We should add that, in recent years, there have been striking \emph{new} ideas for how to demonstrate quantum supremacy. \ These include interactive protocols that exploit trapdoor one-way functions \cite{BCMVV18,KCVY22}, as well as the spectacular result of Yamakawa and Zhandry \cite{YZ22}, which gave an exponential quantum speedup for an $\mathsf{NP}$ search problem relative to a random oracle. \ Alas, pending some breakthrough, none of these ideas seem to be implementable on a NISQ device.

\subsection{Our Contribution}

This paper studies, to our knowledge for the first time, whether current sampling-based supremacy experiments might \emph{themselves} have a useful application outside of physics.\footnote{Some earlier work explored whether BosonSampling might be useful for (e.g.) calculating molecular vibronic spectra \cite{HGP+15} or graph similarity detection \cite{SBI+20}, but those hopes were killed by efficient classical simulations.} \ We focus on the generation of \emph{cryptographically certified random bits}.

Needless to say, it is easy to use a quantum computer---or for that matter, even a Geiger counter next to some radioactive material---to generate as many random bits as we like: bits that quantum mechanics itself predicts will be fundamentally unpredictable. \ The problem is, how do we convince a skeptic over the Internet, with no access to our hardware, that the bits were indeed random, and not secretly backdoored? \ This is not just a theoretical worry: for example, as a byproduct of the Edward Snowden revelations in 2013, the world learned that a NIST pseudorandomness standard known as Dual\_EC\_DRBG was indeed backdoored by the US National Security Agency.\footnote{https://en.wikipedia.org/wiki/Dual\_EC\_DRBG}

Certified randomness has become a significant practical problem---particularly with the rise of \emph{proof-of-stake cryptocurrencies}, which notably include Ethereum\footnote{https://en.wikipedia.org/wiki/Ethereum} (market cap at time of writing: \$163 billion), following its migration on September 15, 2022. \ In proof-of-stake systems, lotteries are continually run to decide which currency holder gets to add the next block to the blockchain. \ There is no trusted authority to manage these lotteries, yet the entire system rests on the assumption that they are conducted honestly and without bias. \ Other applications of certified randomness include non-interactive zero-knowledge proofs, and financial and electoral audits.

One approach to the certified-randomness problem uses blockchains themselves as a source of random bits---with the argument being that anyone who could predict the bits could exploit their predictability to get rich \cite{BCG15}. \ Other approaches look to the social or natural worlds for a source of publicly verifiable entropy: for example, perhaps one could use the least significant digits of the Dow Jones Industrial Average, or the patterns of granules that form on the surface of the Sun.

More relevant for us, since 2009, an exciting line of work has shown how to use measurements on entangled particles as a source of physically certified randomness \cite{Col09,PAM+10,VV12,CY14,MS16,MS17}. \ The idea is that, if the measurement outcomes are observed to violate the Bell/CHSH inequality, then by that very fact, the outcomes \emph{cannot} have been secretly deterministic, unless there was secret communication between the ``Alice'' and ``Bob'' detectors. \ Furthermore, depending on the experimental setup, this communication might need to have occurred faster than light. \ Thus, the outcomes must contain genuine entropy, which can be fed into a randomness extractor to purify it into nearly-uniform random bits. \ The technical part is that, in a Bell/CHSH experiment, the measurement bases must \emph{themselves} be unpredictable---and thus, they need to be chosen judiciously if we want an overall net gain in randomness. \ This is the problem that the line of works \cite{Col09,PAM+10,VV12,CY14,MS16,MS17} has now almost completely solved.

Bell/CHSH-based certified randomness protocols have already been experimentally demonstrated \cite{BEG+18}, and are even in consideration for practical deployment in the NIST Randomness Beacon \cite{nist19}, which generates $512$ random bits per minute.

The central drawback of these protocols is that a user, downloading allegedly random bits from the Internet, has no obvious way to verify that the ``Alice'' and ``Bob'' detectors were out of communication---the key assumption needed for security. \ Indeed, in some Bell/CHSH experiments, ``Alice'' and ``Bob'' are mere feet away! \ But even if they weren't, how would this be proved?

The central insight of this paper is that sampling-based quantum supremacy experiments provide an entirely different route to certified randomness---a route that requires only a single quantum device, while also being practical today. \ In our protocol, a classical verifier uses a small random seed to generate $n$-qubit challenge circuits $C_1,C_2,\ldots$ pseudorandomly. \ The verifier then submits these $C_i$'s one at a time, presumably over the standard Internet, to a quantum computer server. \ For each $C_i$, the server needs to respond quickly---say, in less than one second---with independent samples $s_1,\ldots,s_k$ from $C_i$'s output distribution: that is, the distribution over $\{0,1\}^n$ obtained by running $C_i$ on the initial state $\ket{0^n}$ and then measuring in the computational basis.

The verifier, at its leisure, can then calculate the so-called \emph{Linear Cross-Entropy Benchmark},
$$ \operatorname{LXEB} := \sum_{j=1}^{k} |\bra{s_j}C_i\ket{0^n}|^2, $$
for at least some of the challenge circuits $C_i$. \ \emph{If} the LXEB scores are sufficiently large, our analysis shows that the verifier can then be confident, under plausible computational assumptions, that there must be $\Omega(n)$ bits of genuine min-entropy in the returned samples.

In other words: even a quantum computer should need $\exp(n)$ time to generate samples that pass the LXEB test and yet are secretly deterministic or nearly-deterministic functions of $C_i$. \ For a typical circuit $C_i$, an honest sample from the output distribution will contain $n-O(\log n)$ bits of min-entropy. \ A dishonest quantum computer could \emph{somewhat} reduce the entropy of the returned samples---for example, by generating many samples and then returning only those that start with $0$ bits. \ But doing better, by finding (e.g.) the lexicographically first samples that pass the LXEB test, or the samples that maximize the LXEB score, should be exponentially hard even quantumly, requiring amplitude amplification or the like (while a subexponential classical algorithm wouldn't stand a chance). \ The purpose of our security reductions, which we will explain in detail in \sec{techover}, is just to formalize these simple intuitions.

Assuming the returned samples (or enough of them) pass the LXEB test, the last step of our protocol is to feed them into a classical \emph{seeded randomness extractor}, to produce output bits that are exponentially close in total variation distance to uniformly random.

Stepping back, many people have pointed out the close analogy between
\begin{enumerate}
\item the Bell/CHSH experiments, which ruled out local hidden-variable theories (and which have now been recognized with the Nobel Prize in Physics), and
\item sampling-based quantum supremacy experiments, which seek to rule out ``classical polynomial-time hidden-variable theories.''
\end{enumerate}
This paper shows that the analogy goes even further. \ In both cases, the original purpose of the experiment was just to demonstrate the reality of some quantum phenomenon, and rule out any classical explanation---but we then get certified randomness as a ``free byproduct'' of the demonstration. \ In both cases, the entire setup hinges on a numerical inequality---one that any classical explanation must satisfy, that quantum mechanics predicts can be violated by a large amount, and that realistic experiments can violate albeit by less than the maximum that quantum mechanics predicts. \ In both cases, \emph{any} violation of the inequality turns out to suffice for the certified randomness application.

We note, lastly, that our protocol \emph{inherently requires} the use of a quantum computer. \ This can be seen as follows: consider any server that's simulable in classical probabilistic polynomial-time. \ Then by definition, there can be no efficient way to distinguish that server from a simulation \emph{whose randomness has been replaced by the output of a pseudorandom generator}. \ Indeed, if the pseudorandom generator has an $m$-bit seed, then the best distinguishing algorithm would be expected to take $\exp(m)$ time---which means that even given the $\sim 2^n$ time that we allow for verification, the verifier \emph{still} cannot distinguish an honest server from one with only $m$ bits of true entropy, for any $m\gg n$.

How does our actual quantum protocol evade the above impossibility argument? \ Simply by a fact used again and again in quantum complexity theory: namely, that there is no notion of ``pulling the randomness'' (or quantumness) out of a quantum algorithm, for example to replace it with pseudorandomness, analogous to what is possible with classical randomized algorithms. \ One could also say: our protocol's security analysis will depend on a computational assumption, that the problem of ``Long List Quantum Supremacy Verification'' is hard for the complexity class $\QCAM$, whose classical analogue is simply false. \ The reasons for this, in turn, are closely related to one of the elemental differences between classical and quantum computation, that $\mathsf{PostBPP}$ ($\mathsf{BPP}$ with postselected outputs) is contained in the polynomial hierarchy and can be simulated using approximate counting, whereas $\mathsf{PostBQP}=\mathsf{PP}$ can express $\mathsf{\#P}$-complete problems.

\subsection{Practical Considerations}

Our certified randomness protocol could be demonstrated on existing devices, with $n=60$ qubits or some other number in the ``quantum supremacy regime.'' \ However, there are practical and even conceptual issues to be sorted out before deploying the protocol for proof-of-stake cryptocurrency or any other critical application.

\textbf{Verification cost.} The central drawback of our protocol, as it stands, is that to check the server's outputs, the classical verifier needs to calculate a Linear Cross-Entropy score, and this is expected to take $\sim 2^n$ time---similar to the time needed for classical spoofing. \ This drawback is directly inherited from Random Circuit Sampling and all other current approaches to NISQ quantum supremacy itself.

Because of the verification cost, $n$, the number of qubits, must be chosen small enough that $2^n$ is still within range of the most powerful classical supercomputers available. \ If so, however, the issue is obvious: $2^n$ would \emph{also} be within range of a sufficiently dedicated classical spoofer, who could then predict and control the allegedly random bits.

Nevertheless, we claim that not all hope is lost. \ The crucial observation is that spoofing, to be effective, needs to be \emph{continual}: for example, if the challenge circuits are submitted every second, then the spoofer needs to run nearly every second as well. \ Even if a real quantum computer were used (say) every \emph{other} second, the outputs would contain a lot of genuine min-entropy, which would suffice for a secure protocol. \ The spoofing also needs to be \emph{fast}---as fast as the QC itself.

One might object that, since most classical algorithms to simulate quantum circuits are highly parallelizable, spoofing our protocol within some exacting time limit is ``merely'' a matter of spending enough money on classical computing hardware. \ When (say) $n=60$, though, we estimate that the expenditure, to do $\exp(60)$ operations per second, would run into billions of dollars, outside the means of all but corporations and nation-states.

Verification, by contrast, only needs to be \emph{occasional}. \ Using a tiny amount of seed randomness, the verifier can choose $O(1)$ random rounds of the protocol and spot-check only those. \ Then a malicious server that spoofed even (say) $10\%$ of the rounds would be caught with overwhelming probability. \ Verification can also be done \emph{at leisure}: so long as the verifier is satisfied to catch the spoofer after the fact, the verifier could spend hours or days where the spoofer needed to take less than a second. \ Indeed, to keep the server honest, arguably the verification need not even be done: it's enough to threaten credibly that it \emph{might} be done!

Having said that, of course it would be preferable if the verification could be done in $n^{O(1)}$ time, in some way that retained the protocol's ``NISQiness.''

In our view, whether it's possible to achieve sampling-based quantum supremacy, \emph{on a NISQ device and with efficient classical verification}, has become one of the most urgent open problems in quantum computing theory, even independently of this work. \ Our work further underscores the problem's importance, by showing how a solution could turn secure, practical certified randomness into the first real application of quantum computers.

\textbf{Interactivity.} A second practical issue with our protocol is the need for the verifier continually to generate new challenge circuits that are unpredictable to the quantum computing server. \ One could reasonably ask: if the verifier has that ability, then why does it even \emph{need} the quantum computer to generate random bits?

The short answer is that our protocol offers an ``upgrade'' in the level of unpredictability: the challenge circuits only need to be \emph{pseudorandom} (for reasons to be explained later, against a $\QSZK$ adversary). \ So in particular, the verifier can generate all the circuits deterministically from a single initial random seed. \ The protocol's output, by contrast, is guaranteed to be \emph{genuinely} random.

Indeed, our protocol offers an appealing ``forward secrecy'' property. \ Namely, even if we imagine that the verifier's pseudorandom generator will be broken in the future, so long as the server can't break the PRG \emph{at the time the protocol is run}, the server is forced to generate truly random bits. \ Such bits will of course remain unpredictable, conditioned on anything that doesn't depend on themselves, regardless of any future advances in cryptanalysis.

\textbf{Who verifies the verifier?} Still, there remains a difficulty: the verifier checks the QC's outputs, but who checks the verifier? \ If the verifier just wants random bits for its own private use, then there is no problem: the verifier could use our protocol, for example, to check random bits output by a QC that was bought from an untrusted manufacturer. \ But consider an application like proof-of-stake cryptocurrency, where the certified random bits need to be shared with the world. \ Does the world designate some organization to play the role of the verifier? \ If so, then why couldn't that organization be corrupted or infiltrated, as surely as the organization that owns the quantum computer---bringing us back where we started?

Classical cryptography suggests a variety of potential solutions to this dilemma. \ For example, perhaps a dozen or more classical verifiers each generate their own pseudorandom sequences, and those sequences are then XORed together to produce a single sequence which is used to generate the challenge circuits to send to the quantum computing server. \ If even one verifier wants the sequence to be unpredictable to the server, then it will be, provided that no verifier can see any other verifier's sequence before committing to its own. 

Again one could ask: if we trust such a XOR protocol, then why not just use its outputs directly, and skip the quantum computer? \ Again our answer appeals to the ``randomness bootstrap'': provided we agree that the XOR'ed sequence is unpredictable \emph{in practice, for now}, the quantum computer's output will then be \emph{fundamentally} unpredictable. \ Our protocol thus provides an upgrade in the level of unpredictability.

\subsection{Related Work}

We are not the first to propose using a quantum computer to generate certified random bits, which are secure under some computational hardness assumption. \ Brakerski, Christiano, Mahadev, Vazirani, and Vidick \cite{BCMVV18} gave an elegant scheme based on the assumed hardness of the Learning With Errors (LWE) problem. \ In subsequent work, Mahadev, Vazirani, and Vidick \cite{MVV22} showed that the Brakerski et al.\ protocol generates $\Omega(n)$ random bits per round, which matches our protocol.

The central advantage of the Brakerski et al.\ protocol over ours is that its outputs can be verified in classical polynomial time. \ On the other hand, unlike ours, their protocol seems difficult or impossible to implement on a NISQ device, because it requires evaluating complicated cryptographic functions on superpositions of inputs. \ In addition, their protocol requires the quantum computer to maintain a coherent superposition state \emph{while it interacts with the verifier}, presumably over the Internet. \ This is currently feasible only with certain hardware platforms, such as trapped ions, and not for example with superconducting qubits (whose coherence times are measured in microseconds).

More recently, as a byproduct of their breakthrough on an exponential quantum speedup for $\NP$ search relative to a random oracle, Yamakawa and Zhandry \cite{YZ22} gave a different interactive protocol to certify $\Omega(\log n)$ random bits, in the random oracle model and also assuming the so-called \emph{Aaronson-Ambainis conjecture} \cite{AA14}. \ We do not know whether the Yamakawa-Zhandry protocol remains secure against an entangling adversary, nor whether it accumulates entropy across multiple rounds. \ In any case, theirs is again a protocol that evaluates complicated functions on superpositions of inputs, meaning there is little or no hope of running it on a NISQ device.

In contrast to these earlier works, here we pursue the ``minimalist approach'' to generating certified randomness using a quantum computer: we eschew all cryptography done in superposition, and just examine the output distributions of random or pseudorandom quantum circuits. \ By taking this route, we give up on efficient verification, but we gain feasibility on current hardware, as well as a conceptual unification of certified randomness with sampling-based quantum supremacy itself.

Recently, building on the unpublished announcements by one of us (SA) of the research now reported in this paper, Bassirian, Bouland, Fefferman, Gunn, and Tal \cite{BBFGT21} took some first steps toward analyzing the use of sampling-based quantum supremacy experiments for certified randomness. \ Their first result says that, relative to a random oracle, any efficient quantum algorithm for Fourier Sampling must generate $\Omega(n)$ bits of min-entropy as a byproduct of its operation. \ Their second result says that Long List Quantum Supremacy Verification (LLQSV), the problem that underlies our hardness reduction, is neither in $\BQP$ nor in $\mathsf{PH}$ relative to a random oracle. \ To prove non-containment in $\mathsf{PH}$, they build on the breakthrough oracle separation between $\BQP$ and $\mathsf{PH}$ due to Raz and Tal \cite{RT22}.

These results are of course closely related to ours, but they fall short of a soundness analysis for our certified randomness protocol. \ We go further than \cite{BBFGT21} in at least four respects:
\begin{enumerate}
\item We prove that a plausible hardness assumption about LLQSV implies the generation of certified random bits. \ This reduction does not depend on a random oracle.
\item We give black-box evidence for that hardness assumption. \ (The result of \cite{BBFGT21}, that black-box LLQSV is not in $\mathsf{PH}$, is interesting and new, but neither necessary nor sufficient for us. \ As we'll explain, we need non-containment in the class $\mathsf{QCAM/qpoly}$.)
\item We prove the accumulation of entropy across multiple rounds.
\item In the black-box setting, we prove security against an entangled adversary.
\end{enumerate}
Our proof techniques are also independent of those in \cite{BBFGT21}.

Lastly, let us mention that Brand{\~a}o and Peralta \cite{BP20} have already reported numerical calculations to find appropriate parameter settings for the protocol described in this paper.

\subsection{This Paper's History}

One of us (SA) conceived the certified randomness protocol, as well as basic elements of its soundness analysis (e.g., the $\operatorname{LLQSV}\not\in\mathsf{QCAM/qpoly}$ hardness assumption), in February 2018. \ SA then gave various public talks about the proposal (e.g., \cite{Aar18b}), albeit only sketching the analysis. \ Those talks influenced subsequent work on quantum supremacy: for example, the Google group cited them as motivation in its 2019 paper announcing its $53$-qubit Sycamore experiment \cite{google19}.

Alas, the soundness analysis ended up being too involved for SA to complete alone. \ That and other factors caused a more than four-year delay in writing up this paper. \ Here, we not only complete the analysis that SA announced in 2018: we also prove security, in the random oracle model, \emph{against an adversary who could be arbitrarily entangled with the QC}. \ This goes beyond what SA claimed in 2018, and indeed addresses one of the central open problems raised at that time.

\subsection{Future Directions}

Many important problems remain:

\begin{itemize}
    \item As mentioned before, perhaps the biggest problem is to design a sampling-based quantum supremacy experiment that both runs on a NISQ device \emph{and} admits efficient classical verification. \ If such an experiment were developed, then based on our results here, we predict that it could be readily repurposed to get a secure, efficiently-verifiable certified randomness scheme that runs on existing devices.
    \item Short of that, it would also be interesting to adapt our randomness protocol from Random Circuit Sampling (RCS) to other known quantum supremacy proposals, such as BosonSampling \cite{AA11} and $\mathsf{IQP}$ \cite{BJS11}. \ With BosonSampling, the problem is that we currently lack a crisp, quantitative conjecture about the best that a polynomial-time classical algorithm can do to spoof tests such as the Linear Cross-Entropy Benchmark (LXEB). \ From 2013 work of Aaronson and Arkhipov \cite{AA13}, we know that, in contrast to what we conjecture for RCS, efficient classical algorithms can get \emph{some} depth-independent, $\Omega(1)$ LXEB advantage for BosonSampling, but how much? \ Answering this question seems like a prerequisite to designing a certified randomness protocol, as it would set the lower bound on how well a BosonSampling experiment has to do before it can be used for such a protocol.
    \item Of course, it would be great to know more about the truth or falsehood of the central hardness conjectures on which we base our protocol's security---e.g., that ``Long List Quantum Supremacy Verification'' (LLQSV) lacks a $\mathsf{QCAM/qpoly}$ protocol. \ It would be also great to prove our protocol's security under weaker assumptions. \ Can we at least remove the exponentially long list of circuits, and use a hardness assumption involving a single circuit?
    \item In the setting with an entangled adversary, we can currently prove security \emph{only} in the random oracle model. \ Can we state a plausible hardness assumption that suffices for that setting?
    \item Under some plausible hardness assumption, can we tighten the lower bound on the amount of min-entropy generated per sample---even up to the maximum of $n-O(\log n)$?\footnote{$n-O(\log n)$ is the maximum because the quantum computer could always (say) generate $n^{O(1)}$ samples from the correct distribution until it finds one whose first $O(\log n)$ bits are all $0$'s.}
    \item Likewise, can we show that more and more min-entropy continues to be generated, even if we sample with the same circuit $C$ over and over? \ Clearly there is a limit here: once enough time has elapsed that a spoofer could have explicitly calculated $C$'s entire output distribution, and perhaps even stored it in a giant lookup table, $C$ is no longer secure and needs to be replaced by a new circuit. \ But can we at least go up to that limit? \ To whatever extent we can, our protocol would become much more efficient in practice---especially once we factor in (e.g.) the time needed to calibrate a superconducting QC on a new circuit $C$.
    \item We know, both from the Haar-random approximation and from extensive numerical evidence, that an ideal, honest QC \textit{does} succeed at the Linear Cross-Entropy Benchmark with overwhelming probability, given a random quantum circuit $C$ as input. \ And in some sense, since this fact is never needed in our security analysis, empirical evidence suffices for it! \ All the same, it is strange that a rigorous proof of the fact is still lacking, at least for ``natural'' quantum circuit ensembles. \ In \sec{perfect}, we prove the fact in the random oracle model, but can we prove it outright? \ Recent advances showing that random quantum circuits yield $t$-designs \cite{BHH16,Haf22} take us part of the way, but an additional idea seems needed.
\end{itemize}

\section{Technical Overview}\label{sec:techover}

\subsection{Our Basic Result (without entangled adversary)}\label{sec:intro-no-side-info}

Throughout this paper, we let $C$ be a quantum circuit acting on $n$ qubits, and we let $N=2^n$ be the Hilbert space dimension. \ Let $P_C$ be the probability distribution defined by $p_C(z)=|\bra{z}C\ket{0^n}|^2$. \
It is well-known that when $C\sim\Haar(N)$, the Haar measure over $N\times N$ unitary matrices, we have
\begin{align}
\Exp_C\Exp_{z\sim P_C}[p_C(z)]=\frac{2}{N+1}.\label{eq:haareq}
\end{align}
In 2019, Google \cite{google19} announced an experiment to show quantum advantage on the following task, called Linear Cross-Entropy Benchmarking (LXEB).
\begin{problem}[Linear Cross-Entropy Benchmarking {$\LXEB_{b,k}(\D)$}]\label{prob:lxeb}
  Let $\D$ be a probability distribution over quantum circuits on $n$ qubits. \ Then the $\LXEB_{b,k}(\D)$ problem is as follows: given $C$ drawn from $\D$, output samples $z_1,\ldots,z_k\in \{0,1\}^n$ such that 
  \begin{align}
    \frac{1}{k}\sum_{i=1}^k p_C(z_i) \geq \frac{b}{N}.
  \end{align}
  We sometimes omit the argument $\D$.
\end{problem}

Intuitively, we expect that a polynomial-time classical algorithm should be unable to solve $\LXEB_{b,k}$ for any $b=1+\Omega(1)$. \ By contrast, if an ideal quantum computer simply runs $C$ over and over on the initial state $\ket{0^n}$, measures in the computational basis, and returns the results, then approximating $C$ by a random unitary, by \eq{haareq} we expect the QC to solve $\LXEB_{b,k}$ with $1-1/\exp(k)$ success probability for any constant $b<2$. \ Meanwhile, a noisy QC could be expected to solve $\LXEB_{b,k}$ for some $b$ greater than $1$ but less than $2$---and indeed that's exactly what's observed empirically, with (for example) Google's 2019 experiment achieving $b\sim 1.002$.

In this paper, our aim is to show, not merely a quantum advantage over classical in solving $\LXEB_{b,k}$, but a quantum sampling advantage over \emph{any efficient algorithm---quantum or classical---that returns the same $s_i$'s a large fraction of the time when given the same circuit $C$.}

To do this, we'll use a new and admittedly nonstandard hardness assumption, but one that strikes us as extremely plausible. \ Our assumption concerns the following problem:

\begin{problem}[Long List Quantum Supremacy Verification $\LLQSV(\U)$]\label{prob:llqsv}
  We are given oracle access to $M=O(2^{3n})$ quantum circuits $C_1,\ldots,C_M$, each on $n$ qubits, which are promised to be drawn independently from the distribution $\U$. \ We're also given oracle access to $M$ strings $s_1,,\ldots,s_M\in \{0,1\}^n$. \ Then the task is to distinguish the following two cases:
  \begin{enumerate}
  \item \textbf{No-Case}: Each $s_i$ is sampled uniformly and uniformly from $\bit^n$.
  \item \textbf{Yes-Case}: Each $s_i$ is sampled from $p_{C_i}$, the output distribution of $C_i$.
  \end{enumerate}
\end{problem}

Our hardness assumption, which we call the Long List Hardness Assumption ($\LLHA_B(\U)$), now says the following, for some parameter $B<n$: 
\begin{align}
  \LLQSV(\U)\notin\QCAM\TIME(2^B)/\mathsf{q}(2^B n^{O(1)}).
\end{align}

Here $\QCAM$, or Quantum Classical Arthur Merlin, is the class of problems that admit an $\AM$ protocols with classical communication and a quantum verifier. \ $\QCAM\TIME(T)$ is the generalization of $\QCAM$ where the verifier can use running time $T$ (the communication is still restricted to be polynomial). \ $\QCAM\TIME(T)/\mathsf{q}(A)$ is the same, but where the verifier now receives $A$ bits of quantum advice that depend only on $n$.

Our first main result is then the following.

\begin{theorem}[Single-round analysis, no side information, informal]
  Let $\U$ be a distribution over $n$-qubit quantum circuits, and suppose $\LLHA_B(\U)$ holds. \ Also, let $\A$ be a polynomial-time quantum algorithm that solves $\LXEB_{b,k}(\U)$ with probability at least $q$. \ Then $\A$'s output, $s_1,\ldots,s_k$, satisfies
\begin{align}\label{eq:llqsv-single-informal}
\Pr_{C'\sim\U}\left[H_{\min}(s_1\ldots s_k|C=C')
\geq B/2 \right]\geq  \left(\frac{bq-1}{b-1}-o(1)\right),
\end{align}
where $H_{\min}(\{p_i\}):=\min_i \log_2 \frac{1}{p_i}$ is the min-entropy.
\label{thm:thm1}
\end{theorem}

To illustrate, suppose we set $B:=0.49n$---the best \emph{upper} bound that we know, $B<n/2$, follows from Grover's algorithm. \ Suppose also that $b=1.002$, as in Google's experiment \cite{google19}, and that $k$ is chosen large enough so that $q\ge 0.9990$ by a large deviation inequality. \ Then \thm{thm1} is telling us that $\A$'s output must contain at least $(0.12-o(1))n$ random bits.

Interestingly, while \thm{thm1} is stated in terms of min-entropy, and while our eventual multi-round result will \emph{also} be stated in terms of min-entropy, as an intermediate step it's convenient to switch to Shannon entropy, as this is what entropy accumulation theorems use. \ Of course, since $H_{\min}(\D) \le H(\D)$ for every distribution $\D$, \thm{thm1} immediately implies the same lower bound on Shannon entropy. \ Indeed, since Shannon entropy behaves linearly with respect to expectation, \thm{thm1} implies that
\begin{align}
    H(s_1,\ldots,s_k|C) \geq \frac{B}{2}\cdot \left(\frac{bq-1}{b-1} - o(1)\right).
\end{align}

While $\LLHA_B(\U)$ is admittedly a strong assumption, our next result justifies it by proving that it holds in the random oracle model:

\begin{theorem}[Hardness of $\LLQSV(\U)$, informal]
  Given a random oracle $\O$, let $\U$ be the uniform distribution over $M=2^{O(n)}$ quantum circuits $C_1,\ldots,C_M$, which Fourier-sample disjoint Boolean functions $f_1,\ldots,f_M:\{0,1\}^n\rightarrow\{-1,+1\}$ respectively defined by $A$. \ Then $\LLHA_{B}(\U)$ holds relative to $\O$ for $B=\Omega(n)$.
\end{theorem}

\alert{
Here we outline the proof. \ 
First we give a reduction for $\LLQSV$ from another problem called Boolean Function Bias Detection (BFBD). \ 
In the latter problem, the algorithm is given access to $M$ functions sampled from either a distribution $\D$ or the uniform distribution. \ 
The distribution $\D$ can be described with the following process: \ 
First sample a integer $r\in\{0,1,\ldots,N\}$ with probability $N(1-2r/N)^2\cdot \binom{N}{r}2^{-N}$, sample a random subset $R\subseteq\bit^n$ of size $r$, and finally set $f(x)=-1$ if and only if $x\in R$. \ 
Since both distributions are concentrated around $r=N/2$, a simple hybrid argument leads to a basic lower bound for $\BQP$. \ 
We then extend the hardness result to interactive proof systems, specifically, the class $\QIP[2]$ of two-message quantum interactive proofs, using a similar argument: \ 
Recall that the prover's goal is always to convince the verifier that the given function is sampled from $\D$. \ 
From the observation stated above, we can modify a function $f\sim\D$ on a small number of points to yield a random function. \
Thus a prover which convinces the verifier to accept $f\sim\D$ would also convince the verifier to accept a random function. \ 
Then by the inclusion $\QCAM\subseteq\QIP[2]$, $\LLQSV$ is hard for $\QCAM$. \ 

To prove the desired hardness for the non-uniform class $\QCAM/\qpoly$ (or generalizations of it to use more queries and more advice), we observe that by Aaronson and Drucker's exchange theorem \cite{AD10}, $\QCAM/\qpoly\subseteq\QMA/\cpoly$, so it suffices to show hardness against the latter class. \ 
We then change the oracle model (and not the problem itself) as follows: \
The oracle $\O$ contains $N$ sections, each indexed by an $n$-bit string $x$. \ 
The non-uniform protocol given oracle access to $\O$, input $x$, and all the samples (also indexed by $x$), is challenged to determine whether the sample $s_x$ is sampled from $\O_x$ or uniform (see \prob{llqsv}). \ 

Clearly any solver for LLQSV can determine whether $s_x$ is sampled from $\O_x$ for every $x\in\bit^n$. \ 
By replacing the classical advice with a random guess, we show that any $\QMA/\cpoly$ verifier solving the problem would imply a $Nn^{O(1)}$-query $\QMA$ verifier solving $N$ problems with probability $2^{-\poly(n)}$. \ 
Then we appeal to the \emph{strong direct product theorem} by Sherstov \cite{She11}, who showed that even to achieve success probability $2^{-\Omega(N)}$, computing $N$ independent problems requires $\Omega(Nd)$ queries, where $d$ is the query lower bound of a single problem obtained using the polynomial method. \ 
Finally we derive a contradiction by showing that a single problem has an exponential lower bound. \ 
}

Building upon the above single-round analysis, the next step is to showk \emph{accumulation} of entropy across multiple rounds.
We give a simple $m$-round entropy accumulation process using $\LXEB_{1+\delta,k}$ as the verification for $m=n^{O(1)}$\footnote{Potentially $m$ can be exponentially large, but we do not pursue this here.} and constant $0<\delta<1$.
In each round, for $\gamma=O(\log n/m)$, the verifier sends the same circuit as in the previous round with probability $1-\gamma$, or samples a fresh random circuit with probability $\gamma$.
We define an \emph{epoch} to be an interval of consecutive rounds where the same circuit is sent.\footnote{Brakerski, Christiano, Mahadev, Vazirani and Vidick \cite{BCMVV18} used the same concept of ``epochs'' to analyze their certified randomness protocol.}
With overwhelming probability, there are at most $O(\log n)$ epochs.
The verifier chooses $k$ random samples for each circuit, and checks if the verifier passes $\LXEB_{1+\delta,k}$ for $99\%$ of the given circuits.
Applying an Entropy Accumulation Theorem (EAT, explained in \sec{intro-eat}), we prove the following statement.

\begin{theorem}[Entropy accumulation, no side information, informal]\label{thm:llqsv-eat-informal}
  For $\beta\in[0,1]$, if $\LLHA_{\beta n}$ holds,
  then for integer $k=\Omega(n^2)$ and $m=\Omega(\log n)$,
  there exists an $m$-round entropy accumulation protocol taking $k\cdot m$ samples such that conditioned on the event $\Omega$ of not aborting,
  \begin{align}
      H_{\min}(Z|C)_{\rho|\Omega}
      \geq 
      n\left(\left(0.99-\frac{0.01}{\delta}\right)\frac{\beta}{2} m-O(\sqrt{m})\right)
  \end{align}
   for every device solving $\LXEB_{1+\delta,k}$,
  where $\rho$ is the output state, $Z$ is the responses received from the device, and $C$ is the circuit in each round. %
\end{theorem}

\subsection{Entangled Adversary and Ideal Measurements}

In the previous section, we assumed that an attacker, Eve, trying to predict the quantum computer's outputs had no preshared entanglement with the quantum computer. \ Now we relax that assumption.

To build intuition, we start with the special case where the quantum computer performs an ideal measurement---i.e., it ``just'' applies $C$ to $n$ qubits, followed by a measurement in the standard basis. \ The ``only'' problem is that the qubits might not start in the desired initial state $\ket{0^n}$, but rather in some arbitrary state entangled with Eve's qubits.

We define the following idealized score, called $b$-$\XHOG(\U)$. %
\begin{problem}[$b$-XHOG$(\D)$ {\cite{Kre21}}]
  For a distribution $\D$ over quantum circuits on $n$ qubits, an algorithm $\A$ given access to $C\sim\D$ is said to solve $b$-$\XHOG$ if it outputs a sample $z$ such that 
  \begin{align}
    \Exp_{C\sim\D}\left[\Exp_{z\sim\A^C}\left[p_C(z)\right]\right] \geq \frac{b}{N}.
  \end{align}
\end{problem}
For an algorithm $\A$ which is given access to $C$ and outputs $z$, we define the ``XHOG score'' of $\A$ to be the value $\Exp_C[\Exp_{z\sim\A^C}[p_C(z)]]$.
This score was first considered by Kretschmer for showing a Tsirelson bound for random circuit sampling in the oracle model \cite{Kre21}. \ 
Recall that with the CHSH game, a violation of the classical bound $3/4$ implies certified randomness. \ 
Interestingly, our result may be interpreted as certified randomness from a violation of the classical XHOG score. \ 

We will proceed with our analysis with $b$-XHOG first, and show a von Neumann entropy lower bound $\Omega(\delta n)$ when the device solves $(1+\delta)$-XHOG.
First, the problem itself is linear in the device's output distribution: that is, for two devices $\A$ with score $s_\A$ and $\B$ with score $s_\B$, a third device that runs $\A$ with probability $p$ and $\B$ with probability $(1-p)$ has score $p\cdot s_\A + (1-p)\cdot s_\B$.
The linearity condition coincides with the score calculation from violation of Bell's inequality.

More concretely, recall that to establish certified randomness from a violation of Bell's inequality, two devices $\A$ and $\B$ are asked to play, say, the CHSH game: the verifier sends two questions $x,y$ to the devices and collecting the responses $a,b$.
The verifier sets the score to $1$ if $x\wedge y = a\oplus b$ and 0 otherwise, and the expectation of the score is defined as 
\begin{align}
  \omega = \Exp_{x,y\sim\bit}\left[ \Exp_{a,b\sim\A^x\otimes\B^y(\rho)}[\Id[x\wedge y=a\oplus b]] \right].
\end{align}
It is well-known that the best achievable expectation is $\omega=\cos^2(\pi/8)$. \ 
For certified randomness, it was further shown that when the expectation $\omega\geq \cos^{2}(\pi/8+\epsilon)$, the output of $\A$ has von Neumann entropy lower-bounded by $1-h(\sin 4\epsilon)\approx 1-O(\epsilon)$, where $h(x):=-x\log x-(1-x)\log(1-x)$ is the binary entropy function \cite{ADFRV18}. \ 
Like the $\XHOG$ score, here the score $\omega$ can be exactly computed only by taking infinitely many samples from the same devices. \ 
With a finite number of samples, we can only approximate the score. \ 

Proving a conditional min-entropy lower bound from sample statistics, in an $m$-round sequential process, amounts to the problem of entropy accumulation. \ 
An Entropy Accumulation Theorem (EAT) for certified randomness is usually stated as follows: \
In an $m$-round sequential process, the verifier randomly selects $O(\log m)$ rounds to get an approximation of the score. \
If the approximation is sufficiently close to $\cos^2(\pi/8)$, the number of extractable random bits is at least $\Omega(m)$ times the von Neumann entropy lower bound established in a single-round analysis. \ 

Without loss of generality, let the entanglement shared between the device and Eve be a pure state $\ket{\psi}$.  \ 
For every state $\rho_{ZE}$ classical on $Z$, we show that the conditional von Neumann entropy $H(Z|CE)_\rho\geq H(Z|C)_\rho-\chi(Z:CE)_\rho$, where $\chi$ is the Holevo quantity. \ 
To see why they they must use weak entanglement, we can write $\ket{\psi} := \sum_x \alpha_x \ket{\psi_x}\ket{\phi_x}$ for orthonormal bases $\{\ket{\psi_x}\}$ and $\{\ket{\phi_x}\}$ in the Schmidt decomposition. \ 
We show that the device solving $b$-XHOG for $b\geq 1+\delta$, the amplitude $\alpha_x$ must concentrate at a single compoment, say $x^*$, such that $|\alpha_{x^*}|^2\geq\delta$.
By the observation that the Holevo quantity equals the entanglement entropy, we establish an upper bound $O((1-\delta) n)$ on the Holevo quantity. \  

In this simplified setting, we have already seen that if the device has a large XHOG score, then they must use weak entanglement to pass the verification. \ 
However, the entire analysis relies on the assumption that the device must perform the ideal measurement. \ 

\subsection{A Fully General Device}

Next, we consider the setting in which the adversary may share arbitrary entanglement with Eve. \ 
We give an unconditional proof for certified randomness in the random oracle model. \ 
Instead of setting a binary-valued score for each question-answer pair as in \sec{intro-no-side-info}, for question $C$ and response $z$, the score is set to $p_C(z)$. \ 
Then, in a single round analysis, we first show that if the device passes $(1+\delta)$-XHOG score, then the von Neumann entropy of the joint state is at least $\Omega(n)$. \ %

\begin{theorem}[Single-round analysis, informal]\label{thm:main-informal}
  Every device $\A$ that on input the first system of a bipartite state $\rho_{DE}$ and given oracle access to a Haar random $C$, makes $T\leq 2^{n/7}$ queries and solves $(1+\delta)$-$\XHOG$ must output a state $2^{-\Omega(n)}$-close to a state $\psi_{ZE}$ classical on system $Z$ such that 
  \begin{align}\label{eq:main}
    H(Z|CE)_{\psi} \geq 0.99\delta n - O(\log T).
  \end{align}
  Furthermore, there is a single-query device that solves $(2-2^{-n})$-$\XHOG$.
\end{theorem}

To prove \thm{main-informal}, the key observation is that one can approximate, to diamond distance $2^{-\Omega(n)}$, any device $\A$ making $T$ queries to a Haar random $C$ by another device $\F$ which does not make any queries, but is given $k$ samples $z_1,\ldots,z_k\sim P_C$ for $k=T^2\cdot 2^{O(n)}$. \ %
For each $\A$, we call the associated device $\F$ the simplified device. \ 
With probability $1-O(k^2/N)$, these samples does not contain any collision. \ 
In this event (no collision occurring), $\F$ solves $(1+\delta)$-XHOG implies that $\F$ outputs $z\in S=\{z_1,\ldots,z_k\}$ with probability at least $\delta-o(1)$. \ 
Intuitively, this robustly certifies that $\A$'s output must be $\epsilon$-close to a strategy where the output is prepared by sampling from $C$ for $k$ times and choosing one of the samples. %

From this point of view, a simplified device solving $(1+\delta)$-XHOG is equivalent to winning the following game with probability at least $\delta-o(1)$: \ 
Given $k$ independent samples $S$ from $P_C$, outputs a string $z\in S$. \ 
Though the game looks quite trivial, it yields a sharp lower bound of the von Neumann min-entropy. \ 
By the no-communication theorem, Eve, even if she learns $P_C$, has no information about the samples given to $\F$, but Eve can potentially control the output distribution when the device sees a particular set of samples. \ 
Since with high probability over $C$, $P_C$ has min-entropy $n-O(\log n)$, we show that averaging over any distribution supported on these samples, the resulting distribution has von Neumann entropy at least $0.99\delta n - O(\log T)$. \

Our lower bound in \thm{main-informal} is close to optimal. \ 
Consider a device which samples from $P_C$ with probability $\delta$ and outputs a uniform string obtained by performing a standard basis measurement on EPR pairs shared with Eve with probability $1-\delta$. \ 
In the former event, the device solves $b$-XHOG for $b\approx 2$, whereas in the latter, the output is a uniformly random string, which solves $1$-XHOG. \  
Thus, by linearity, the device solves $(1+\delta)$-$\XHOG$. \ 
The output joint classical-quantum state from the device is a probabilistic mixture of the two states. \ 
Moreover, with overwhelming probability over choices of $C$, the Shannon entropy of $P_C$ is $n-O(\log n)$. \ 
Then by the concavity of von Neumann entropy, the output has conditional von Neumann entropy $\delta n-o(n)$. \ 

To accumulate the entropy, we give a sequential process which is very similar to the one introduced in \sec{intro-no-side-info}: \ 
The verifier samples $t=O(\log n)$ different circuits, and asks for at least $k$ samples for each circuit. \ 
Upon receiving the samples, the verifier chooses $k$ random samples for each circuit, and checks if the device passes $\LXEB_{1+\delta,k}$ for a constant $\delta$ for $99\%$ of the circuits. \ 
We show that the accumulation process certifies $\Omega(\delta mn)$ bits. \ 

\begin{theorem}[Entropy accumulation from a general device, informal]\label{thm:eat-general-informal}
  For integer $k=\Omega(n^2)$, $m=\Omega(k\log n)$, there exists an entropy accumulation protocol taking $m$ samples such that conditioned on the event $\Omega$ of non-aborting,
  \begin{align}
      H_{\min}(Z|CE)_{\rho|\Omega}
      \geq n\left(0.99\delta m - O(\sqrt{m})\right)
  \end{align}
   for devices solving $\LXEB_{1+\delta,k}$,
  where $\rho$ is the output state, $Z$ is the samples from the device, $C$ is the circuits and $E$ is the information held by Eve.
\end{theorem}
More concretely, taking $\delta=0.1$, this bound is $n\cdot\left(0.099 m - O(\sqrt m)\right)$ by taking $m$ samples from the device.
We note that this bound is seemingly weaker than the bound in \thm{llqsv-eat-informal}, but the number of samples is $m$ (instead of $km$ as in \thm{llqsv-eat-informal}). 
The minimal sample complexities in the protocols are no different---both are $\Omega(n^2\log n)$---for a perfect device to pass the verification with overwhelming probability. \ 
For technical reasons, in the latter protocol, $k$ samples for each verification are randomly chosen (from all samples sent by the device corresponding to the same challenge circuit) and received sequentially. \ 
In contrast, in the former protocol, in each round the device is asked to send $k$ samples, and the verifier checks one round for each circuit. \ 

We also note that \thm{llqsv-eat-informal} and \thm{eat-general-informal} are  incomparable results. \ 
In particular, the security analysis for \thm{eat-general-informal} heavily relies on the model in which the device is given access to the circuit and the distribution (the Haar measure) over circuits. \
In contrast, the security analysis based on LLHA may still hold when the device is given access to a description of circuits sampled from other distributions. \ 
We leave it as an open question whether there exists a hardness assumption under which linear cross-entropy benchmarking certifies min-entropy against an entangling adversary in the plain model. \

\subsection{Entropy Accumulation}\label{sec:intro-eat}

The proof of \thm{eat-general-informal} is based on the entropy accumulation theorem (EAT) by Dupuis, Fawzi and Renner \cite{DFR20}, with modifications explained as follows. \ 
Let $f$ be an affine function, called the min-tradeoff function, such that in a single-round analysis, one can show that $H(Z|E)_\rho\geq f(q)$ for distribution $q=(p,1-p)$ and any state $\rho$ whose acceptance probability is $p$. \ 
In an $m$-round sequential process, the verifier checks $\gamma m$ rounds (called test rounds) by computing the decision bits from the samples, and computes an approximate distribution $\tilde q=(\tilde p, 1-\tilde p)$.
The min-entropy round across the $m$ rounds is then $m\cdot f(\tilde q)-O(\sqrt{m})$. \ 
Thus an EAT reduces a multi-round analysis to a lower bound on the single-round von Neumann entropy. \ 

Since we adopt the $b$-XHOG score for a bound of the von Neumann entropy in a single round analysis, the score obtained from the test rounds is no longer computed from binary random variables. 
Thus we define a new min-tradeoff function $f'$ which maps the \emph{score} to a lower bound of the von Neumann entropy. \ 
Then we show that if an approximation of the score, defined as the average of $p_{C_i}(z_i)$ is more than $s$, then the accumulated entropy is at least $m\cdot f'(s)-O(\sqrt{m})$. \ 

The entropy accumulation procedure allows for \emph{spot checking}, that is, in the $m$-round process, instead of computing $p_{C_i}(z_i)$ for every round $i\in[m]$, the verifier only computes $p_{C_i}(z_i)$ for a subset of indices $i$ of size $O(n^2\log n)$. \ 
In more details, the verifier changes the circuits for $O(\log n)$ times, and in each epoch the verifier computes the average of $k=O(n^2)$ samples. \ 
The number of test rounds is set for the device that takes i.i.d.\ samples from $p_C$ on each circuit $C$ to pass the verification with overwhelming probability. \ 
By Hoeffding's inequality, a device that samples from $p_C$ outputs $k$ samples whose average score is concentrated above $2-O(1)$ for a typical $C$ with overwhelming probability. \
If the verifier passes $\LXEB_{b,k}$ for the epoches of a sufficiently large fraction $\Omega(1)$, the average is above $1+\Omega(1)$ with overwhelming probability, and by the entropy accumulation theorem, the conditional min-entropy is $\Omega(nm)$. \

\subsection{Pseudorandomness and Statistical Zero Knowledge}

The protocols for certified randomness rely on perfect randomness for generating the challenge circuit. \ 
However, by a counting argument, the challenge space is doubly exponentially large, and it requires exponentially many random bits to compute a truly random circuit. \ 
To produce a net gain in randomness, we must rely an efficiently computable function which uses polynomially many random bits and generates pseudorandomness with security level sufficient for our purpose. \ 

However, the standard notion of pseudorandom functions (PRFs) against quantum polynomial-time adversaries does not seem to be sufficient, since it only guarantees the output of the device is \emph{pseudorandom}! \ 
Thus for certified randomness, we require a stronger pseudorandom function, when a truly random function is replaced with which, the output remains \emph{statistically indistinguishable} from the uniform distribution. \ 

To provide such a security guarantee, we construct a pseudorandom function indistinguishable from a truly random function for any $\QSZK$ protocols. \ 
To see why such a security level is sufficient, we recall facts about the class $\QSZK$ which consists of promise problems that admit a quantum statistical zero-knowledge protocol. \ 
A $\QSZK$ protocol is one that consists of a proof system, i.e., a quantum polynomial-time verifier and an unbounded prover, and an efficient quantum simulator which simulates the interaction of the proof system without access to a witness. \ 

Watrous showed that $\QSZK$ has a natural complete problem called the \emph{quantum state distinguishability problem} (QSD) \cite{Wat02}. \ 
In this problem, the instance is a tuple of two efficiently computable quantum circuits $Q_0,Q_1$. \ 
For $\alpha\in(0,1]$, the verifier is challenged to determine the trace distance $\|\rho_0-\rho_1\|_{\tr}$ is at least $\alpha$, or at most $\alpha^2$, where $\rho_b$ is a marginal state obtained by computing $Q_b$ for $b\in\bit$. \ 
It is known for this class, there is an amplification procedure, and therefore the gap can be made exponentially close to 1 \cite{Wat02}. \ 
More recently, Menda and Watrous showed that relative to a random oracle, $\mathsf{UP}\not\subset\QSZK$ \cite{MW18}. \  
Ben-David and Kothari defined a query measure on statistical zero-knowledge proof, and showed that the positive-weighted adversary method can only prove suboptimal lower bounds for certain problems \cite{BK19}. \ 

For certified randomness, we define the $\QSZK$-distinguishability between two distributions over functions, and a similar definition can be extended to distributions over unitaries.
\begin{definition}[$\QSZK$-distinguishability, informal]
  Two distributions $\D_0,\D_1$ over functions are said to be $\mathrm{QSZK}$-distinguishable if there exist a pair of algorithms $\A,\B$ such that the averaged trace distance between $\A^F$ and $\B^F$'s output states has non-neglgigible difference between $F\sim\D_0$ and $F\sim\D_1$.
  The distributions are said to be $\QSZK$-indistinguishable if no such algorithms exist.
\end{definition}
A $\QSZK$-secure pseudorandom function is defined as one $\QSZK$-indistinguishable from a random function. 
We propose an assumption, called pseudorandom function assumption (PRFA), that there exists a $\QSZK$-secure pseudorandom function. \ 
We justify the assumptions are valid by giving a construction for algorithms given oracle access to the function. \ 
\begin{theorem}[Pseudorandom functions, informal]
  There exists a $\QSZK$-secure pseudorandom function with key length $O(n)$ relative to a random oracle. \ 
\end{theorem}

Similarly, we say a pseudorandom unitary is $\QSZK$-secure if it is $\QSZK$-indistinguishable from a random unitary. \ 
We propose a similar assumption, called pseudorandom unitary assumption (PRUA), that there exists a $\QSZK$-secure pseudorandom unitary, and prove the existence relative to an oracle with key length $O(n)$. \ 

Under these assumptions, when replacing a random circuit with a pseudorandom one, the output remains statistical indistinguishable from a uniform distribution, conditioned on Eve's side information. \ 
To see why, recall that De, Portmann, Vidick, and Renner \cite{DPVR12} showed that Trivesan's randomness extractor \cite{Tre01} is quantum-proof. That is, the output from the randomness extractor together with Eve's side information is a quantum state $\rho$ statistically indistinguishable from $\sigma\otimes \rho_E$, where $\sigma$ is a maximally mixed state and $\rho_E$ is the marginal state held by Eve. \
If the device given a pseudorandom circuit outputs a quantum state that changes the distance by a non-negligible amount from $\sigma\otimes\rho_E$, then such a device implies a $\QSZK$ protocol that distinguishes a pseudorandom circuit from a random one. \ 

To see there is a net gain in randomness, the protocol samples $O(\log n)$ pseudorandom circuits, each of which takes $O(n)$ random bits for the keys of the pseudorandom function, and finally it produces $\Omega(mn)$ random bits. \ 
For $m=\poly(n)$, we have a polynomial expansion. \ 

While we do not know whether a weaker assumption can work for certified randomness, the security level seems necessary against an entangling adversary. \ 
Indeed, if there is no quantum side information, then all we need is to use a pseudorandom circuit against adversaries solving the \emph{statistical difference from uniform} problem. \ 
In the purely classical setting, the problem is known to be complete for $\NISZK$, a subclass of $\SZK$ consisting of problems that admits a non-interactive statistical zero-knowledge protocol \cite{GSV99}. \ 
However, in the presence of quantum side information, an unbounded Eve can prepare any $\rho_E$, and security against $\QSZK$ seems necessary. \

\section{Preliminaries}\label{sec:preliminaries}

As we said, a circuit $C$ acts on $n$ qubits, and $N=2^n$. \ 
The binary entropy function $h:[0,1]\to\mathbb R$ is defined as $h(x):=-x\log x-(1-x)\log(1-x)$ and $h(0)=h(1)=0$. \ 
For matrix $A$, we denote by $\|A\|_p:=\tr(|A|^p)^{1/p}$ the Schatten $p$-norm of $A$. \ 
Furthermore, we denote $\|A\|_{\op}$ the operator norm and $\|A\|_F:=\left(\sum_{i=1}^n\sum_{j=1}^n |A_{ij}|^2\right)^{1/2}$ the Frobenius norm of an $n\times n$ matrix $A$. \ 
The trace distance between two quantum states $\rho,\sigma$ is defined as $\|\rho-\sigma\|_{\tr}:=\frac{1}{2}\|\rho-\sigma\|_1$. \ 
The fidelity $F(\rho,\sigma)$ between two quantum states $\rho,\sigma$ is defined as $\|\rho^{1/2}\sigma^{1/2}\|_1$.
The function $F$ is symmetric, i.e., $F(\rho,\sigma)=F(\sigma,\rho)$. \ 
If $\rho$ is a pure state $\ket{\psi}$, then the fidelity $F(\rho,\sigma)=\bra{\psi}\sigma\ket{\psi}$. \ 

A quantum processes is completely positive trace preserving (CPTP) map. \ 
For integer $n,m>0$, let $\Phi:M_n(\mathbb C)\to M_{m}(\mathbb C)$ denote a linear transformation from complex-valued $n\times n$ matrices $M_n(\mathbb C)$ to $M_m(\mathbb C)$. \
The diamond norm of $\Phi$ is defined as $\|\Phi\|_\diamond:=\max_{X\in M_n(\mathbb C):\|X\|_1\leq 1}\|\Phi\otimes\Id_n(X)\|$. \ 
The diamond distance between two quantum processes $\Phi$ and $\Psi$ is defined as $\|\Phi-\Psi\|_\diamond$. \ 

For Hilbert space $A$, we denote $\mathbb S(A)$ the set of normalized quantum state in $A$. \ 
For Hilbert spaces $A,B$, we denote $\CPTP(A,B)$ the set of CPTP maps from linear operators on $A$ to linear operators on $B$. \
The set of unitary operators on $A$ be $\mathbb U(A)$. \ 

\subsection{Complexity Classes}

We assume the readers to have the familiarity with standard classical complexity classes and the class $\BQP$. \ 
Here we briefly introduce quantum complexity classes related to our work. \ 
The class $\QMA$ (which stands for ``Quantum Merlin-Arthur'') is a quantum analogue of $\MA$, defined as follows. 
\begin{definition}[$\QMA$]
  The complexity class $\QMA$ consists of languages $L$ for which there exists a quantum polynomial-time algorithm $V$ such that the following conditions hold.
  \begin{itemize}
  \item If $x\in L$, there exists a quantum state $\rho$ on $\poly(|x|)$ qubits such that $\Pr[V(x,\rho)\text{ accepts}]\geq 2/3$.
  \item If $x\notin L$, for every quantum state $\rho$ on $\poly(|x|)$ qubits, $\Pr[V(x,\rho)\text{ accepts}]\leq 1/3$.
  \end{itemize}
\end{definition}

The quantum polynomial-time algorithm $V$ is also called the $\QMA$ verifier, and the state $\rho$ can be thought of as a state sent by a $\QMA$ ``prover'' of unbounded power. \ 
The class $\QCMA$ is also a quantum analogue of $\MA$ and can be defined similarly as $\QMA$, except that the witness state $\rho$ is restricted to a classical string. \ 
To be specific about the running time or the query complexity of a $\QMA$ verifier, we denote $\QMA(T)$ to be a $\QMA$ verifier running in time $T$ or making $T$ queries in a relativized world. \ 
The same notation also applies to other classes. \ 

The class $\QCAM$ (which stands for ``Quantum-Classical Arthur Merlin'') is a quantum analogue of $\AM$, and can be defined in terms of a two-message protocol. \
In the first message, the quantum polynomial-time verifier sends a random string $r$ of size polynomial in the size of the instance to the prover Merlin. \ 
Merlin then sends a response $w$ (also called the witness) which is also of size polynomial in the size of the instance. \ 
Note that the witness $w$ can arbitrarily depend on the instance and the random string $r$. \ 
A language is in $\QCAM$ if there exists an Arthur which outputs the correct answer for every instance with probability at least $2/3$. \ 
\begin{definition}[$\QCAM$]
  The complexity class $\QCAM$ consists of languages $L$ for which there exists a quantum polynomial-time algorithm $V$ (also called Arthur) and a polynomial $p$ such that the following conditions hold. 
  \begin{itemize}
  \item If $x\in L$, then there exist a polynomial $q$ and a classical string $w\in\bit^{q(|x|)}$ such that 
    \begin{align}
      \Pr_{r\in\bit^{p(|x|)}}[V(x,r,w)\text{ accepts}] \geq 2/3. 
    \end{align}

  \item If $x\notin L$, then for every polynomial $q$ and every string $w\in\bit^{q(|x|)}$, 
    \begin{align}
      \Pr_{r\in\bit^{p(|x|)}}[V(x,r,w)\text{ accepts}] \leq 1/3. 
    \end{align}
  \end{itemize}
\end{definition}

A $k$-message interactive proof system consists of two algorithms, the computationally unbounded prover $P$ and a polynomial-time verifier $V$, and there are $k$ message exchanges between $V$ and $P$. \ 
We will also say such a protocol is a verifier of length $k$. \ 
The class $\QIP$ is an interactive proof system in which the verifier runs in quantum polynomial time, and each message can be a quantum state. \ 
Such a protocol is also called a quantum interactive proof system. \

\begin{definition}[{$\QIP[k]$}]
  The complexity class $\QIP[k]$ consists of languages $L$ for which there exists a quantum polynomial-time algorithm $V$ (also called the verifier) of length $k$ such that the following conditions hold.
  \begin{itemize}
  \item If $x\in L$, then there exists a prover $P$ which makes the verifier accepts with probability $2/3$. 
  \item If $x\notin L$, then for every prover $P$, the verifier $V$ accepts with probability at most $1/3$. 
  \end{itemize}
\end{definition}

The complexity class $\QSZK$ (which stands for ``Quantum Statistical Zero-Knowledge) consists of languages that admit a quantum statistical zero-knowledge protocol, defined by Watrous \cite{Wat02}.

\begin{definition}[{$\QSZK$}]\label{dfn:qszk}
  A quantum statistical zero-knowledge proof system for a language $L$ consists of an unbounded $P$ (called the honest prover), a quantum polynomial-time verifier $V$ such that the following holds.
  \begin{itemize}

    \item Completeness and soundness: $(V,P)$ is a proof system for $L$.

      \item Zero-knowledge: there exist a negligible function $\eta$ and a set of preparable states $\{\sigma_{x,i}\}$ such that if $x\in L_1$, $\|\sigma_{x,i}-\mathrm{view}_{V,P}(x,i)\|_{\tr}\leq \eta(|x|)$. \ 
      Here $\mathrm{view}_{V,P}(x,i)$ is the verifier's view after the $i$-th round, i.e., the mixed state of the verifier and the message qubits after $i$-th message have been sent during an execution of the proof system on input $x$.
  \end{itemize}
\end{definition}
Watrous showed that $\QSZK$ has a natural complete problem called Quantum State Distinguishability ($\QSD$) \cite{Wat02}, defined as follows. \ 
\begin{definition}[$(\alpha,\beta)$-$\QSD$]
  The promise problem $(\alpha,\beta)$-$\QSD=(\QSD_1,\QSD_0)$ for $0\leq\alpha<\beta^2\leq 1$ consists of a pair of quantum circuits $(Q_0,Q_1)$ such that
  \begin{itemize}
  \item if $(Q_0,Q_1)\in\QSD_1$, then $\|\rho_0-\rho_1\|_{\tr}\geq \beta$, and
  \item if $(Q_0,Q_1)\in\QSD_0$, then  $\|\rho_0-\rho_1\|_{\tr}\leq \alpha$,
  \end{itemize}
  where $\rho_b$ is obtained by applying $Q_b$ on the zero state followed by partial tracing on some of the qubits.
\end{definition}

Watrous showed that the complete and soundness parameters $(\alpha,\beta)$ can be amplified to $(2^{-n},1-2^{-n})$ by giving a transformation from $Q_b$ to a circuit that has size polynomial in $n$ and $|Q_b|$ for $b\in\bit$ \cite{Wat02}. \
Since the parameters do not matter for this problem, we will denote $\QSD$ the same problem with a constant gap. \

The problem has a very simple $\QSZK$ protocol: \
The verifier tosses a random coin $b$, and sends $\rho_b$ to the prover. The prover performs the optimal measurement that saturates the trace distance and outputs a bit $b'$. \
The verifier accepts if $b=b'$. \
The completeness follows from the fact that the states's trace distance is negligibly close to one, and this implies there exists a measurement that perfectly distinguishes the states. \
For the soundness, since the states are negligibly close in trace distance, the prover does not succeed with non-negligibly advantage over random guessing. \
To show the protocol is zero-knowledge, the quantum simulator applies the verifier's quantum operation first. \
After receiving the reponse $b'$ from the prover, it sets $b'=b$. \

The $\QSZK$-completeness of $\QSD$ relativizes. \ 
In particular, Menda and Watrous \cite{MW18} showed that for oracle $A$, a problem $L_A$ in $\QSZK^A$ if there exists a reduction from $L_A$ to $\QSD^A$.
\begin{theorem}[{\cite[Theorem~1]{MW18}}]\label{thm:mw18}
  For alphabet $\Sigma,\Gamma$, let $L\subseteq \Gamma^*$ be a language and $A\subseteq \Sigma^*$ be an oracle.
  The language $L_A$ is contained in $\QSZK^A$ if and only if there exists a polynomial-time uniform family of pairs of relativized quantum circuits $(Q_0^A,Q_1^A)$ with the following properties:
  \begin{itemize}
  \item If $x\in L_A$, then $(Q_0^A,Q_1^A)\in\QSD_1^A$.
  \item If $x\notin L_A$, then $(Q_0^A,Q_1^A)\in\QSD_0^A$.
  \end{itemize}
\end{theorem}

Ben-David and Kothari \cite{BK19} studied independently the so-called $\QSZK$ complexity of function $f$, denoted $\mathrm{QSZK}(f)$, which is defined as the minimum number $k$ made by a pair of query algorithms $\A,\B$ given oracle access to $x$ such that for every $x$ such that (i) if $f(x)=1$, then $\|\A^x-\B^x\|_{\tr}\geq 2/3$, and (ii) if $f(x)=0$, then $\|\A^x-\B^x\|_{\tr}\leq 1/3$.

\subsection{The Polynomial Method}

The quantum polynomial method by Beals, Buhrman, Cleve, Mosca, and de Wolf is a standard technique for proving quantum lower bound of query problems \cite{BBC+01}. \ 
Specifically, we will use the degree lower bound by Markov and the strong direct product theorem by Sherstov \cite{She11}.
\begin{lemma}\label{lem:markov}
  Let $p:\mathbb R\to\mathbb R$ be a polynomial. \ 
  For real numbers $a,b$, if
  \begin{align}
    \max_{a\leq x\leq b}|p(x)-p(y)| \leq H,
  \end{align}
  then 
  \begin{align}
    |p'(x)| \leq \frac{H}{b-a} \deg(p)^2,
  \end{align}
  where $p'$ is the first derivative of $p$ and $\deg(p)$ is the degree of $p$.
\end{lemma}
\begin{theorem}[{Strong direct product theorem \cite[Theorem~1.5]{She11}}]\label{thm:sherstov}
  Fix functions $f_1,\ldots,f_k:\{-1,+1\}^m\to\{-1,+1\}$. \ 
  Then solving $(f_1,\ldots,f_k)$ with worst-case probability $2^{-\Omega(k)}$ requires 
\begin{align}
\Omega\left(\min_{S\subseteq[k]:|S|=0.99k}\left(\sum_{i\in S}\deg_{1/5}(f_i)\right)\right),
\end{align}
where $\deg_\epsilon(f)$ stands for the least degree of a real polynomial that approximates $f$ within $\epsilon$ pointwise. \ 
\end{theorem}

\subsection{Quantum Information Theory}

The amount of extractable randomness is the conditional min-entropy $H_{\min}(Z|E)$ which describe the amount of randomness system $Z$ has conditioned on Eve's information $E$.
The smooth min-entropy is formally defined as 
\begin{align}
  H_{\min}(Z|E)_{\rho} = \sup_{\sigma_E}\left\{ -\inf_{\lambda}\left\{\lambda: \rho_{ZE} \leq 2^{-\lambda}\Id_Z\otimes\sigma_E\right\} \right\}.
\end{align}
In the case where both $Z$ and $E$ are classical, they are formally defined as random variables, and both $\rho_{ZE}$ and $\sigma_E$ are diagonal in the same basis.

We will also consider a smooth version of conditional min-entropy, which relaxes the above notion by considering an $\epsilon$-close pair of random variables, in total variation distance:
\begin{align}
  H_{\min}^\epsilon(Z|E)_\rho = \sup_{\tilde\rho:\|\tilde\rho_{ZE}-\rho_{ZE}\|_{\tr}\leq\epsilon} H_{\min}(Z|E)_\rho.
\end{align}

Another useful quantity in the family of quantum R\'{e}nyi entropies is the von Neumann entropy.
The von Neumann entropy of a quantum state $\rho_A$ is defined as $H(A)_\rho:=-\tr(\rho\log\rho)$.
The conditional von Nemann entropy of a bipartite state $\rho_{AB}$ is defined as $H(A|B)_\rho=H(AB)_\rho-H(B)_\rho$.

For classical random variable $X$ distributed according to a distribution $P$, we denote $H(X)$ or $H(P)$ the von Neumann entropy of $X$, and $H_{\min}(X)$ or $H_{\min}(P)$ the min-entropy of $X$.

The following inequalities for von Neumann entropy will be useful.
  \begin{lemma}\label{lem:triangle}
    For finite-dimensional Hilbert spaces $A,B,C$ and a tripartite state $\rho_{ABC}$ on $A\otimes B\otimes C$, it holds that
    \begin{align}
      H(A|B)_\rho + H(B|C)_\rho \geq H(A|C)_\rho.
    \end{align}
  \end{lemma}
  \begin{proof}
    By strong sub-additivity, $H(A|BC)_\rho\leq H(A|B)_\rho$ and thus
    \begin{align}\nonumber
      H(B|C)_\rho 
      &= H(AB|C)_\rho - H(A|BC)_\rho \\
      &\geq H(A|C)_\rho - H(A|B)_\rho.
    \end{align}
  \end{proof}

  \begin{lemma}\label{lem:convexity}
    For finite-dimensional Hilbert space $A$, let $\rho,\sigma$ be two normalized quantum states on $A$. 
    For $\lambda\in[0,1]$,
    \begin{align}
      H(A)_{(1-\lambda)\rho + \lambda\sigma}
      \leq (1-\lambda) H(A)_\rho + \lambda H(A)_\sigma + h(\lambda)
    \end{align}
    where $h$ is the binary entropy function.
  \end{lemma}
  \begin{proof}
    Define the quantum state 
    \begin{align}
      \psi_{BA} = (1-\lambda)\proj{0}_B\otimes \rho_A + \lambda\proj{1}_B \otimes \sigma_A.
    \end{align}
    By definition, the von Neumann entropy of $\psi$ is 
    \begin{align}\nonumber
      H(BA)_\psi 
      &= -\tr((1-\lambda)\rho\log(1-\lambda)\rho))-\tr(\lambda\sigma\log \lambda\sigma)) \\\nonumber
      &= -(1-\lambda)\log(1-\lambda) - (1-\lambda)\tr(\rho\log\rho) - \lambda\log\lambda - \lambda\tr(\sigma\log\sigma) \\
      &= H(A)_\rho + H(A)_\sigma + h(\lambda). 
    \end{align}
    Moreover, since $\tr_B(\psi)=(1-\lambda)\rho + \lambda\sigma$, 
    $H(A)_\psi = H(A)_{(1-\lambda)\rho + \lambda\sigma}$.
    Since conditional von Neumann entropy is always non-negative, $H(BA)_\psi\geq H(A)_\psi$, and we conclude the proof.
  \end{proof}

  \begin{lemma}[{\cite[Theorem~11.9]{NC02}}]\label{lem:vn-data}
    For Hilbert space $A$, let $\rho$ be a quantum state on $A$ and $\{P_i\}$ be a complete set of projective measurements on $A$.
    Then the entropy of the post-measurement state $\sigma=\sum_i P_i\rho P_i$ is at least as great as the original entropy, i.e., $H(A)_\sigma \geq H(A)_\rho$. 
  \end{lemma}

By \lem{convexity} and concavity of von Neumann entropy, we can lower bound the mutual information of a probabilistic mixture of states by a convex combination of the mutual information of each component.
Recall that by definition, for Hilbert space $A,B$ and bipartite state $\rho_{AB}$, $I(A:B)_\rho=H(A)+H(B)-H(AB)$. 

\begin{lemma}\label{lem:mutual}\nonumber
  For finite-dimensional Hilbert space $A$ and $B$, let $\rho_{AB},\sigma_{AB}$ be a bipatite state.
  For $\lambda\in[0,1]$,
  \begin{align}
    I(A:B)_{(1-\lambda)\rho+\lambda\sigma} 
    \geq (1-\lambda)I(A:B)_\rho + \lambda I(A:B)_\sigma - h(\lambda).
  \end{align}
\end{lemma}
\begin{proof}
  By convexity of von Neumann entropy, for $X\in\{A,B\}$,
  \begin{align}
    H(X)_{(1-\lambda)\rho+\lambda\sigma} \geq (1-\lambda) H(X)_\rho + \lambda H(X)_\sigma.
  \end{align}
  Then by \lem{convexity},
  \begin{align}\nonumber
    I(A:B)_{(1-\lambda)\rho+\lambda\sigma}
    &= H(A)_{(1-\lambda)\rho+\lambda\sigma}
    +H(B)_{(1-\lambda)\rho+\lambda\sigma}
    -H(AB)_{(1-\lambda)\rho+\lambda\sigma} \\
    &\geq (1-\lambda)I(A:B)_\rho + \lambda I(A:B)_\sigma - h(\lambda).
  \end{align}
\end{proof}

\subsection{Haar Random Unitaries}\label{sec:haar-intro}

We will rely one the following facts.
For a Haar random unitary $C\in\mathbb C^{N\times N}$, for every $i\in\bit^n$, it holds that the distribution $P$ of density $P_z=p_C(z):=|\bra{z}C\ket{0}|^2$ is distributed according to the Dirichlet distribution $\Dir(1^N)$ on the probability simplex \cite{DMR13,Kre21}.

\subsubsection{The Dirichlet Distributions}

Let $\Dir(\alpha)$ denote the Dirichlet distribution for concentration hyperparameter $\alpha$.
The moments for $(X_1,\ldots,X_N)\sim\Dir(\alpha)$ for $\alpha=(\alpha_1,\ldots,\alpha_N)$ is well-studied. 
First, the mean of each random variable $\Exp[X_i]=\frac{\alpha_i}{\alpha_0}$, where $\alpha_0=\sum_{i=1}^N\alpha_i$.
Moreover,
\begin{align}
  \Exp\left[\prod_{i=1}^N X_i^{\beta_i}\right] = \frac{B(\alpha+\beta)}{B(\alpha)},
\end{align}
where $B(\alpha):=\frac{1}{\Gamma(\alpha_0)}\prod_{i=1}^N\Gamma(\alpha_i)$.

The Dirichlet distribution is the conjugate prior distribution of the categorical distribution and the multinomial distribution \cite{FKG10}.
If the prior distribution is sampled according to the Dirichlet distribution, the posterior is also a Dirichlet distribution with a different hyperparameter.
In particular, let the data points be $z_1,\ldots,z_k\sim P$ where $P\sim\Dir(\alpha)$.
Then the posterior distribution $P|(z_1,\ldots,z_k)\sim\Dir(\alpha+m)$, where $m$ is the vector of the number of occurrences for the data points in each category.

We can sample a probability distribution from the Dirichlet distribution $\Dir(\alpha)$ using the following process: First sample $Q_i\sim\Gamma(\alpha_i,1)$ independently for $i\in[N]$, where $\Gamma(\alpha_i,1)$ is the Gamma distribution with parameters $\alpha_i,1$.
Then compute $\bar Q=\sum_{i=1}^N Q_i$ and set $P_i=Q_i/\bar Q$ for each $i\in[N]$. 
Thus it would be useful to briefly introduce facts about the Gamma function.
Recall that the Gamma distribution $\Gamma(\alpha,\beta)$ has pdf
\begin{align}
  f(x; \alpha,\beta) = \frac{\beta^\alpha e^{-\beta x} x^{\alpha-1}}{\Gamma(\alpha)}.
\end{align}
In particular, for $\alpha=\beta=1$, the pdf is $f(x; 1, 1)= e^{-x}$.
This implies that the CDF of $\Gamma(1,1)$ is 
\begin{align}
  F(x; 1,1) = 1-e^{-x}.
\end{align}

Using these facts, we prove a few lemmas which will be useful later.
First, the maximum value of $P\sim\Dir(1^N)$ is $O(n)/N$ in expectation.   
\begin{lemma}\label{lem:haar-min-avg}
    Let $P\sim\Dir(1^N)$. Then
    \begin{align}\label{eq:haar-min-avg-1}
      \Exp_{P\sim\Dir(1^N)}\left[\max_z P_z\right] \leq \frac{2\ln N+7}{N}.
    \end{align}
    and 
    \begin{align}\label{eq:haar-min-avg-2}
      \Exp_{P\sim\Dir(1^N)}\left[H_{\min}(P)\right] \geq n-\log n-O(1).
    \end{align}
  \end{lemma}
  \begin{proof}
    Let $F(x)=1-e^{-x}$ be the CDF of $\Gamma(1,1)$ and $Q_1,\ldots,Q_{N}\sim\Gamma(1,1)$.
    Also let $Q=(Q_1,\ldots,Q_N)$.
    The CDF of $\max_z Q_z$ is
    \begin{align}
      G(x) = \Pr[\forall z, Q_z\leq x] = (1-e^{-x})^N.
    \end{align}
    Thus the expectation 
    \begin{align}\nonumber
      \Exp_{Q\sim\Gamma(1,1)^N}\left[\max_z Q_z \right] 
      &= \int_0^{\infty} \d x (1-(1-e^{-x})^N) \\\nonumber
      &= \int_0^1 \d y \frac{1-y^N}{1-y} \\\nonumber
      &= \int_0^1\d y (1+y+y^2+\ldots y^{N-1}) \\\nonumber
      &= 1 + \frac{1}{2} + \ldots +\frac{1}{N} \\
      &\leq \ln N + 1.
    \end{align}
    As shown in the proof of \lem{haar-min}, $\bar Q=\sum_z Q_z$ is concentrated:
    \begin{align}\label{eq:q-concentration}
      \Pr_Q[|\bar Q-N|\leq N/2] \geq 1-\frac{4}{N}.
    \end{align}
    Let $\Omega:=\{Q:\bar Q\geq N/2\}$.
    From \eq{q-concentration}, $\Pr_{Q\sim\Gamma(1,1)^N}[Q\in\Omega]\geq 1-4/N$. 
    This means that 
    \begin{align}\nonumber
      \Exp_{Q\sim\Gamma(1,1)^N}\left[\frac{\max_z Q_z}{\bar Q}\right]
      &\leq 
      \Exp_Q\left[\left.\frac{\max_z Q_z}{\bar Q}\right| Q\in\Omega\right]
      + \Pr[Q\notin\Omega] \\\nonumber
      &\leq
        \frac{2}{N}\Exp_Q\left[\left.\max_z Q_z\right| Q\in\Omega\right] + \frac{4}{N} \\\nonumber
      &\leq
        \frac{2}{N}\frac{1}{\Pr[Q\in\Omega]}\Exp_Q\left[\max_z Q_z\right] + \frac{4}{N} \\
      &\leq \frac{2\ln N+7}{N}.
    \end{align}
    This show that \eq{haar-min-avg-1} is correct.
    For \eq{haar-min-avg-2}, by Jensen's inequality,
    \begin{align}
      \Exp_{P\sim\Dir(1^N)}\left[ H_{\min}(P) \right]
      \geq -\log \Exp_{P\sim\Dir(1^N)}\left[\max_z P_z\right] \geq n-\log n - O(1).
    \end{align}
  \end{proof}

In fact, the maximum of $P\sim\Dir(1^N)$ is concentrated around $O(n)/N$.
  \begin{lemma}\label{lem:haar-min}
    It holds that
    \begin{align}
      \Pr_{P\sim\Dir(1^N)}\left[ \max_z P_z \leq \frac{4\ln N}{N} \right] \geq 1-\frac{6}{N}.
    \end{align}
  \end{lemma}
  \begin{proof}

    Let $F(x)=(1-e^{-x})$ be the CDF of $\Gamma(1,1)$.
    Thus instead consider 
    \begin{align}\label{eq:69}\nonumber
      \Pr_{Q\sim\Gamma(1,1)^N}\left[\max_z Q_z \leq 2\ln N\right]
      &= 
      \Pr_{Q\sim\Gamma(1,1)^N}\left[\forall z, Q_z \leq 2\ln N\right] \\\nonumber
      &= F(2(\ln 2) n)^N \\\nonumber
      &= (1-N^{-2})^N \\
      &\geq 1-2/N.
    \end{align}
    Also $\bar Q=\sum_{z} Q_z$ is concentrated since 
    the mean $\Exp[\bar Q]=N$ and the variance $\sigma^2=\Exp[\bar Q^2]-\Exp[\bar Q]^2=N(\Exp[\bar Q_0^2]-\Exp[\bar Q_0]^2)=N$.
    By Chebyshev inequality,
    \begin{align}\label{eq:70}
      \Pr_Q[|\bar Q-N|\leq N/2]\geq 1-\frac{4}{N}.
    \end{align}
    Combining \eq{69} and \eq{70},
    \begin{align}
      \Pr_Q\left[\frac{\max_z Q_z}{\bar Q}\leq \frac{4\ln N}{N}\right] \geq 1-\frac{6}{N}.
    \end{align}
  \end{proof}

\subsubsection{The Performance of a Perfect Device\label{sec:perfect}}

In this paper, we define a perfect device to be one that is given $C$, outputs a sample $z\sim p_C$, where $p_C$ is the distribution defined by $p_C(z):=|\bra{z}C\ket{0^n}|^2$.
The first result is well-known and has been proven using different mathematical tools: a perfect device solves $b$-XHOG for $b\approx 2$.
Here we prove the result using properties of the Dirichlet distribution.

\begin{lemma}\label{lem:completeness}
  Sampling from $z\sim P_C$ for Haar random $C$, it holds that
  \begin{align}
    \Exp_{C\sim\Haar(N),z\sim P_C}[p_C(z)] = \frac{2}{N+1}.
  \end{align}
\end{lemma}
\begin{proof}
  The algorithm $\A$ outputs the given sample $z$. %
  For $z\in\bit^n$, let $P_z=|\bra{z}C\ket{0}|^2$ be a random variable for Haar random $C$, and $P=(P_0,\ldots,P_{N-1})$ be a vector of random variables distributed on the probability simplex.
  Since $P\sim\Dir(1^N)$, $P|z\sim\Dir(m+1^N)$, where $m=(0,\ldots,0,m_z=1,0,\ldots,0)$.
  The score can be calculated as follows:
  \begin{align}
    \Exp_{P\sim\Dir(1^N)}[P_z|m] = \Exp_{P\sim\Dir(m+1^N)}[P_z]= \frac{m_z+1}{N+1}.
  \end{align}
  Thus for $z$ such that $m_z=1$, and the expectation is $\frac{2}{N+1}$.
\end{proof}

In fact, the score for $C\sim\Haar(N)$ is concentrated around $\frac{2}{N+1}$.
Thus sampling from $C\sim\Haar(N)$, with overwhelming probability, a perfect device answers with a sample that has score close to $\frac{2}{N+1}$.
\begin{lemma}[Concentration of collision probability]\label{lem:collision-concentration}
  Let $S_C:= \Exp_{z\sim p_C}[p_C(z)]$.
  Then $S_C$ is concentrated in the sense that 
  \begin{align}
    \Pr_C\left[ \left|S_C - \frac{2}{N+1}\right| \leq \frac{\epsilon}{N+1}\right] \geq 1-O\left(\frac{1}{\epsilon^2 N}\right).
  \end{align}
\end{lemma}
\begin{proof}
We consider the moments of the random variable $S_C=\sum_z p_C(z)^2$ for Haar random $C$.
By \lem{completeness}, $\mu = \Exp_C[S_C] = \frac{2}{N+1}$.
\begin{align}\nonumber\label{eq:pc-collision-exp}
  \Exp_C[S_C^2] 
  &= \Exp_{C}\left[\sum_{z,z'}p_C(z)^2p_C(z')^2\right] \\\nonumber
  &= N\Exp_C[p_C(0)^4] + N(N-1)\Exp_C[p_C(0)^2p_C(1)^2]  \\\nonumber
  &= \frac{N!4!}{(N+3)!} + N(N-1)\frac{(N-1)!2!2!}{(N+3)!} \\
  &= \frac{20+4N}{(N+1)(N+2)(N+3)}
\end{align}
Therefore, the variance is 
\begin{align}\nonumber
  \Exp_C[S_C^2]-\Exp_{C}[S_C]^2
  &= \frac{(4N+20)(N+1) - 4(N+2)(N+3)}{(N+1)^2(N+2)(N+3)} \\
  &= \frac{4(N-1)}{(N+1)^2(N+2)(N+3)} = O\left(\frac{1}{N^3}\right).
\end{align}
By Chebyshev inequality, %
\begin{align}\nonumber\label{eq:collision-concentration}
  \Pr_{C}\left[\frac{2+\epsilon}{N+1}\geq S_C\geq \frac{2-\epsilon}{N+1}\right] 
  &= \Pr_C\left[|S_C-\mu|\leq \frac{\epsilon}{N+1}\right] \\
  &\geq 1 - \frac{4(N-1)}{\epsilon^2(N+2)(N+3)} = 1- O\left(\frac{1}{\epsilon^2N}\right).
\end{align}
This implies that with probability $1-O(\frac{1}{\epsilon^2N})$ over $C$, when $S_C\geq (2-\epsilon)/N$.
\end{proof}

This lemma shows that the collision probability of $P_C$ is sharply concentrated around its mean. 
That is, for almost every $C$, if the device samples from $P_C$, the expectation of $p_C(z)$ is very close to $\frac{2}{N+1}$.
This implies that a device sampling from $P_C$ solves $\LXEB_{b,k}$ with constant probability for $k=\Omega(n^2)$ and $b\geq 1.98$.
\begin{lemma}\label{lem:lxeb-perfect}
  For integer $k$, with probability $1-O(1/N)$ over $C\sim\Haar(N)$, sampling $z_1,\ldots,z_k\sim P_C$,
  \begin{align}
      \Pr_{z_1,\ldots,z_k\sim P_C}\left[\frac{1}{k}\sum_{i=1}^k p_C(z_i) \geq \frac{1.98}{N+1}\right] \geq 1-2^{-\Omega(k/n^2)}.
  \end{align}
\end{lemma}
\begin{proof}
  By \lem{collision-concentration} with $\epsilon=0.01$, 
  \begin{align}
      \Pr_{C\sim\Haar(N)}\left[S_C \geq \frac{1.99}{N+1}\right] \geq 1-O(1/N).
  \end{align}
  Furthermore, by \lem{haar-min}, with probability $1-O(1/N)$, $\max_z p_C(z)\leq 4n/N$.
  For $C$ satisfying $S_C\geq 1.99/(N+1)$ and $\max_z p_C(z)\leq 4n/N$, by Hoeffding's inequality,
  \begin{align}\nonumber
  \Pr_{z_1,\ldots,z_k\sim P_C}\left[
  \frac{1}{k}\sum_{i=1}^k p_C(z_i) \leq \frac{1.98}{N+1}
  \right]
  &\leq 
    \Pr_{z_1,\ldots,z_k\sim P_C}\left[
  \left|\frac{1}{k}\sum_{i=1}^k p_C(z_i) -S_C\right| \geq \frac{0.01}{N+1}
  \right] \\
    &\leq 
    2e^{-0.01 k/(16n^2)} = 2e^{-k/(1600n^2)}.
  \end{align}
\end{proof}

Therefore, it suffices to take $O(n^2)$ independent samples from $O(\log n)$ different circuits.
\begin{lemma}\label{lem:lxeb-product}%
  For integer $k$, circuit $C$ and a tuple $d$ of $k$ strings in $\bit^n$, let $V_k(C,d)$ be defined as
  \begin{align}
     V_k(C,d) := \Id\left[\frac{1}{k}\sum_{i=1}^k p_C(z_i) \geq \frac{1.98}{N+1}\right]
  \end{align}
  for $d=(z_1,\ldots,z_k)$.
  Then for $k=O(n^2)$ and $m=\Omega(\log n)$, sampling a perfect device $\A$ $mk$ times yields $d_1,\ldots,d_m$ satisfying 
  \begin{align}
      \Pr_{C_1,\ldots,C_m\sim\Haar(N),\{d_i\sim\A(C_i)\}_{i=1}^m}\left[\sum_{i=1}^m V_k(C_i,d_i) \geq 0.99m\right] \geq 1-1/n.
  \end{align}
\end{lemma}
\begin{proof}
  By \lem{lxeb-perfect}, for $k=cn^2$ for sufficiently large constant $c$, 
  \begin{align}
      \mu := \Pr_{C\sim\Haar(N),d\sim\A(C)}[V(C,d)=1] \geq 0.995.
  \end{align}
  Then by Hoeffding's inequality,
  \begin{align}\nonumber
      &\Pr_{C_1,\ldots,C_m\sim\Haar(N),\{d_i\sim\A(C_i)\}_{i=1}^m}\left[\sum_{i=1}^m V_k(C_i,d_i) \leq  0.99m\right] \\\nonumber
      &\qquad\leq 
      \Pr_{C_1,\ldots,C_m\sim\Haar(N),\{d_i\sim\A(C_i)\}_{i=1}^m}\left[\left|\frac{1}{m}\sum_{i=1}^m V_k(C_i,d_i) -\mu\right|\geq  0.005\right] \\\nonumber
      &\qquad \leq 2^{-2\cdot 0.005^2m} \\
      &\qquad \leq 1/n
  \end{align}
  for $m=c'\log n$ for sufficiently large $c'$.
\end{proof}

This implies that sampling from a perfect device $mk$ times to yield an approximation of the XHOG score close to $1.94/(N+1)$.
The bound can be improved to $2-c$ for an arbitrarily small constant $c>0$.
\begin{corollary}
  For $k=O(n^2)$ and $m=\Omega(\log n)$, sampling a perfect device $\A$ for $k$ times on $m$ independent circuits sampled from the Haar measure,
  with probability $1-1/n$,
  \begin{align}
      \frac{1}{mk}\sum_{i=1}^m\sum_{j=1}^k p_{C_i}(z_{ij}) \geq \frac{1.94}{N+1},
  \end{align}
  where $z_{ij}$ denotes the $j$-th sample from $C_i$.
\end{corollary}
\begin{proof}
  The corollary holds from \lem{lxeb-product} and the implication that 
  \begin{align}
      \frac{1}{m}\sum_{i=1}^m V_k(C_i,d_i) \geq 0.98
      \implies 
      \frac{1}{mk}\sum_{i=1}^m\sum_{j=1}^k p_{C_i}(z_{ij}) \geq \frac{1.94}{N+1},
  \end{align}
  for $d_i=(z_{i1},\ldots,z_{ik})$.
\end{proof}

\section{An Entropy Accumulation Theorem}\label{sec:eat}

In this section, we modify the entropy accumulation theorem (EAT) from the one given by Dupuis, Fawzi, and Renner~\cite{DFR20}. \
Our proof follows closely from that of \cite{DFR20} except with the following minor changes: \
In \cite{DFR20}, the min-tradeoff function $f:\mathbb P\to\mathbb R$ is defined over the set of probability distributions $\mathbb P$. \
Here our min-tradeoff function $f:\mathbb R^{\geq 0}\to\mathbb R$, where the input corresponds to the score of the device. \
To see the modification does lead to an useful tool that reduces the analysis of a multi-round entropy accumulation process to single-round analysis of von Neumann entropy, we formally prove the theorem from scratch. \

Recall that for von Neumann entropy, by definition, the chain rule holds:
\begin{align}
  H(A_1A_2|B)_\rho = H(A_1|B)_\rho + H(A_2|A_1B)_\rho
\end{align}
for every tripartite state $\rho_{A_1A_2B}$.
However, the same equality does not hold for other R\'{e}nyi entropies. \
Instead, Dupuis, Fawzi, and Renner \cite{DFR20} show the following statement. \ 
\begin{theorem}[{\cite[Theorem~3.2]{DFR20}}]\label{thm:chain-1}
  Let $\rho_{A_1A_2 B}$ be a density operator and $\alpha\in(0,\infty)$.
  Then
  \begin{align}\label{eq:eat}
    H_\alpha(A_1A_2|B)_\rho = H_\alpha(A_1|B)_\rho + H_\alpha(A_2|A_1 B)_\nu,
  \end{align}
  where
  \begin{align}
    \nu_{A_1A_2B} = \nu_{A_1B}^{1/2} \rho_{A_2|A_1B} \nu_{A_1B}^{1/2} 
    \qquad\text{ with}\qquad
    \nu_{A_1B} = \frac{\left(\rho_{A_1B}^{1/2} \rho_B^{\frac{1-\alpha}{\alpha}} \rho_{A_1B}^{1/2}\right)^\alpha}{\tr\left(\rho_{A_1B}^{1/2} \rho_B^{\frac{1-\alpha}{\alpha}} \rho_{A_1B}^{1/2}\right)^\alpha}.
  \end{align}
\end{theorem}
Note that $\nu_{A_1A_2B}$ is normalized and $\nu_{A_1B}$ is the marginal of $\nu_{A_1A_2B}$ obtained by tracing out the system $A_2$.
Though there is a state $\nu$ such that the equality in \eq{eat} holds, for our purpose, we only care about some family of states for which an inequality can be derived.
In particular, the following inequality will be useful.
\begin{theorem}[{\cite[Theorem~3.3]{DFR20}}]\label{thm:chain-2}
  Let $\rho_{A_1B_1A_2B_2}$ be a density operator and $\alpha\in(0,\infty)$ such that the Markov chain condition holds.
  Then 
  \begin{align}
    \inf_\nu H_\alpha(A_2|B_2A_1B_1)_\nu
    \leq H_\alpha(A_1A_2|B_1B_2)_\rho - H_\alpha(A_1|B_1)_\rho
    \leq
    \inf_\nu H_\alpha(A_2|B_2A_1B_1)_\nu
  \end{align}
\end{theorem}
\thm{chain-2} can be obtained by showing that the Markov chain condition implies that $H_{\alpha}(A_1|B_1B_2)_\rho=H_\alpha(A_1|B_1)_\rho$ and by \thm{chain-1}.
We can further represent \thm{chain-2} in terms of quantum channels.
\begin{theorem}[{\cite[Corollary~3.5]{DFR20}}]\label{thm:chain-3}
  Let $\rho_{RA_1B_1}$ be a density operator on $R\otimes A_1\otimes B_1$, $\M\in\CPTP(R,A_2B_2)$ and $\alpha\in(0,\infty)$.
  If $\M(\rho_{RA_1B_1})$ satisfies the Markov chain condition $A_1\leftrightarrow B_1\leftrightarrow B_2$, then
  \begin{align}
    \inf_\omega H_\alpha(A_2|B_2A_1B_1)_{\M(\omega)}
    \leq
    H_\alpha(A_1A_2|B_1B_2)_{\M(\rho)} - H_\alpha(A_1|B_1)_\rho
    \leq
    \sup_\omega H_\alpha(A_2|B_2A_1B_1)_{\M(\omega)},
  \end{align}
  where the supremum and infimum range over density operator on $R\otimes A_1\otimes A_2$.
  Moreover, if $\rho_{RA_1B_1}$ is pure, then it suffices to optimize over pure states $\omega_{RA_1B_1}$.
\end{theorem}
\thm{chain-3} can be obtained by applying \thm{chain-2} and presenting a state $\omega$ that saturates the inequalities.
In our case, the Markov chain condition trivially holds since the system $B_i$ is empty.

Let $\tilde E_{i-1}$ be a system isomorphic to $R_{i-1}E$ and $\mathbb{P}$ denote the set of distributions.
Adding a system $\tilde R_{i-1}$ is meant to purify the state on $R_{i-1}E$ so the state $\omega_{R_{i-1}E\tilde R_{i-1}}$ is pure and its marginal is any state on $R_{i-1}$ input to $\M_i$.

We modify the definition of min-tradeoff function from \cite{DFR20}, formally stated as follows.

\begin{definition}\label{dfn:min-tradeoff}
  A real-valued function on $f:\mathbb R^{\geq 0}\to\mathbb R$ is called a min-tradeoff function for $\mathcal M_i$ if it satisfies 
  \begin{align}
    f(s) \leq \inf_{\nu\in\Sigma_i(G,s)} H(A_i | E \tilde E_{i-1})_\nu 
  \end{align}
  where $\Sigma_i(G,s):=\{\rho_{A_i R_i E\tilde E_{i-1}}=(\M_i\otimes\Id_{E\tilde R_{i-1}})(\omega_{R_{i-1}E\tilde R_{i-1}}): \tr(G \rho_{A_i})=s\}$, i.e., the set of states whose marginal on $A_i$ has score $s$ evaluated using a diagonal, positive semi-definite matrix $G$.
\end{definition}
While the domain of $f$ is the set of non-negative real numbers, 
we restrict our attention to the properties of $f$ in the interval $[0,2]$ for normalization purposes.
Let $\|\nabla f\|_\infty$ be the infinity norm of $\nabla f$ restricted to $[0,2]$, and $g_{\max},g_{\min}$ be the maximum and the minimum of $f$ over $[0,2]$, i.e.,
\begin{align}
  g_{\max} := \max_{s\in[0,2]} f(s), \qquad g_{\min} = \min_{s\in[0,2]} f(s)
\end{align}
In particular, it holds that $\frac{1}{2}|g_{\max}-g_{\min}|\leq \|\nabla f\|_{\infty}$, where $\|\nabla f\|_{\infty}$ is the infinite norm of $f$, by setting the domain in $[0,2]$.
Note that the restriction is without loss of generality: given a game with score in $[0,a]$ with some diagonal $G'$, we can set $G=2G'/a$.

To determine the infimum, it suffices to consider only pure states $\omega_{R_{i-1}\tilde R_{i-1}}$ since by strong subadditivity, adding a purification system cannot increase $H(A_i|E\tilde R_{i-1})$.
For an event $\Omega\subseteq \A^m$, we denote 
\begin{align}
\Pr_\rho[\Omega]:=\sum_{(a_1,\ldots,a_m)\in\Omega}\tr(\rho_{E,a_1,\ldots,a_m})
\end{align}
for classical-quantum state $\rho$ classical on $A^m$ the trace of the state projected onto the subspace spanned by $\ket{w_1,\ldots,w_m}$ for $(w_1,\ldots,w_m)\in\Omega$.
For any state $\rho_{A_1^m B_1^m E}$ classical on $A_1^mB_1^m$, we denote the conditional state
\begin{align}
  \rho_{A^m E|\Omega}
  := \frac{1}{\Pr_\rho[\Omega]} \sum_{(a_1,\ldots,a_m)\in\Omega} \proj{a_1,\ldots,a_m} \otimes \rho_{E,a_1,\ldots,a_m}.
\end{align}
Clearly, the state is normalized.

The following theorem states it suffices to bound $H_{\alpha}^\uparrow(Z^m|E)_{\M_m\circ\ldots\circ\M_1(\rho)|\Omega}$ from below.

\begin{proposition}[{\cite[Lemma~B.10]{DFR20}}]\label{prp:min-alpha}
  For any density operator $\rho$, and non-negative operator $\sigma$ any $\alpha\in(1,2]$, and any $\epsilon\in(0,1)$, 
  \begin{align}
    H_{\min}^\epsilon(Z^m|E)_{\M_m\circ\ldots\circ\M_1(\rho)|\Omega} \geq H_\alpha^\uparrow(Z^m|E)_{\M_m\circ\ldots\circ\M_1(\rho)|\Omega} - \frac{\log(2/\epsilon^2)}{\alpha-1}.
  \end{align}
\end{proposition}
To account for the event $\Omega$, we follow the ideas from \cite{DFR20}: \
First, we introduce systems $D_1,\ldots,D_m$ and normalized states $\{\tau(z):z\in\mathcal Z\}$ such that 
\begin{align}\label{eq:hd}
  H(D_i)_{\tau(z)}=\bar g - f(G_z), 
\end{align}
where $G_z:=\bra{z}G\ket{z}$ and CPTP maps $\D_i:\CPTP(Z_i,Z_iD_i)$, where 
\begin{align}
  \D_i(\proj{z}) := \proj{z} \otimes \tau(z)_{D_i}.
\end{align}
Also we define $\bar\M_i:=\D_i\circ\M_i$. \
This allows us to apply the following proposition due to Metger, Fawzi, Sutter, and Renner~\cite{MFSR22}.
\begin{proposition}[{\cite[Lemma~4.5]{MFSR22}}]\label{prp:alpha-up}
  For $\alpha>1$ and normalized state $\rho$,
  \begin{align}\label{eq:alpha-up}
    H_\alpha^\uparrow(Z^m|E)_{\M_m\circ\ldots\circ\M_1(\rho)|\Omega}
    \geq H_\alpha^\uparrow(Z^mD^m|E)_{\M_m\circ\ldots\circ\M_1(\rho)|\Omega}
    - \max_{z\in\Omega} H_\alpha(D^m)_{\bar\M_m\circ\ldots\circ\bar\M_1(\rho)_{z}}.
  \end{align}
\end{proposition}
Let $g_{\min}$ and $g_{\max}$ denote the maximum and minimum of the range of $f$. \ 
Furthermore, let $\bar g:=\frac{1}{2}(g_{\min}+g_{\max})$. \ 
We say a distribution $q$ is induced by samples $z_1,\ldots,z_m$ if 
\begin{align}
  q(z) = \frac{1}{m}|\{i:z_i=z\}|.
\end{align}
For affine $f$ and distribution $q$ induced by the samples $z_1,\ldots,z_m$, 
\begin{align}\label{eq:nabla}
  \left|\bar g - f\left(\Exp_{z\sim q}G_z\right)\right|
  \leq \frac{1}{2}|g_{\max}-g_{\min}|
  \leq \|\nabla f\|_\infty.
\end{align}
We next bound each term on the rhs of \eq{alpha-up} individually.
\begin{proposition}\label{prp:alpha-dm}
  For $\alpha>1$ and every $z=(z_1,\ldots,z_m)\in\mathcal Z^m$,
  \begin{align}
    H_\alpha(D^m)_{\bar\M_m\circ\ldots\circ\bar\M_1(\rho)_{D,z}} \leq m\|\nabla f\|_\infty,
  \end{align}
\end{proposition}
\begin{proof}
Since each $\tau(z_i)_{D_i}$ is determined from $z_i$, the marginal state is product, i.e.,
\begin{align}\label{eq:tau-product}
  \bar\M_m\circ\ldots\circ\bar\M_1(\rho)_{D_1\ldots D_m} = \tau(z_1)_{D_1}\otimes \ldots \otimes \tau(z_m)_{D_m}.
\end{align}
Thus 
\begin{align}\nonumber
  \frac{1}{m} H_\alpha(D^m)_{\bar\M_m\circ\ldots\circ\bar\M_1(\rho)} 
  &= \frac{1}{m}\sum_{i=1}^m H_\alpha(D_i)_{\tau(z_i)} \\\nonumber
  &\leq \frac{1}{m}\sum_{i=1}^m H(D_i)_{\tau(z_i)} \\\nonumber
  &= \frac{1}{m}\sum_{i=1}^m \bar g - f(G_{z_i}) \\\nonumber
  &= \bar g - f\left(\Exp_{z\sim q} G_z\right) \\
  &\leq \|\nabla f\|_\infty.
\end{align}
The first equality holds from \eq{tau-product}.
The first inequality holds by the monotonicity of R\'{e}nyi entropyies in $\alpha$.
The second equality holds from \eq{hd}.
The third equality holds because $f$ is an affine function.
The last inequality holds from \eq{nabla}. 
\end{proof}
The first term of \eq{alpha-up} can be further simplified using the chain rule \thm{chain-3}.
Note that the Markov chain trivially holds since in this paper, we only consider the case where each system $B_i$ in \thm{chain-3} is empty for $i\in[m]$.
\begin{proposition}\label{prp:chain}
  For $\alpha\in(1,\infty)$, %
  \begin{align}\nonumber\label{eq:chain-2}
    &H_\alpha^\uparrow(Z^mD^m|E)_{\sigma|\Omega} \\
    &\qquad\geq 
  \sum_{i=1}^m \inf_{\omega_{i-1}\in\mathrm S(R_{i-1}E\tilde E_{i-1})} H_{\alpha}(Z_iD_i|E\tilde E_{i-1})_{\bar\M_i(\omega_{i-1})} - \frac{\alpha}{\alpha-1}\log\left(\frac{1}{\Pr_{\sigma}[\Omega]}\right),
  \end{align}
  where $\sigma:=\M_m\circ\cdots\circ\M_1(\rho)$ and $\tilde E_{i-1}$ is a system isomorphic to $R_{i-1}E$ (see \dfn{min-tradeoff}).
\end{proposition}
\begin{proof}
  By direct calculation,
  \begin{align}\nonumber
    H_\alpha^\uparrow(Z^mD^m|E)_{\M_m\circ\ldots\circ\M_1(\rho)|\Omega}
    &\geq H_\alpha(Z^mD^m|E)_{\M_m\circ\ldots\circ\M_1(\rho)|\Omega} \\
    &\geq H_\alpha(Z^mD^m|E)_{\M_m\circ\ldots\circ\M_1(\rho)} - \frac{\alpha}{\alpha-1}\log\left(\frac{1}{\Pr_{\sigma}[\Omega]}\right).
  \end{align}
  The first inequality holds $H_\alpha^\uparrow(A|B)_\rho\geq H_\alpha(A|B)_\rho$ for any finite-dimensional Hilbert spaces $A,B$ and bipartite state $\rho$. \ 
  The second inequality holds from \cite[Lemma~B.5]{DFR20}.

  Next, we apply the chain rule \thm{chain-3} on the first term and have
  \begin{align}
    H_\alpha(Z^mD^m|E)_{\M_m\circ\ldots\circ\M_1(\rho)}
    \geq 
  \sum_{i=1}^m \inf_{\omega_{i-1}\in\mathrm S(R_{i-1}E\tilde E_{i-1})} H_{\alpha}(Z_iD_i|E\tilde E_{i-1})_{\bar\M_i(\omega_{i-1})},
  \end{align}
  where we introduce an purifying system $\tilde E_i$ for each $i$ such that $\omega_{i-1}\in\mathrm S(R_{i-1}E\tilde E_{i-1})$ is a pure state.
\end{proof}

Now we bound each term in \eq{chain-2}.
In fact, the bound will be depend on the dimension of $D_i$. \
Recall that each $D_i$ is introduce for analysis purposes, and we can choose a sufficient large dimension $d_{D_i}=d_D:=\lceil 2^{\|\nabla f\|_\infty}\rceil$ such that \eq{hd} can be satisfied for every $i\in[m]$. \
This implies the following proposition.
\begin{proposition}\label{prp:alpha-zdee}
  For $\alpha\in(1,1+\log(2d_Zd_D+1))$, 
  \begin{align}
    H_\alpha(Z_iD_i|E\tilde E_i)_{\bar\M_i(\omega_{i-1})}
    &\geq
    H(Z_iD_i|E\tilde E_i)_{\bar\M_i(\omega_{i-1})} - (\alpha-1)(\|\nabla f\|_\infty+\log(2d_Z+1))^2.
  \end{align}
\end{proposition}
\begin{proof}
  By \cite[Lemma~B.9]{DFR20}, 
  \begin{align}\nonumber
    H_\alpha(Z_iD_i|E\tilde E_i)_{\bar\M_i(\omega_{i-1})}
    &\geq 
    H(Z_iD_i|E\tilde E_i)_{\bar\M_i(\omega_{i-1})} - (\alpha-1)\log^2(2d_Zd_D+1) \\
    &\geq
    H(Z_iD_i|E\tilde E_i)_{\bar\M_i(\omega_{i-1})} - (\alpha-1)(\|\nabla f\|_\infty+\log(2d_Z+1))^2.
  \end{align}
  
\end{proof}

Combining \prp{chain} and \prp{alpha-zdee}, we have the following corollary.
\begin{corollary}\label{cor:zde}
  For $\alpha\in(1,1+2/V)$,
  \begin{align}\nonumber
    &H_\alpha^\uparrow(Z^mD^m|E)_{\M_m\circ\ldots\circ\M_1(\rho)|\Omega} \\
    &\qquad\geq 
  \sum_{i=1}^m \inf_{\omega_{i-1}\in\mathrm S(R_{i-1}E\tilde E_{i-1})} H(Z_iD_i|E\tilde E_i)_{\bar\M_i(\omega_{i-1})} - m\left(\frac{\alpha-1}{4}\right)V^2 - \frac{\alpha}{\alpha-1}\log\left(\frac{1}{\Pr_{\sigma}[\Omega]}\right),
  \end{align}
  where $V=2(\log(2d_Z+1)+\|\nabla f\|_\infty)$.
\end{corollary}
\begin{proof}
  By direct calculation, 
  \begin{align}\nonumber
    &H_\alpha^\uparrow(Z^mD^m|E)_{\sigma|\Omega} \\\nonumber
    &\qquad\geq 
  \sum_{i=1}^m \inf_{\omega_{i-1}\in\mathrm S(R_{i-1}E\tilde E_{i-1})} H_{\alpha}(Z_iD_i|E\tilde E_{i-1})_{\bar\M_i(\omega_{i-1})} - \frac{\alpha}{\alpha-1}\log\left(\frac{1}{\Pr_{\sigma}[\Omega]}\right) \\
    &\qquad\geq
  \sum_{i=1}^m \inf_{\omega_{i-1}\in\mathrm S(R_{i-1}E\tilde E_{i-1})} H(Z_iD_i|E\tilde E_i)_{\bar\M_i(\omega_{i-1})} - m\left(\frac{\alpha-1}{4}\right)V^2 - \frac{\alpha}{\alpha-1}\log\left(\frac{1}{\Pr_{\sigma}[\Omega]}\right).
  \end{align}
  The first inequality holds from \prp{chain}. \
  The second inequality holds from \prp{alpha-zdee}.
\end{proof}

Next we simplify the first term in \cor{zde}.
\begin{proposition}\label{prp:zdee}
  \begin{align}
    H(Z_iD_i|E\tilde E_i)_{\bar\M_i(\omega_{i-1})}
    \geq \bar g.
  \end{align}
\end{proposition}
\begin{proof}
  By direct calculation,
  \begin{align}\nonumber
    H(Z_iD_i|E\tilde E_{i-1})_{\bar\M_i(\omega_{i-1})}
    &=
    H(Z_i|E\tilde E_{i-1})_{\bar\M_i(\omega_{i-1})}
    + H(D_i|Z_iE\tilde E_{i-1})_{\bar\M_i(\omega_{i-1})} \\\nonumber
    &=
    H(Z_i|E\tilde E_{i-1})_{\bar\M_i(\omega_{i-1})}
    + H(D_i|Z_i)_{\bar\M_i(\omega_{i-1})} \\\nonumber
    &=
    H(Z_i|E\tilde E_{i-1})_{\bar\M_i(\omega_{i-1})}
    + \Exp_{z\sim r} H(D_i)_{\tau(z)} \\\nonumber
    &=
    H(Z_i|E\tilde E_{i-1})_{\bar\M_i(\omega_{i-1})}
    + \bar g - \Exp_{z\sim r} f(G_z) \\\nonumber
    &=
    H(Z_i|E\tilde E_{i-1})_{\bar\M_i(\omega_{i-1})}
    + \bar g - f\left(\Exp_{z\sim r} G_z\right) \\
    &\geq
    \bar g,
  \end{align}
  where $r$ is the distribution obtained by taking the marginal of $\bar\M_i(\omega_{i-1})$ on $Z_i$. \
  The first equality holds by the chain rule of von Neumann entropy. \
  The second holds because the marginal state $\bar\M_i(\omega_{i-1})_{D_iZ_i}=\sum_z\tau_i(z)_{D_i}\otimes r(z)\proj{z}_{Z_i}$. \
  The third holds by the definition of conditional von Neumann entropy and that of $r$. \
  The fourth holds from \eq{hd}. \
  The fifth equality holds since $f$ is affine.
  The inequality holds since by definition, $H(Z_i|E\tilde E_{i-1})_{\bar\M_i(\omega_{i-1})}\geq f\left(\Exp_{z\sim r}G_z\right)$, where $r=\bar\M_i(\omega_{i-1})_{Z_i}$.
\end{proof}
Combining \prp{alpha-up}, \prp{alpha-dm}, \cor{zde}, and \prp{zdee},
we have the following corollary.
\begin{corollary}\label{cor:zme}
  For $\alpha\in(1,1+2/V)$,
  \begin{align}
    H_\alpha^\uparrow(Z^m|E)_{\M_m\circ\ldots\circ\M_1(\rho)|\Omega}
    &\geq 
      m f(s)- m\left(\frac{\alpha-1}{4}\right)V^2 - \frac{\alpha}{\alpha-1}\log\left(\frac{1}{\Pr_\sigma[\Omega]}\right),
  \end{align}
  where $\sigma:=\M_m\circ\cdots\circ\M_1(\rho)$, and $s:=\frac{1}{m}\sum_{i=1}^m G_{z_i}$ is the score evaluated by taking the average of $G_{z_i}$ from each round $i\in[m]$.
\end{corollary}
\begin{proof}
  By direct calculation,
  \begin{align}\nonumber
    H_\alpha^\uparrow(Z^m|E)_{\M_m\circ\ldots\circ\M_1(\rho)|\Omega}
    &\geq H_\alpha^\uparrow(Z^mD^m|E)_{\M_m\circ\ldots\circ\M_1(\rho)|\Omega}
    - m\|\nabla f\|_\infty \\\nonumber
    &\geq 
    m\bar g - m\left(\frac{\alpha-1}{4}\right)V^2 - \frac{\alpha}{\alpha-1}\log\left(\frac{1}{\Pr_\sigma[\Omega]}\right) - m\|\nabla f\|_\infty \\
    &\geq 
      m f(s)- m\left(\frac{\alpha-1}{4}\right)V^2 - \frac{\alpha}{\alpha-1}\log\left(\frac{1}{\Pr_\sigma[\Omega]}\right).
  \end{align}
  The first inequality holds from \prp{alpha-up} and \prp{alpha-dm}. \
  The second holds from \cor{zde}. \
  The third holds because $f(s)\leq \bar g-\|\nabla f\|_{\infty}$.
\end{proof}

Finally, by \prp{min-alpha} and \cor{zme}, we have the following corollary.
\begin{corollary}\label{cor:min}
  For $\alpha\in(1,1+2/V)$, 
  \begin{align}
    H_{\min}^\epsilon(Z^m|E)_{\M_m\circ\ldots\circ\M_1(\rho)|\Omega}
    \geq m f(s) - m\left(\frac{\alpha-1}{4}\right)V^2 - \frac{1}{\alpha-1}\log\left(\frac{2}{\epsilon^2\Pr_\sigma(\Omega)^2}\right).
  \end{align}
\end{corollary}

Note that \cor{min} holds for any $\alpha\in(1,1+2/V)$.
Thus we optimize the bound by finding a good $\alpha$.
\begin{corollary}\label{cor:eat}
  For $V=2(\log(2d_Z+1)+\|\nabla f\|_\infty)$, 
  \begin{align}
    H_{\min}^\epsilon(Z^m|E)_{\M_m\circ\ldots\circ\M_1(\rho)|\Omega}
    \geq m f(s) - \sqrt{m} V\sqrt{\log\frac{2}{\Pr_\sigma[\Omega]^2\epsilon^2}}
  \end{align}
\end{corollary}
\begin{proof}
To optimize the parameter $\alpha$, we set 
\begin{align}
  m\left(\frac{\alpha-1}{4}\right) V^2  = \frac{1}{\alpha-1} \log\left(\frac{2}{\epsilon^2\Pr_\sigma[\Omega]^2}\right),
\end{align}
which gives
\begin{align}
  \alpha = 1 + \left( \frac{4}{mV^2}\log\left(\frac{2}{\epsilon^2\Pr_\sigma[\Omega]^2}\right)\right)^{1/2}.
\end{align}
This implies that 
\begin{align}
  m\left(\frac{\alpha-1}{4}\right) V^2 
  = \frac{\sqrt{m}V}{2}\sqrt{\log\left(\frac{2}{\epsilon^2\Pr_\sigma[\Omega]^2}\right)}.
\end{align}
\end{proof}

In a spot-checking protocol, the verifier in each round tosses a biased coin $T_i\sim\Bernoulli(\gamma)$ for probability $\gamma\in[0,1]$. \
If $T_i=1$, then the protocol enters a test round, in which case the verifier counts the score. \
Otherwise, if $T_i=0$, then the protocol enters a generation round, in which case the verifier does not calculate the score. \
Effectively, this sets the score to zero.

To see how spot-checking works with our modification, we follow the idea from \cite{DFR20}. 
In particular, we multiply $G$ by the factor $\gamma$, i.e., we consider a new linear operator $G'=\gamma G$.
Furthermore, we choose a new min-tradeoff function $f'=\frac{1}{\gamma}f$. \
With the new choices $G'$ and $f'$, the above analysis establishes the same lower bound on the smooth conditional min-entropy. \

The change for spot-checking does seem to allow us to choose an arbitrarily small $\gamma$. \
However, an arbitrarily small $\gamma$ should not work since  taking $\gamma\to 0$ would imply the entropy accumulation protocol does not require any verification. \ %
Instead, we want to choose $\gamma$ sufficiently large such that a \emph{good} device can satisfy $\Omega$ with probability asymptotically close to one. \
For example, in an entropy accumulation protocol from a violation of Bell's inequality, the score is a value in $[\omega_c,\omega_q]$, where $\omega_c$ and $\omega_q$ are constants describing the best score achievable from a classical and a quantum device respectively. \
By Hoeffding's inequality, it suffices to take $\gamma = O((\log m)/m)$ for a perfect device to succeed with probability $1-O(1/m)$. \

In the following sections of this paper, we are aiming to prove a lower bound on the conditional min-entropy in an LXEB-based accumulation protocol. \
In \sec{perfect}, we have shown that the concentration of the collision probability over choices of $C$ allows us to conclude that it suffices to verify samples from $O(\log n)$ circuits for a perfect device to succeed with probability $1-O(1/n)$ by the same reasoning. \

\section{A General Device with No Side Information}\label{sec:general-empty}

In this section, we consider the following situation where the device does not share an entanglement with the adversary, but the circuit can be learned.
In this setting, we give a protocol in which conditioned on the event $\Omega$ of passing $\LXEB_{1+\delta,k}$, the entropy is accumulated, for sufficiently large $\delta=\Omega(1)$.
More formally, we aim to show $H_{\min}^\epsilon(Z_1\ldots Z_m | C_1\ldots C_m)_{\rho|\Omega}$ has a lower bound $\Omega(nm)$.

In \sec{llqsv}, we define a problem, called the Long List Quantum Supremacy Verification (LLQSV) problem, which is to determine whether a string $s$ is sampled from a random circuit $C$, or it is independently sampled according to the uniform distribution. \ 
Our hardness assumption, called the Long List Hardness Assumption (LLHA), states that $\LLQSV$ is hard for $\QCAM$ protocols with access to a quantum advice state. \
In \sec{llqsv-single}, we show that if LLHA holds, then any device must generate min-entropy $\Omega(n)$ with probability $\Omega(1)$ over choices of $C$. \ 
Since conditional von Neumann entropy is calculated by taking the expectation, this implies that every device passing $\LXEB_{1+\delta,k}$ must establish a conditional von Neumann entropy lower bound $\Omega(n)$ on the input circuit. 

We state our protocol in \fig{llqsv}. In particular, the verifier randomly selects $\Omega(\log n)$ test rounds by chossing to verifier each round with probability $\gamma=\Omega((\log n)/m)$, in which the verifier determines whether the device passes $\LXEB_{1+\delta,k}$. \
By the conditional von Neumann entropy lower bound established in \sec{llqsv} and the entropy accumulation theorem shown in \sec{eat}, in \sec{llqsv-eat}, we establish an $\Omega(nm)$ lower bound of the smooth conditional min-entropy.

Since LLHA is a seemingly strong assumption, we must justify it. \ 
In \sec{llqsv-ro}, we prove that ralative to a random oracle, the assumption that $\LLQSV$ is hard for $\QCAM$ protocols with quantum advice. \
This implies that there exists a circuit distribution with which $\LXEB_{1+\delta,k}$ can be used to certify randomness in a sequential process, if the device is given oracle access to the random circuit. \ 

\subsection{The Long List Quantum Supremacy Verification Problem}\label{sec:llqsv}

In this section, we formally define the problem LLQSV, as follows. \
\begin{problem}[Long List Quantum Supremacy Verification (LLQSV)]\label{prob:LLQSV}
Given a list of $M=O(N^3)$ circuit-string tuples $\{(C_i,s_i):i\in[M]\}$, distinguish the following cases: 
\begin{itemize}
\item \textbf{Yes-case}: for each $i\in[M]$, $C_i\sim\D$ and $s_i$ is sampled from $C_i$, i.e., $s_i\sim p_{C_i}$.
\item \textbf{No-case}: for each $i\in[M]$, $C_i\sim\D$ and $s_i$ is sampled uniformly (hence independent of $C_i$).
\end{itemize}
\end{problem}

We will show that if 
\begin{align}
\LLQSV(\D)\notin \QCAM\mathsf{TIME}(2^B n^{O(1)})/\mathsf q(2^B n^{O(1)}), 
\end{align}
then any quantum algorithm that runs in $n^{O(1)}$ time solving $\LXEB_{b,k}$ with probability $q$, must output $s_1,\ldots,s_k$ of min-entropy at least $B$ with probability $\frac{bq-1}{b-1}$ over choices of $C$. \
Thus the parameter $B$ determines the min-entropy lower bound. \
We prove this assuming what we call the Long List Hardness Assumption (LLHA). \ 
\begin{assumption}[$\LLHA_B(\D)$]\label{asm:llha}
  There exists no $\QCAM$ protocol in which the quantum Arthur solves $\LLQSV(\D,V)$ in time $2^B n^{O(1)}$ given access to a quantum advice of length $2^B n^{O(1)}$.
  In other words, $\LLQSV(\D)\notin\QCAM\TIME(2^B n^{O(1)})/\mathsf q(2^B n^{O(1)})$.
\end{assumption}

\subsection{LLHA implies Certified Randomness}\label{sec:llqsv-single}

In this section, we show that $\LLHA_B(\D)$ implies that any device passing $\LXEB_{b,k}$ must output samples of min-entropy $\Omega(n)$ given $C\sim\D$. \
Our reduction relies on the Goldwasser-Sipser protocol for approximate counting \cite{GS86}, explained as follows. \
Recall that the instance consists of $M$ tuples $(C_i,s_i)$ for $i\in[M]$, where $C_i\sim\D$ and $s_i$ is either sampled from $p_{C_i}$ or from the uniform distribution $\U$. \
Let $V$ be a quantum-classical Merlin-Arthur protocol running in time $2^Bn^{O(1)}$ such that $V$ accepts at least $\kappa$ tuples for a yes instance and accepts at most $(1-\Omega(\epsilon))\kappa$ tuples, both with probability $1-2^{-\Omega(n)}$. \
Then this immediately yields a quantum-classical Arthur-Merlin protocol: \
\begin{enumerate}
\item[] Input: Both Arthur and Merlin receives an instance $(C_1,s_1),\ldots,(C_M,s_M)$. \

\item Arthur samples a random hash function $h:[M]\to[R]$ for $R$ determined later. \

\item Merlin sends $i\in[M]$ and a proof $\pi$. \ %

\item Arthur accepts if $V$ accepts $(C_i,s_i)$ with witness $\pi$ and $h(i)=y$.
\end{enumerate}
The gap of the protocol is $\Omega(\epsilon^2/\alpha)$, proved as follows. \
\begin{lemma}\label{lem:goldwasser-sipser}
  For real number $\alpha> 1$ and integer $\kappa$, there exists an Arthur-Merlin protocol which on input the description of a set $S$, determines $\alpha \kappa\geq |S|\geq \kappa$ or $|S|\leq (1-\epsilon)\kappa$ with gap at least $\frac{\epsilon^2}{4\alpha}$ using a hash function of range size $R=2\alpha\kappa/\epsilon$.
\end{lemma}
\begin{proof}
  It suffices to get an tight upper bound and a lower bound of the probability 
  \begin{align}
    \Pr_{h,y}[\exists x\in S, h(x)=y]
  \end{align}
  for a random hash function $h$ of range size $R$. \
  By union bound, an upper bound is $\sum_{x\in S}\Pr_{h,y}[h(x)=y]=\frac{|S|}{R}$. \
  For a lower bound, by the inclusion-exclusion principle\footnote{The principle states that $\Pr[\bigvee_i E_i]\geq \sum_i \Pr[E_i]-\sum_{i<j}\Pr[E_i\wedge E_j]$.}
  \begin{align}\nonumber
    \Pr_{h,y}[\exists x\in S, h(x)=y] 
    &\geq \frac{|S|}{R} - \sum_{x\leq x', x,x'\in S} \Pr_{h,y}[h(x)=h(x')=y] \\\nonumber
    &= \frac{|S|}{R} - \binom{|S|}{2}\frac{1}{R^2} \\
    &= \frac{|S|}{R}\left(1 - \frac{|S|-1}{2R}\right) > \frac{|S|}{R}\left(1 - \frac{|S|}{R}\right)
  \end{align}
  In the yes case, $\alpha\kappa\geq |S|\geq\kappa$ and $R=2\alpha\kappa/\epsilon$, and
  \begin{align}\nonumber
    \Pr_{h,y}[\exists x\in S,h(x)=y]
    &\geq \frac{|S|}{R}\left(1 - \frac{|S|}{R}\right) \\\nonumber
    &\geq \frac{\kappa}{R}\left(1-\frac{\alpha\kappa}{R}\right) \\ 
    &\geq \frac{\kappa}{R}\left(1-\frac{\epsilon}{2}\right),
  \end{align}
  whereas in the no case, $|S|\leq (1-\epsilon)\kappa$, and
  \begin{align}
    \Pr_{h,y}[\exists x\in S,h(x)=y] \leq \frac{|S|}{R} \leq \frac{\kappa}{R}(1-\epsilon).
  \end{align}
  The gap is at least $\frac{\kappa\epsilon}{2R}=\frac{\epsilon^2}{4\alpha}$.
\end{proof}

Assume that there is a quantum device $\A$ which solves $\LXEB_{b,k}$ with probability $q$ over choices of $C$, and outputs $k$ tuples $(s_1,\ldots,s_k)$ of min-entropy at most $B/2$ with probability $p$ such that $p$ and $q$ satisfy
\begin{align}\label{eq:llqsv-contra-4}
  p > \frac{b}{b-1} (1-q) + \epsilon.
\end{align}
For breaking $\LLHA_B(\D)$, we consider the following random variable $Y_\tau(C,s)$ for $\tau\in[0,1]$:
\begin{align}\label{eq:y-tau}
  Y_\tau(C,s)
  := 
  \left\{
  \begin{array}{ll}
    1 & \text{if $\exists d=(z_1,\ldots,z_k)$, $\Pr[\A(C)=d]\geq\frac{\tau}{2^{B/2}}$ and $s\in\A(C)$} \\ %
    0 & \text{otherwise.}
  \end{array}
  \right.
\end{align}
Here $s\in\A(C)$ means that the samples $s\in\{z_1,\ldots,z_k\}$ for samples $(z_1,\ldots,z_k)\sim\A(C)$. \
Define
\begin{align}\label{eq:mu1-mu0}
  \mu_1(\tau) &= \Exp_{C\sim\D,s\sim p_C}[Y_\tau(C,s)], &
  \mu_0(\tau) &= \Exp_{C\sim\D,s\sim \U}[Y_\tau(C,s)],
\end{align}
where $\U$ is the uniform distribution over $\bit^n$. \

Let $p(\tau)$ be the probability that $\A(C)$'s maximum probability is at least $\tau/2^{B/2}$ for $C\sim\D$, i.e.,
\begin{align}\label{eq:p-tau}
  p(\tau) := \Pr_{C\sim\D}\left[\max_d \Pr[\A(C)=d] \geq \frac{\tau}{2^{B/2}}\right].
\end{align}
The following lemma shows that the ratio $\mu_1(\tau)/\mu_0(\tau)$ is bounded by $b\cdot (p(\tau)+q-1)/p(\tau)$. \
\begin{lemma}\label{lem:y-tau}
  For $\tau\in[0,1]$, let $Y_\tau$ be defined as in \eq{y-tau}, $\mu_1(\tau),\mu_0(\tau)$ as in \eq{mu1-mu0} and $p(\tau)$ as in \eq{p-tau}.
  Then
  \begin{align}
    \frac{\mu_1(\tau)}{\mu_0(\tau)} \geq b\cdot \frac{p(\tau)+q-1}{p(\tau)}.
  \end{align}
\end{lemma}
\begin{proof}
  Note that in each case, we are bounding
  \begin{align}
    \Exp_{C\sim\D,s}[Y_\tau(C,s)] 
    = \Pr_{C\sim\D,(z_1,\ldots,z_k)\sim\A(C),s}\left[\exists d, \Pr[\A(C)=d]\geq \frac{\tau}{2^{B/2}} \wedge s\in\{z_1,\ldots,z_k\}\right],
  \end{align}
  where $s\sim p_C$ in the yes case, or $s\sim\U$ in the no case.
  Define
  \begin{align}
    G_\tau(C) := \Id\left[\max_d\Pr[\A(C)=d]\geq\frac{\tau}{2^{B/2}}\right].
  \end{align}
  By \eq{p-tau}, $p(\tau)=\Exp_{C\sim\D}[G_\tau(C)]$.
  \begin{itemize}
  \item For the yes case, 
    \begin{align}\nonumber
      \mu_1(\tau) 
      &= \Exp_{C\sim\D,s\sim p_C}[Y_\tau(C,s)] \\\nonumber
      &= \Pr_{C\sim\D,O\sim\A(C),s\sim p_C}\left[G_\tau(C) \wedge s\in O\right] \\
      &\geq\Pr_{C\sim\D,O\sim\A(C),s\sim p_C}\left[G_\tau(C) \wedge s\in O \wedge V(C,O)\right],
    \end{align}
    where $V(C,O)=1$ if $\LXEB_{b,k}$ accepts $O=(z_1,\ldots,z_k)$, i.e.,
    \begin{align}\label{eq:llqsv-lxeb}
    \sum_{i=1}^k p_C(z_i) \geq \frac{bk}{N}
      \end{align}
      and 0 otherwise.
    By union bound, 
    \begin{align}\label{eq:y-yes-1}\nonumber
      \Pr_{C\sim\D,O\sim\A(C),s\sim p_C}\left[G_\tau(C)\wedge V(C,O)=1\right]
      &\geq \Pr_{C\sim\D}[G_\tau(C)] + \Pr_{C\sim\D,O\sim\A(C)}[V(C,O)] - 1 \\
      & = p(\tau)+q-1.
    \end{align}
    Furthermore, the conditional probability
    \begin{align}\label{eq:y-yes-2}
      \Pr_{C\sim\D,O\sim\A(C),s\sim p_C}\left[s\in O \left|G_\tau(C)\wedge V(C,O)\right.\right] \geq \frac{bk}{N}
    \end{align}
    since if $V(C,O)=1$, the probability that $\Pr_{s\sim p_C}[s\in O]=\sum_{z\in O}p_C(z)\geq bk/N$ by \eq{llqsv-lxeb}.
    Combining \eq{y-yes-1} and \eq{y-yes-2},
    \begin{align}\label{eq:y-yes}
      \mu_1(\tau) = \frac{bk}{N} (p(\tau)+q-1).
    \end{align}

    \item For the no case, since $s$ is independently sampled,
      \begin{align}\nonumber\label{eq:y-no}
        \mu_0(\tau) 
        &= \Pr_{C\sim\D,O\sim\A(C),s\sim\U}\left[ G_\tau(C) \wedge s\in O \right] \\ 
        &= \frac{k}{N}\Pr_{C\sim\D}[G_\tau(C)] = \frac{k}{N}\cdot p(\tau).
      \end{align}
  \end{itemize}
  By \eq{y-yes} and \eq{y-no}, we conclude the proof.
\end{proof}
Furthermore, we show that the lower bound is monotonically non-increasing.
\begin{lemma}\label{lem:non-increasing}
  The ratio $(p(\tau)+q-1)/p(\tau)$ is monotonically non-increasing for $\tau\in[0,1]$.
\end{lemma}
\begin{proof}
  For $1\geq \alpha\geq \beta\geq 0$, $p(\alpha)\leq p(\beta)$, and %
  \begin{align}
    \frac{p(\beta)+q-1}{p(\alpha)+q-1} \geq \frac{p(\beta)}{p(\alpha)}
  \end{align}
  This implies that $\frac{p(\tau)+q-1}{p(\tau)}$ is monotonically non-increasing for $\tau\in[0,1]$.
\end{proof}

\lem{y-tau} and \lem{non-increasing} imply that 
\begin{align}\label{eq:llqsv-contra-6}
  \frac{\mu_1(\tau)}{\mu_0(\tau)} \geq \frac{\mu_1(1)}{\mu_0(1)} \geq b\cdot \frac{p+q-1}{p} \geq 1+\epsilon.
\end{align}
While this seems to give us a sufficient condition for applying the Goldwasser-Sipser protocol, a caveat is that we do not know how to verify if $Y_\tau$ accepts $(C,s)$ (in particular, $Y_1$) in time $2^Bn^{O(1)}$.
The reason is that the condition that $\A(C)$ outputs a string $d$ with probability greater than $\tau/2^B$ cannot be efficiently verified, even when a witness $d$ is given.
The next step is to show that for sufficiently large $T$, there exists $j\in[T]$ such that for $\tau=1/2+j/T$,
  $\mu_1(\tau)\geq (1+\epsilon/2)\mu_0(\tau-1/T)$.

\begin{lemma}\label{lem:gap}
  Let $\mu_1$ and $\mu_0$ be defined as in \eq{mu1-mu0}.
  Assume that \eq{llqsv-contra-4} holds.
  Then for $T\geq \frac{8}{\epsilon}\log(\frac{N}{\epsilon})$, there exists $j\in[T]$ such that 
  \begin{align}
    \mu_1(1/2+j/T) \geq (1+\epsilon/2) \mu_0(1/2+(j-1)/T).
  \end{align}
\end{lemma}
\begin{proof}
  We prove the lemma by contrapositive.
  Suppose that for every $j\in[T]$, 
  \begin{align}\label{eq:llqsv-contra-5}
    \mu_1(1/2+j/T)<(1+\epsilon/2) \mu_0(1/2+(j-1)/T).
  \end{align}
  Then expanding the ratio into a telescoping product, 
  \begin{align}\nonumber
    \frac{\mu_1(1/2)}{\mu_1(1)} 
    &= \prod_{j=1}^T \frac{\mu_1(1/2+j/T)}{\mu_1(1/2+(j+1)/T)} \\\nonumber
    &\geq (1+\epsilon/2)^{-T}\prod_{j=1}^T \frac{\mu_1(1/2+j/T)}{\mu_0(1/2+j/T)} \\\nonumber
    &\geq \left(\frac{1+\epsilon}{1+\epsilon/2}\right)^T \\
    &\geq (1+\epsilon/4)^T.
  \end{align}
  The first inequality holds by \eq{llqsv-contra-5}.
  The second holds by \eq{llqsv-contra-6}. %
  The third holds by the inequality $1+\epsilon-(1+\epsilon/2)(1+\epsilon/4)=\epsilon/4\cdot(1-\epsilon/2)\geq 0$ for $\epsilon\leq 2$.
  Taking $T\geq \frac{8}{\epsilon}\log(n/\epsilon)$,
  \begin{align}\nonumber
    \mu_1(1/2) 
    &> (1+\epsilon/4)^T\cdot\mu_1(1) \\\nonumber
    &\geq \frac{k}{N} \cdot \epsilon (1+\epsilon/4)^T \\\nonumber
    &\geq \frac{k}{N} \cdot \epsilon e^{\epsilon T/8} \\
    &\geq k.
  \end{align}
  This contradicts the fact that $\mu_1(1/2)\leq 1$.
\end{proof}

With the gap, we can consider the following Merlin-Arthur protocol which verifies that $s\sim p_C$ in time $2^B\cdot T^2\cdot n^{O(1)}$: 
\begin{enumerate}
\item[] Both Arthur and Merlin are given access to a circuit $C$ and a string $s\in\bit^n$.
\item Merlin sends $d=(z_1,\ldots,z_k)$.
\item Arthur gets $K+1$ samples $d_1,\ldots,d_K,O\sim\A(C)$ for integer $K$ determined later. \ 
Arthur accepts if 
  \begin{align}\label{eq:fraction}
    \frac{1}{K}|\{\ell:d_\ell=d\}| \geq \left(\tau-\frac{1}{2T}\right) 2^{-B/2},
  \end{align}
  and $s\in O$, and rejects otherwise.
\end{enumerate}
This quantum-classical Merlin-Arthur protocol solves the problem that either $(C,s)$ is accepted by $Y_{\tau}$, or rejected by $Y_{\tau-1/T}$. \
\begin{lemma}\label{lem:y-to-ma}
  For $\tau\in[1/2,1]$, $\eta=O(1/N^2)$ and $K\geq 4T^22^Bn$, the above process satisfies the following conditions: \
  \begin{enumerate}
  \item If $Y_\tau$ accepts $(C,s)$, then Arthur accepts with probability $1-\eta$, and
  \item if $Y_{\tau-1/T}$ rejects $(C,s)$, then Arthur rejects with probability $1-\eta$.
  \end{enumerate}
\end{lemma}
\begin{proof}
  If $Y_\tau$ accepts $(C,s)$, then there exists $d=(z_1,\ldots,z_k)$ such that $\Pr[\A(C)=d]\geq \tau/2^B$ and $s\in\{z_1,\ldots,z_k\}$. \
  By Hoeffding's inequality, the fraction $\tilde \nu$ in \eq{fraction} no more than $(\tau-\frac{1}{2T})2^{-B/2}$ occurs with probability at most 
  \begin{align}
    \Pr\left[ |\tilde\nu-\Pr[\A(C)=d]| > \frac{1}{2T}2^{-B/2} \right] \leq 2e^{-2\frac{K}{4T^22^B}} \leq 2e^{-2n}.
  \end{align}
  
  If $Y_{\tau-1/T}$ rejects $(C,s)$, then for every $d$ such that $s$ is contained in $d$, $\Pr[\A(C)=d]<(\tau-1/T)/2^{B/2}$. \
  Again, by Hoeffding's inequality, \eq{fraction} occurs with probability at most $2e^{-2n}$. \
\end{proof}

For $M\geq N^3$, the quantum-classical Merlin-Arthur protocol accepts more tuples in an yes instance than a no instance, with overwhelming probability. \
\begin{lemma}\label{lem:llqsv-qcma}
  For $M\geq N^3$ and $\eta'=2^{-\Omega(N)}$, there exists an integer $\kappa\in[M]$ such that
  with probability at least $1-\eta'$, Arthur accepts at least $\kappa$ tuples in an yes instance and accept at most $\kappa'=(1-\epsilon/5+O(\epsilon^2))\kappa$ tuples in a no instance. \
\end{lemma}
\begin{proof}
  Recall that $\mu_1(\tau)=\Exp_{C\sim\D,s\sim p_C}[Y_\tau(C,s)]$. \
  Since each tuple is independently sampled both in the yes instance and in the no instance, 
  by \lem{y-to-ma},
  Authur accepts each sample with probability at most $\mu_0(\tau)+\eta$. \
  Thus by Hoeffding's inequality, for $\kappa'=(\mu_0(\tau)+\eta)(1+\epsilon/8)M$, 
  \begin{align}\nonumber
    \Pr_{C_1,\ldots,C_M\sim\D,s_1,\ldots,s_M\sim\U}[\text{Arthur accepts more than $\kappa'$ tuples}]
    &\leq 2e^{-M\epsilon^2(\mu_0(\tau)+\eta)^2/32} \\\nonumber
    &= 2e^{-\Omega(M\epsilon^2k^2p^2/N^2)} \\
    &\leq 2e^{-\Omega(M\epsilon^4k^2/N^2)},
  \end{align}
  since by \eq{llqsv-contra-4}, $\mu_0(\tau)=p(\tau)k/N\geq pk/N\geq \epsilon k/N$.
  Taking $M\geq N^3$, the upper bound is $2e^{-\Omega(N)}$.

  Similarly, for $\kappa=(\mu_1(\tau)-\eta)(1-\epsilon/8)M$, 
  \begin{align}\nonumber
    \Pr_{C_1,\ldots,C_M\sim\D,s_i\sim p_{C_i}}[\text{Arthur accepts fewer than $\kappa$ tuples}]
    &\leq
      2e^{-M\epsilon^2(\mu_1(\tau)-\eta)^2/32} \\
    &\leq 
      2e^{-\Omega(M\epsilon^2p^2k^2/N^2)},
  \end{align}
  since $\eta=O(1/N^2)$ and $\mu_1(\tau)-\eta= \Omega(\mu_1(\tau))=\Omega(\mu_0(\tau))=\Omega(pk/N)$.
  
  It remains to give an upper bound of the ratio: Since $\eta=o(\epsilon)$,
  \begin{align}\nonumber
    \frac{\kappa'}{\kappa} 
    &\leq \frac{\mu_0(\tau)+\eta}{\mu_1(\tau)-\eta} \frac{1+\epsilon/8}{1-\epsilon/8} \\\nonumber
    &\leq \frac{\mu_0(\tau)+\eta}{\mu_1(\tau)-\eta} \frac{1}{1-\epsilon/4} \\\nonumber
    &\leq \frac{1}{1+\epsilon/4-O(\eta)}\\\nonumber
    &\leq \frac{1}{1+\epsilon/5}\\
    &\leq 1-\epsilon/5 +O(\epsilon^2).
  \end{align}
\end{proof}

We have proved all the statements needed for proving our main theorem in this section. \
Now we describe the quantum-classical Arthur-Merlin protocol. \

\begin{enumerate}
\item[] Both Merlin and Arthur are given access to $(C_1,s_1),\ldots,(C_M,s_M)$ and advice strings including integers $T$, $j\in[T]$ and $R$.
\item Arthur chooses a random hash function $h:[M]\to [R]$ and $y$ uniformly from $[R]$.
\item Merlin sends $i$ and $d=(z_1,\ldots,z_k)$.
\item For $K\geq 4T^22^Bn$, Arthur takes $K+1$ samples $d_1,\ldots,d_K,O\sim\A(C_i)$. \
  He accepts if 
  \begin{enumerate}
    \item $\frac{1}{K}|\{\ell:d_\ell=d\}| \geq \left(\tau-\frac{1}{2T}\right) 2^{-B/2}$, where $\tau=1/2+j/T$,

    \item $s_i\in O$, and
    \item $h(i)=y$.
  \end{enumerate}
\end{enumerate}

We prove the main theorem in this section. \
\begin{theorem}\label{thm:llqsv-basic}
  If there exists a device $\A$ which runs in quantum polynomial time and satisfies \eq{llqsv-contra-4}, 
  then there is a quantum-classical Arthur-Merlin protocol which on input an $O(n)$-bit advice string, solves $\LLQSV_B(\D)$. \
  In other words, $\LLQSV_B(\D)\in\QCAM\TIME(2^Bn^{O(1)})/O(n)$. \
\end{theorem}
\begin{proof}
  We show that the above protocol solves $\LLQSV_B(\D)$ with a constant gap. \
  By \lem{llqsv-qcma}, $\kappa'\leq (1-\epsilon/6)\kappa$ with probability $1-2^{-\Omega(N)}$. \
  Then by \lem{goldwasser-sipser}, for $R=12\alpha\kappa/\epsilon$, since $\alpha=1+O(\epsilon)$, the gap is $\Omega(\epsilon)$. \ 
  Thus running an $(1/\epsilon)^{O(1)}$-fold parallel repetition of the above protocol yields a constant gap. \
  For the length of the advice string, since $1/\epsilon=n^{O(1)}$, $T=n^{O(1)}$ and $R\leq M$ and it suffices to choose $M=N^3$, the total length is $O(n)$. \
\end{proof}

As a corollary, if $\LLHA_B(\D)$ is true, then 
\begin{align}
  p \leq \frac{b}{b-1}(1-q) + n^{-\omega(1)}.
\end{align}
\begin{corollary}\label{cor:llqsv-basic}
  For integer $n$, assume that $\LLHA_B(\D)$ holds for distribution $\D$ over circuit acting on $n$ qubits. \
  Then for every device $\A$ passes $\LXEB_{b,k}$ with probability $q$ over choices of $C\sim\D$, 
  \begin{align}
    \Pr_{C\sim\D}\left[H_{\min}(\A(C)) \geq B/2\right] \geq \frac{bq-1}{b-1} - n^{-\omega(1)}.
  \end{align}
\end{corollary}

Since min-entropy is the smallest quantity in the family of R\'{e}nyi entropies, we establish a lower bound on the von Neumann entropy.
\begin{theorem}\label{thm:llqsv-single}
    For integer $n$, assume that $\LLHA_B(\D)$ holds for distribution $\D$ over circuits acting on $n$ qubits. \
    Then for any device which on input a circuit $C$, outputs a classical state $\psi$ over $\bit^{nk}$ solving $\LXEB_{b,k}$ with probability $q$, it holds that 
    \begin{align}
        H(Z|C)_\psi \geq \frac{B}{2}\left(\frac{bq-1}{b-1} - n^{-\omega(1)}\right).
    \end{align}
\end{theorem}
\begin{proof}
  By \cor{llqsv-basic}, since min-entropy is the smallest quantity in the family of R\'{e}nyi entropies, 
  \begin{align}\nonumber
      \Pr_{C\sim\D}\left[ H(Z)_{\psi^C} \geq B/2\right] 
      &\geq \frac{bq-1}{b-1} - n^{-\omega(1)}
  \end{align}
  where $\psi^C$ is the distribution output by the device conditioned on $C$. \
  Since for every $C$, the von Neumann entropy is non-negative, by definition of conditional von Neumann entropy, we conclude the proof. \
\end{proof}

\subsection{Entropy Accumulation}\label{sec:llqsv-eat}

To apply the EAT we showed in \sec{eat}, we must give a min-tradeoff function $f$ such that for devices outputting a classical state $\psi$ and solving $\LXEB_{b,k}$ with probability $q$, $H(Z|C)_\psi\geq f(q)$.
From \thm{llqsv-single}, the affine function can be defined as 
\begin{align}\label{eq:f}
    f(q) := \frac{bq-1}{b-1}\frac{B}{2} - c,
\end{align}
for $c=n^{-\omega(1)}$ independent of $q$.

For spot checking, by \lem{lxeb-product}, we change the circuit $\Omega(\log n)$ times so that a perfect device can pass the verification with overwhelming probability.
With the parameters, we are ready to give our protocol in \fig{llqsv}.

\begin{figure}
    \hrule\vspace{.5em}
  Input: security parameter $n$, a distribution $\D$ over circuits on $n$ qubits, the threshold constant $b\in[1,2]$, the number of samples $k=O(n^2)$ per iteration, the number of rounds $m$, and the fraction $\gamma=O((\log n)/m)$ of circuit updates.\\

  The protocol:
  
  \vspace{.5em}
  \begin{enumerate}
    \item %
      For $i=1,\ldots,m$, run the following steps:
  \begin{enumerate}
  \item The verifier samples $T_i\sim\Bernoulli(\gamma)$. If $T_{i-1}=1$ (when $i>1$) or $i=1$, the device samples $C_i\sim\D$.
  Otherwise, the device sets $C_i=C_{i-1}$.
  The verifier sends $C_i$ to the device (and keeps $T_i$ secret).
  
  \item The device returns $k$ samples $d_i=(z_1,\ldots,z_k)$.
  
  \item If $T_i=1$, the verifier sets 
  \begin{align}
      W_i = \Id\left[\frac{1}{k}\sum_{i=1}^k p_C(z_i) \geq \frac{b}{N} \wedge E_i\right].
  \end{align}
  where $E_i=0$ if there exist distinct $\ell,\ell'\in\{j:C_j=C_i\}$ such that the samples $d_\ell=(z_{\ell 1},\ldots,z_{\ell k})$ and $d_{\ell'}=(z_{\ell' 1},\ldots,z_{\ell' k})$ are not all distinct. \ 
  (This check is used to prevent the device repeats responses for any two rounds using the same challenge circuit.) \ 
  If $T_i=0$, the verifier sets $W_i=\bot$.
  \end{enumerate} 
  
  \item Let $t=|\{i:T_i=1\}|$ be the number of test rounds. The verifier computes 
    \begin{align}
      W = \sum_{i:T_i=1} W_i.
      \end{align}
      If $W\geq 0.99 t$, then the verifier accepts and outputs $(d_1,\ldots,d_m)$ to the quantum-proof randomness extractor.
  \end{enumerate} 
  \hrule\vspace{1em}
  \caption{The entropy accumulation protocol based on $\LLHA$.
}
  \label{fig:llqsv}
  \end{figure}

We prove that the entropy accumulates.
In the following theorem, let $D=D_1\ldots D_m$ denote the responses received from the device, $C=C_1\ldots C_m$ denote the circuit sent from the device, and $T=T_1\ldots T_m$ the flags indicating whether a test round is executed (see \fig{llqsv}).
\begin{theorem}
  Assume that $\LLHA_B(\D)$ holds for distribution $\D$.
  Conditioned on the event $\Omega$ that the verifier does not abort in the protocol in \fig{llqsv}, 
  \begin{align}
      H_{\min}^\epsilon(D |CT)_{\rho|\Omega} 
      \geq 
      n\left((0.99-0.01/\delta) \beta m- O\left( \sqrt{m}\right)\cdot \sqrt{\log \frac{2}{p^2 \epsilon^2}} \right),
  \end{align}
  where $\beta=\frac{B}{2n}$, $\rho$ is the output state, and $p$ is the probability of non-aborting.
  Furthermore, there exists a device which solves $\LXEB_{2,k}$ for $k=O(n^2)$ with probability $1-o(1)$.
\end{theorem}
\begin{proof}
  Let $f$ be as defined in \eq{f}.
  To apply \cor{eat}, we set $d_Z=N$ and $\|\nabla f\|_\infty = \frac{b}{b-1}\frac{B}{2}$, and therefore $V\leq (n+1)+\frac{b}{b-1}\frac{B}{2}$. 
  For $b=1+\delta$ and $q\geq 0.99$,
  \begin{align}
      \frac{bq}{b-1} \geq \frac{(1+\delta)0.99-1}{\delta} \geq 0.99-0.01/\delta.
  \end{align}

  This implies a lower bound
  \begin{align}\nonumber
      &(0.99-0.01/\delta) B m/2 - \sqrt{m} \left((n+1)+(1/\delta+1)B/2\right)\sqrt{\frac{2}{p^2\epsilon^2}} \\\nonumber
      &\qquad= n\left((0.99-0.01/\delta)\beta m - \sqrt{m}\cdot (1+1/n+(1/\delta+1)\beta)\sqrt{\frac{2}{p^2\epsilon^2}}\right) \\
      &\qquad\geq
      n\left((0.99-0.01/\delta)\beta m - O(\sqrt{m})\cdot \sqrt{\frac{2}{p^2\epsilon^2}}\right),
  \end{align}
  where $\beta=\frac{B}{2n}$.
\end{proof}

\subsection{LLHA Relative to a Random Oracle}\label{sec:llqsv-ro}

Now, we turn our attention to justifying \asm{llha} relative to a random oracle.
Given access to a random function $f:\bit^n\to\{+1,-1\}$, there is a simple algorithm which samples from the Fourier spectrum:
\begin{align}\label{eq:fourier}
  \ket{0^n}
  \xmapsto{H^{\otimes n}} \sum_{x\in\bit^n} \ket{x}
  \xmapsto{O_f} \sum_{x\in\bit^n} f(x)\ket{x}
  \xmapsto{H^{\otimes n}} \sum_{y\in\bit^n} \hat f(x)\ket{x},
\end{align}
where $\hat f(x):=\frac{1}{N}\sum_{y\in\bit^n}(-1)^{x\cdot y}f(y)$ is the Fourier coefficient of $f$. \
Given oracle access to $M$ random Boolean functions and the samples, we defined a black-box version of the LLQSV problem. \
For concreteness, the algorithm is given access to a unitary $\O:\ket{i,x}\mapsto f_i(x)\ket{i,x}$, where $i$ ranges from 1 to $M$ and $f_1,\ldots,f_M$ are random functions sampled from the uniform distribution $\F_n$ over $n$-bit Boolean functions. \
The black-box version of $\LLQSV$ is formally defined as follows. \
\begin{problem}[{Black-box $\LLQSV$}]
  For $M=2^{O(n)}$, given access to strings $s_1,\ldots,s_M\in\bit^n$ and to functions $f_1,\ldots,f_M\sim\F_n$ through the unitary $\O:\ket{i,x}\mapsto f_i(x)\ket{i,x}$, determine whether, in the yes case, $s_i\sim |\hat f_i|^2$, or, in the no case, each $s_i$ is sampled from the uniform distribution (hence independent of $f_i$). %
\end{problem}

The proof is organized in the following steps. \
We prove $\LLQSV$ is not in $\QIP[2]$, the class of problems which admits a two-message quantum interactive proof system, relative to a random oracle. \
The proof is closely related to the black-box LLQSV lower bound for $\BQP$ by Bassirian, Bouland, Fefferman, Gunn, and Tal \cite{BBFGT21}, which also uses a hybrid argument \cite{BBBV97}, but here we strengthen the hardness to $\QIP[2]$. Thus the problem is not in $\QCAM$. \

However, the hybrid argument does not immediately lead to hardness for $\QCAM/\qpoly$. \
By Aaronson and Drucker's exchange theorem, $\QCAM/\qpoly\subseteq\QMA/\cpoly$, and thus it suffices to show hardness for the latter.
We then give a query lower bound using the polynomial method, and appeal to the strong direct product theorem by Sherstov \cite{She11}. 

\subsubsection{Two-message quantum interactive proofs}\label{sec:llqsv-qip2}

To show LLHA holds relative to $\O$, we first provide intuition.
If the algorithm is not given access to $\O$, then for both cases, $s_1,\ldots,s_M$ are uniform.
Thus, alternatively, we can view $\O$ as a distribution which may or may not depend on uniform $s_1,\ldots,s_M$.
This fact is stated in the following lemma.
\begin{lemma}
  Let $\F_n$ be the uniform distribution over $n$-bit Boolean functions $\bit^n\to\{+1,-1\}$ and $\U_n$ be the uniform distribution over $n$-bit strings.
  The following two sampling processes are equivalent, i.e. their output distributions are identical:
  \begin{itemize}
  \item Sample $f\sim\F_n$ and then $s\sim |\hat f|^2$. Output $(f,s)$.
  \item Sample $s\sim\U_n$ and then $f\sim\G_{n,s}$, where $\G_{n,s}(f)=\F_n(f)\cdot|\hat f(s)|^2 \cdot N$. Output $(f,s)$.
  \end{itemize}
\end{lemma}
\begin{proof}
  Let the probability density of the first distribution be $\mu(f,s)=\F_n(f)\cdot |\hat f(s)|^2$.
  The marginal 
  \begin{align}\nonumber
    \mu(s) 
    &= \sum_f \F_n(f)\cdot |\hat f(s)|^2 \\\nonumber
    &= \Exp_f[|\hat f(s)|^2] \\
    &= \frac{1}{N}.
  \end{align}
  It remains to calculate $\mu(f|s)$:
  \begin{align}\nonumber
    \mu(f|s) 
    &= \frac{\mu(f,s)}{\mu(s)} \\ 
    &= \F_n(f) \cdot |\hat f(s)|^2 \cdot N.
  \end{align}
\end{proof}

We then define the following promise problem, called Boolean Function Bias Detection ($\BFBD$). %
\begin{definition}[{Boolean Function Bias Detection ($\BFBD$)}]\label{dfn:llqsv-bfbd}
The Boolean Function Bias Detection problem is to distinguish between the follow cases:
  \begin{itemize}
  \item \textbf{Yes-case}: For each $i\in[M]$, sample $f\sim\G_{n}$, where $\G_{n}$ is the probability distribution of density 
      $\G_{n}(f) := N(1-2\Delta(f))^2\F_n(f)$,
    where $\Delta(f)$ is the fraction of elements that evaluates to $-1$, i.e., $\Delta(f):=\frac{1}{N}|\{x:f(x)=-1\}|$.
  \item \textbf{No-case}: For each $i\in[M]$, sample the distribution $f\sim\F_n$.
  \end{itemize}
\end{definition}
For function $f$, we will also call $\Delta(f)$ the distance of $f$. \
We now prove that $\BFBD$ reduces to $\LLQSV$, and thus it suffices to establish a lower bound for $\BFBD$. \
\begin{lemma}\label{lem:llqsv-reduction}
  $\BFBD$ reduces to $\LLQSV$.
\end{lemma}
\begin{proof}
  The reduction works as follows: \
  Sample $s_1,\ldots,s_M$ uniformly and simulate the oracle $\Q:\ket{i,x}\mapsto (-1)^{s_i\cdot x}\ket{i,x}$. \
  Run the protocol for $\LLQSV$ using the oracle $\Q\O$, and effectively the protocol is given oracle access to the function $g_i=\chi_{s_i}\cdot f$, where $\chi_{s_i}(x):=(-1)^{s_i\cdot x}$ for each $i\in[M]$. \
  To see why the reduction works, it suffices to show that each oracle distributed identically as the associated case for $\LLQSV$. \
  In the no case, the distribution of $\chi_s\cdot f$ for uniform $s$ and $f$ is distritubed uniformly and independently of $s$. \
  In the yes case, let $g=\chi_s\cdot f$: the density $\mu$ of $g$ can be calculated as follows: \
  \begin{align}\nonumber
    \mu(g|s) 
    &= N(1-2\Delta(g\cdot\chi_s))^2\cdot \F_n(g\cdot\chi_s) \\
    &= N(1-2\Delta(g\cdot\chi_s))^2\cdot \F_n(g).
  \end{align}
  The second equality holds since $\F_n$ is uniform. \
  It remains to show that $\hat g(s)=1-2\Delta(g\cdot \chi_s)$:
  \begin{align}\nonumber
    \hat g(s) 
    &= \frac{1}{N} \sum_{x} \chi_s(x) f(x) \\\nonumber
    &= \frac{1}{N} \sum_{x:\chi_s(x)f(x)=1} \chi_s(x) f(x) - \frac{1}{N} \sum_{x:\chi_s(x)f(x)=-1} \chi_s(x) f(x) \\
    &= 1-2\Delta(g\cdot\chi_s).
  \end{align}
\end{proof}

Though our purpose is to show $\LLQSV$, or equivalently $\BFBD$, is not in $\QCAM$, here we prove a stronger result: $\BFBD$ is not in $\QIP[2]$. \
First, we observe that each distribution can be sampled by choosing the distance $\Delta$ first, and then sampling a function $f$ according to the uniform distribution among all functions of distance $\Delta$. \
More formally, $\BFBD$ is equivalent to the following problem:
\begin{itemize}
\item \textbf{No-case}: For each $i\in[M]$, sample $D_i\sim\Binomial(N,1/2)$, the binomial distribution with $N$ trials and bias $1/2$, and set $\Delta_i=D_i/N$, i.e., the density is
  \begin{align}
    p_0(\Delta_i) = \binom{N}{N\Delta_i} 2^{-N}.
  \end{align}
  Then sample $f_i\sim \H_{\Delta_i}$, the uniform distribution over all functions $g$ of distance $\Delta(g)=\Delta_i$.
\item \textbf{Yes-case}: For each $i\in[M]$, sample $\Delta_i$ according to the distribution $p_{1}$ of density
  \begin{align}
    p_1(\Delta_i) := %
    p_0(\Delta_i)
    \cdot N(1-2\Delta_i)^2.
  \end{align}
  Then sample $f_i\sim\H_{\Delta_i}$.
\end{itemize}

Then we show that for each distribution, $\Delta$ is concentrated around $1/2$, this allows us to only consider the event that $\Delta$ is sufficiently close to $1/2$.
\begin{lemma}\label{lem:llqsv-concentration}
  For $b\in\bit$ and $\alpha>0$,
  \begin{align}
    \Pr_{\Delta_1,\ldots,\Delta_M\sim p_b}\left[\exists i, |\Delta_i-1/2|\geq \frac{\alpha}{\sqrt{N}}\right] \leq 2MNe^{-2\alpha^2}.
  \end{align}
\end{lemma}
\begin{proof}
  In the no case, since each element of $f$ is a fair coin, by Hoeffding's inequality and union bound,
  \begin{align}
    \Pr_{\Delta_1,\ldots,\Delta_M\sim p_0}\left[\exists i,|\Delta_i-1/2| \geq \frac{\alpha}{\sqrt{N}}\right] \leq 2Me^{-2\alpha^2}.
  \end{align}
  In the yes case,
  \begin{align}\nonumber
    \Pr_{\Delta_1,\ldots,\Delta_M\sim p_1}\left[ \exists i, |\Delta_i-1/2| \geq \frac{\alpha}{\sqrt{N}} \right]
    &\leq M 
    \Pr_{\Delta\sim p_1}\left[ |\Delta-1/2| \geq \frac{\alpha}{\sqrt{N}} \right] \\
    &\leq 2MNe^{-2\alpha^2}.
  \end{align}
  The second inequality holds since for $p_1(\Delta)\leq N\cdot p_0(\Delta)$ and by Hoeffding's inequality. \
\end{proof}

For each distribution, consider the event $\Omega$ that $|\Delta-1/2|\leq t$, where $t=((\ln MN^2)/N)^{1/2}=O((n/N)^{1/2})$. \
By \lem{llqsv-concentration}, conditioned on $\Omega$, the gap decreases by at most $O(1/N)$. \
Then for $(\Delta_1,\ldots,\Delta_M)\in\Omega$, we are dealing with the function distributed from $H_{\Delta_i}$ for each $i$. \
Since $\H_{\Delta}$ can be obtained by flipping $|\Delta-1/2|N$ bits of $f\sim\H_{1/2}$, we can use BBBV to give a query lower bound. \

Before we give the proof, let us characterize the behavior of a two-message quantum interactive proof system for distinguishing two distributions $\D_1$ and $\D_0$. \
In a two-message protocol, the verifier first makes $T_1$ queries to $f$ either sampled from $\D_1$ or $\D_0$ and computes a bipartite quantum state $\sigma_{AB}$. \
The verifier sends the first system $A$ to the (unbounded) prover, which performs an arbitrary quantum process $\P\in\CPTP(A,A')$ and returns a second quantum message $\tau_{A'}$. \
Let the resulting state be $\tau_{A'B}:=\P\otimes\Id_B(\sigma_{AB})$.
  Then the verifier continues making $T_2$ queries on input $\tau_{A'B}$ and outputs a decision bit. \
  A $T$-query protocol is defined to be one making $T=T_1+T_2$ queries. \
The protocol is said to solve the problem if the verifier outputs $b$ given $\D_b$ with probability at least $2/3$. \

Our goal is to determine $b$, given access to $M$ samples $\Delta_1,\ldots,\Delta_M\sim p_b$, encoded as the oracle $\O$. \
To show that no protocol can solve the problem, we start with any verifier which outputs 1 with probability at least $2/3$ for an yes instance. \ 
We prove that when given a no instance, there is a prover who convinces the verifier with probability at least $2/3-\eta$ for negligible $\eta$. \
More concretely, we make random modifications from an yes instance to a no instance, in the following steps:
\begin{enumerate}
\item Given $f_1,\ldots,f_M$, for $i\in[M]$, compute $\Delta_i=\Delta(f_i)$ and $S_i=\{x: f_i(x)=-1\}$.
  Also let $\bar S_i=\bit^n\backslash S_i$.
\item Sample $\Delta_1',\ldots,\Delta_M'\sim p_0$. 
  For $i\in[M]$,
    if $\Delta_i'\geq\Delta_i$, choose a random subset $R$ of $S_i$ such that $|R|=\Delta_i'-\Delta_i$. Set $g_i(x)=f_i(x)\cdot (-1)^{\Id_{R}(x)}$, where $\Id_R$ is the indicator function for set $R\subseteq\bit^n$; otherwise, 
    choose a random subset $R$ of $\bar S_i$ such that $|R|=\Delta_i-\Delta_i'$.
    Set $g_i(x)=f_i(x)\cdot (-1)^{\Id_{R}(x)}$.
  \item Output $g_1,\ldots,g_M$.
\end{enumerate}
We then prove that the modification yields a no instance.
\begin{lemma}\label{lem:llqsv-modification}
  The following two processes are equivalent, i.e., their output distributions are identical:
  \begin{enumerate}
  \item Output $g_1,\ldots,g_M\sim\F_n$.
  \item Sample $f_1,\ldots,f_M\sim\G_n$ and perform the modification in the above. Output $g_1,\ldots,g_M$.
  \end{enumerate}
\end{lemma}
\begin{proof}
  It suffices to show that for every $\Delta_1,\ldots,\Delta_M$, performing the modification on the product distribution $\H_{\Delta_1}\times\ldots\times\H_{\Delta_M}$ yields the distribution $\F_n$ since $\G_n$ is a probabilistic mixture of $\H_{\Delta_1}\times\ldots\times\H_{\Delta_M}$ for $\Delta_1,\ldots,\Delta_M\sim p_0$.
  For each $i\in[M]$, recall that $f_i\sim\H_{\Delta_i}$ is a random subset of size $\Delta_i N$. \
  If $\Delta_i'\geq\Delta_i$, removing $\Delta_i'-\Delta_i$ elements from $S_i$ yields a random subset of size $\Delta_i'$.
  The other case follows similarly.
\end{proof}

We provide intuition on why there exists a cheating prover for every $o((N/n)^{1/4})$-query $\QIP[2]$ protocol. \
By \lem{llqsv-concentration}, with probability $O(1/N)$ over choices of functions sampled from $\G_n$, $|\Delta_i-1/2|\leq O((n/N)^{1/2})$ for every $i\in[M]$. \
Again by \lem{llqsv-concentration}, with probability $O(1/N^{1/2})$, sampling $\Delta_1',\ldots,\Delta_M'\sim p_0$, $|\Delta_i'-1/2|\leq O((n/N)^{1/2})$. \
This implies that randomly flipping $\sum_{i=1}^M |\Delta_i'-\Delta_i|N=O(\sqrt{nN}M)$ elements from a set of at least $NM(1/2-o(1))$ elements yields a no instance.
By BBBV's hybrid argument \cite{BBBV97}, solving $\BFBD$ must make $O((N/n)^{1/4})$ queries.

\begin{theorem}\label{thm:llqsv-bfbd-qip2}
  Every $\QIP[2]$ protocol solving $\BFBD$ must make $\Omega((N/n)^{1/4})$ queries.
\end{theorem}
\begin{proof}
  Let $\V$ be any verifier which outputs 1 for an yes instance with proabability $p\geq 2/3$. \
  Starting from an yes instance $f_1,\ldots,f_M\sim\G_n$, by \lem{llqsv-modification}, we perform the random modification to yield a function $g_1,\ldots,g_M$. \
  Let $\Delta_i=\Delta(f_i)$ and $\Delta_i'=\Delta(g_i)$. \
  By \lem{llqsv-concentration}, with probability $O(1/N)$, $|\Delta_i-1/2|$ and $|\Delta_i'-1/2|$ are bounded by $O((n/N)^{1/2})$. \
  Conditioned on this event happening, $|\Delta_i-\Delta_i'|\leq O((n/N)^{1/2})$. \

  For every functions $F=(f_1,\ldots,f_M)$, let the prover be a quantum channel $\P_{F}\in\CPTP(A,A')$. \
  Given access to $\O=\O(F)$ which encodes $F$, let $\V_{\O}^{(1)}$ and $\V_{\O}^{(2)}$ be the unitary channels performed by the verifier in the first step and the second step of the protocol respectively. \
  For $c\in\{1,2\}$, we can without loss of generality write 
  \begin{align}\label{eq:llqsv-verifier}
    V^{(c)}_{\O} = U_{T_c}^{(c)}\O\ldots U_1^{(c)}\O U_{0}^{{(c)}}, 
  \end{align}
  and $\V_\O(\rho)=V_\O\rho V_\O^\dag$. \
  The cheating prover proves to the verifier that it is given oracle access to $f_1,\ldots,f_M\sim\G_n$, but the verifier is given access to $\O(G)$ encoding $G=(g_1,\ldots,g_M)$ obtained from the above modification. \
  The overall quantum channel $\T_{F}$ is defined $\V_{\O(G)}^{(2)}\circ \P_{F}\circ\V_{\O(G)}^{(1)}$ under a random modification. \
  Without loss of generality, by introducing a purifying system $E$ inaccessible to the verifier, we can assume that $\P_{F}\in\CPTP(AE,A'E')$ is a unitary channel for some system $E'$ also inaccessible to the verifier.

  Our goal is to show $\V_{\O(F)}^{(2)}\circ\P_F\circ\V_{\O(F)}^{(1)}$ is close to $\T_F$. \
  By triangle inequality,
  \begin{align}\nonumber
    \eta &:= \Big\|\Exp_G\V_{\O(G)}^{(2)}\circ\P_F\circ\V_{\O(G)}^{(1)}(\proj{0}) - \V_{\O(F)}^{(2)}\circ\P_F\circ\V_{\O(F)}^{(1)}(\proj{0})\Big\|_{\tr} \\\nonumber
    &\leq 
    \Exp_G\Big\|\V_{\O(G)}^{(2)}\circ\P_F\circ\V_{\O(G)}^{(1)}(\proj{0}) - \V_{\O(G)}^{(2)}\circ\P_F\circ\V_{\O(F)}^{(1)}(\proj{0})\Big\|_{\tr} \\\nonumber
    &\qquad+ \Exp_G\Big\|\V_{\O(G)}^{(2)}\circ\P_F\circ\V_{\O(F)}^{(1)}(\proj{0}) - \V_{\O(F)}^{(2)}\circ\P_F\circ\V_{\O(F)}^{(1)}(\proj{0})\Big\|_{\tr} \\
    &\leq 
    \Exp_G\Big\|\V_{\O(G)}^{(1)}(\proj{0}) - \V_{\O(F)}^{(1)}(\proj{0})\Big\|_{\tr} + \Exp_G\Big\|\V_{\O(G)}^{(2)}(\rho_F) - \V_{\O(F)}^{(2)}(\rho_F)\Big\|_{\tr},
  \end{align}
  where $\rho_F:=\P_F\circ\V_{\O(F)}^{(1)}(\proj{0})$ is a quantum state independent of the random modification. \
  Thus it suffices to give an upper bound on the second term, and an upper bound for the first follows using a similar argument since the zero state is also a state independent of the random modification. \
  
  Let $\ket{\psi_F}$ be a purification of any quantum state $\rho_F$ independent of the random modification. \ 
  For $c\in\{1,2\}$, consider the sequence of states:
  \begin{align}\nonumber
    \ket{\phi_{T_c}^G} &= U_{T_2}^{(c)} \O(F) U_{T_2-1}^{(c)} \O(F)\ldots U_1^{(c)} \O(F) U_0^{(c)}\ket{\psi_F}, \\\nonumber
    \ket{\phi_{T_c-1}^G} &= U_{T_2}^{(c)} \O(G) U_{T_2-1}^{(c)} \O(F)\ldots \O(F) U_1^{(c)}\O(F) U_0^{(c)}\ket{\psi_F}, \\\nonumber
    \vdots \\\nonumber
    \ket{\phi_{1}^G} &= U_{T_2}^{(c)} \O(G) U_{T_2-1}^{(c)} \O(G)\ldots \O(G) U_1^{(c)} \O(F) U_0^{(c)}\ket{\psi_F}, \\
    \ket{\phi_{0}^G} &= U_{T_2}^{(c)} \O(G) U_{T_2-1}^{(c)} \O(G)\ldots \O(G) U_1^{(c)} \O(G) U_0^{(c)}\ket{\psi_F}.
  \end{align}
  The Euclidean distance between the first and the last hybrids can be bounded by bounding the distance between two adjacent hybrids:
  \begin{align}\label{eq:llqsv-bbbv-euclidean}
    \|\ket{\phi_{T_1}^{G}}-\ket{\phi_0^G}\|
    \leq 2\sum_{i=0}^{T-1} \|\Id_W\ket{\psi_t}\|,
  \end{align}
  where $\ket{\psi_t}:=U_t\O(F)U_{t-1}\ldots U_1\O(F)U_0\ket{\psi_F}$ is the intermediate state after $t$ queries and the set $W:=\{(i,x):g_i(x)\neq f_i(x)\}$. \
  Note that $\ket{\psi^t}$ is independent of $W$.
  Recall that we fix $F$ and produce $G$ by flipping random elements, by \eq{llqsv-bbbv-euclidean}, taking the expectation over $W$ yields
  \begin{align}\nonumber
    \Exp_W\|\V^{(2)}_{\O(G)}(\rho_F)-\V^{(1)}_{\O(F)}(\rho_F)\|_{\tr} 
    &\leq 2\Exp_W\|\ket{\phi_{T_1}^G}-\ket{\phi_0^G}\| \\
    &\leq 4T_1\|Q\|_{\op}^{1/2},
  \end{align}
  where $Q:=\Exp_W\Id_W$.
  Since the modification picks a random subset of size $|W|=\sum_{i=1}^M|\Delta_i-\Delta_i'|N=O(M(nN)^{1/2})$ on a set of size at least $M(N/2-O((nN)^{1/2}))$, $\|Q\|\leq O((n/N)^{1/2})$.
  Thus $\eta\leq O(T(n/N)^{1/4})$. \
  Thus the verifier accepts the no case with probability at least $p-\eta-O(1/N)>1/3$.
\end{proof}

The statement in \thm{llqsv-bfbd-qip2} can be further strengthened to $\QIP=\QIP[3]$ using a similar hybrid argument, but since we will not need this result, we do not pursue this further. \
By the equivalence of $\BFBD$ and $\LLQSV$ shown in \lem{llqsv-reduction} and \thm{llqsv-bfbd-qip2}, we immediately have the following corollary.
\begin{corollary}\label{cor:llqsv-qma}
  There exists a constant $c>0$ such that $\LLQSV\notin\QIP[2](c(N/n)^{1/4})$ relative to a random oracle.
\end{corollary}

\subsubsection{Strengthen the hardness}

The final step is to ``lift'' from a lower bound against a uniform complexity class like $\mathsf{QIP}[2]$ or its subclass $\mathsf{QMA}$, to a lower bound against the nonuniform complexity class $\mathsf{QCAM/qpoly}$ (or its generalization to use more time and more advice). \ 

As a first observation, $\mathsf{QCAM} = \mathsf{BP\cdot QCMA}$, and the $\mathsf{BP\cdot}$ operator can be simulated by polynomial-size classical advice using Adleman's trick. \ Thus $\mathsf{QCAM/poly} = \mathsf{QCMA/poly}$ and $\mathsf{QCAM/qpoly} = \mathsf{QCMA/qpoly}$. \ As a second observation, Aaronson and Drucker \cite{AD10} proved that $\mathsf{QCMA/qpoly} \subseteq \mathsf{QMA/poly}$. \ Furthermore, all of these results relativize.

We apply the strong direct product theorem (SDPT) by Sherstov \cite{She11} to lift the hardness to $\QMA/\cpoly$. \ 
Recall that Sherstov showed that the query complexity obeys SDPT when the query lower bound of a single instance is proved using the \emph{polynomial method} \cite{BBC+01}. \ 

In the first step, we give an oracle separation $\coNP\not\subset\QMA/\cpoly$, and the proof for $\LLQSV\notin\QMA/\cpoly$ will basically follow the same ideas with minor modifications.

\paragraph{Warmup: $\coNP\not\subset\QMA/\cpoly$.}

To give an oracle separation, the oracle $F$ encodes $N=2^n$ instances of the problem. \ 
The oracle $F$ contains $N$ sections, and each section $F_x\in\bit^N$ is indexed by an $n$-bit string $x$. \
Let an yes instance be the all-one function, i.e., $f_0(x)=1$ for every $x\in\bit^n$. \ 
The no instance is obtained by flipping $\epsilon N$ points for $\epsilon\geq\eta$ and $\eta=100/N$. \  
Each section $F_x$ is either the yes case or the no case, with probability 1/2. \ %
The task is on input $x$, determine in which case $F_x$ is. \ %
The problem is clearly in $\coNP$, and therefore the proof of the following theorem is devoted to showing that the problem is not in $\QMA/\cpoly$. \

\begin{theorem}\label{thm:llqsv-conp-qmapoly}
  There exists an oracle relative to which $\coNP\not\subset\QMA/\cpoly$. \ 
\end{theorem}
\begin{proof}
  Let $\D_x$ denote equal mixture of the yes case and the no case for each $x\in\bit^n$. \ %
  For $x\in\bit^n$, let $F_{-x}$ denote all the sections but $F_x$ and similarly for $\D_{-x}$. \ 
  Let $f_0=1^N$ denote the yes instance and $\rho_{F,x}$ be the witness states given oracle access to $F$ and input $x$.
For a sufficiently small constant $c$, $\eta=100/N$, $T=o((1/\eta)^{1/2}/n)$ and $S=cN$, our idea is to show that a set of witness states $\{\rho_{x}:x\in\bit^n\}$ that convince any $\QMA(T)/S$ verifier $V$ with probability at least $2/3$ when $F_x=f_0$ can be used to convince the same verifier to accept a no instance for some $x\in\bit^n$ with probability more than $1/3$. \
Thus the $V$ does not solve the problem. \   

Suppose toward contradiction that it is not the case. 
Then $V$ accepts when $F_x=f_0$ with probability at least 2/3, i.e.,
\begin{align}
  \Pr_{F_x=f_0,F_{-x}\sim\D_{-x}}\left[ V^F(w_F,x,\rho_{x})=1\right] \geq 2/3,
\end{align}
where $w_F\in\bit^S$ is a classical advice that depends on $F$.
Furthermore, when the modification is applied to $F_x$, the verifier can detect the change with the aid of the advice. \  
More concretely, for $y\in\bit^n$, let $M_y$ denote a probabilistic algorithm applied $F$ to change $F_y=f_0$ to a no instance, i.e., $M_y(F)_{x}$ is a no case if $x=y$ and $M_y(F)_x=F_x$ otherwise. \ 
By the assumption we make, with the modification, for every $x\in\bit^n$,
\begin{align}
  \Pr_{F_x=f_0,F_{-x}\sim\D_{-x},M_x(F)=G}\left[V^{G}(w_G,x,\rho_{x})=0 \right]\geq 2/3. 
\end{align}
Note that the advice can be sensitive to the change. \ 
Since the witness state does not depend on the modification, we effectively remove the witness out of picture. \ 

Let $\D_1$ denote the distribution of the yes case and $\D_0$ denote the distribution of the no case. \  
Applying a witness-preserving amplification \cite{MW04}, let $V'$ be the $\QMA(c'Tn)/S$ protocol which solves the problem with probability at least $1-1/N^2$ for a sufficient large constant $c'$. \ 
Then we can construct an $\QMA(c'TnN)/S$ verifier $\tilde V$ which solves $N$ problems: The prover sends the witness state $\rho:=\bigotimes_{x\in\bit^n}\rho_{x}$ in one message, and the verifier runs $V'$ on every $x\in\bit^n$, $\rho_{x}$ and the given advice $w_F$, i.e., $\tilde V^F(w_F,\rho):=\bigotimes_{x\in\bit^n}(V')^{F}(w_F,x,\rho_{x})$. \  
By a union bound, with probability at least $1-1/N$, $\tilde V$ outputs the correct $N$-bit string, i.e., 
\begin{align}\label{eq:llqsv-advice}
  \Pr_{b\sim\U_N,F\sim\D_b}\left[\tilde V^{F}(w_{F},\rho)=b\right] \geq 1-1/N,
\end{align}
where $\U_N$ denote the uniform distribution over $N$-bit strings and $\D_b$ denote the distribution $\bigtimes_{x\in\bit^n}\D_{b_x}$. \ 

Now we replace the advice $w_F$ with a random string (independent of $F$), and by \eq{llqsv-advice}, this yields 
\begin{align}
  \Pr_{b\sim\U_n,F\sim\D_b,r\sim\U_S}\left[\tilde V^F(r,\rho)=b\right] \geq \frac{1-1/N}{2^S}\geq \frac{1}{2^{S+1}}. \ 
\end{align}
This effectively yields a $c'TnN$-query algorithm solving $N$ problems with probability at least $2^{-S-1}$. \ 
For each problem, let $p$ denote the real polynomial that approximates the AND function to error $1/3$. \
Thus the constraints we have are (i) $p(0)\geq 2/3$, and (ii) $p(\epsilon)\leq 1/3$ for every $1\geq\epsilon\geq\eta$ and $\epsilon N$ is an integer. \ 
By \lem{markov}, $\deg(p)\geq \Omega(\sqrt{1/\eta})$. \ 
Now by \thm{sherstov}, for a sufficiently small constant $c$ and $S=cN$, solving $N$ independent problems with probability $2^{-S-1}$ requires $2cN\sqrt{1/\eta}$ queries. \ 
Thus it must hold that $c'TnN\geq 2cN\sqrt{1/\eta}$ queries and $T=\Omega(\sqrt{1/\eta}/n)$. \ 
\end{proof}

\paragraph{$\LLQSV\notin\QMA/\cpoly$.}

The proof for $\LLQSV\notin\QMA/\cpoly$ basically follows the idea for $\coNP\not\subset\QMA/\cpoly$, with the following modifications. \ 
First, the distribution $\D_1$ does not describe a fixed function; instead, we replace it with the yes case in $\BFBD$ (\dfn{llqsv-bfbd}). \ 
For this, we basically apply an averaging argument. \
Second, by \lem{llqsv-concentration}, the no case of $\BFBD$ is obtained by modifying the yes case on a random subset of size $O(\sqrt{n/N})$ with probability at least $1-1/N$. \
The $\QMA/\cpoly$ verifier is no longer required to solve the problem for \emph{every} $\epsilon\geq\eta$, which is essential for a lower bound using the polynomial method. \ 
We show that via a reduction, we can get sufficiently many constraints to get a tight lower bound. \ 

Here we describe the problem $\LLQSV$ relative to the following oracle: \ 
The oracle $F$ contains $N$ sections, each indexed by an $n$-bit string $x$. \ 
Each section $F_x$ is sampled uniformly from $n$-bit Boolean functions, and the sample $s_x$ is either sampled from $|\hat F_x|^2$, or a uniform $n$-bit string (independent of $F_x$).  
On input $x$, the protocol is challenged to determine $s_x$ is sampled from $|\hat F_x|^2$ in the yes case or the uniform distribution in the no case. 

\begin{theorem}
  There exists a constant $c>0$ such that $\LLQSV\notin\QMA(c(N/n)^{1/4}/n)/(cN)$ relative to a random oracle. \ 
\end{theorem}
\begin{proof}
  The proof basically follows that of \thm{llqsv-conp-qmapoly}. \
  First by \lem{llqsv-reduction}, let $\D_1$ and $\D_0$ denote the distributions in the yes and the no cases of $\BFBD$ (\dfn{llqsv-bfbd}) respectively. \
  In this picture, each section of the oracle is either sampled from $\D_1$ or from $\D_0$. \ 
  Let $\D$ denote the equal mixture of $\D_1$ and $\D_0$. \ 

  The idea is to show that for $T=o((N/n)^{1/4}/n)$, any set of witness states $\{\rho_{F,x}:F\in\bit^{N\times N}, x\in\bit^n\}$ that convince any $\QMA(T)/S$ verifier $V$ with probability $2/3$ when $F_x\sim\D_1$ can be used to convince the verifier that $F_x\sim\D_0$ for some $x\in\bit^n$ with probability more than $1/3$. \
  Suppose toward contradiction that for every $x$, when $F_x\sim\D_0$, $V$ accepts with probability at most $1/3$.
  Let the modification denote $M_x$ to change any $F_x$ such that $M_x(F_x)$ is distributed according to $\D_0$, and other sections are unchanged. \ 
  By an averaging argument, for every $x\in\bit^n$, there exists a function $F_x$ and a random modification of size $r_x=O(\sqrt{nN})$ such that the verifier detects if the modification has been applied with probability at least $2/3-o(1)$. \ 
  Then by random self-reducibility, there exists a verifier which outputs the correct answer on every modification of size $r_x$ with probability at least $2/3-o(1)$ as well. \ 
  
  By a witness-preserving amplification, there is a $\QMA(O(Tn))/S$ verifier which detects if there is a random modification of size $r_x$ with probability at least $1-1/N^2$ on every $x$. \ 
  The assumption implies that there is a $\QMA(O(TNn))/S$ verifier $\tilde V$ which solves $N$ problems (i.e., $\tilde V$ determines if there is a modification of size $r_x$ for every $x\in\bit^n$) with probability at least $1-O(1/N)$ on the witness state $\rho_F:=\bigotimes_{x\in\bit^n}\rho_{F,x}$ independent of the modifications. \ 
  Now we replace the advice string with a random string, and yield a $O(TNn)$-query $\QMA$ verifier which is given the witness state $\rho_F$ and solves $N$ problems with probability at least $2^{-S-1}$ on the witness $\rho_F$. \ 
  Then we apply \thm{sherstov} to get $\deg(p)=O(Tn)$ for any real polynomial $p$ that approximate the function to compute to error $1/5$ for $S=c N$ for a sufficiently small constant $c$. \ 

  It remains to show that every $T'$-query algorithm detecting modifications on a random subset $R$ of size $r\leq O(\sqrt{nN})$ in a set $Q$ of size $q=\Omega(N)$ implies that a $T'$-query algorithm which detects a random modification $R'\subseteq Q'$ for $|R'|\in\{1,2,\ldots,|Q'|\}$ and $|Q'|=\Omega(q/r)$. \ 
Thus we can use the polynomial method to conclude that every $T'$-query algorithm solving the problem must satisfy $2T'=\deg(p)$, for any polynomial $p$ satisfying the following constraints: (i) $p(0)\geq 2/3$ and (ii) $p(i)\leq 1/3$ for $i\in\{1,2,\ldots,|Q'|\}$. \ 
By \lem{markov}, $\deg(p)=\Omega(\sqrt{|Q'|})=\Omega((N/n)^{1/4})$ and thus $T=\Omega(T'/n)=\Omega((N/n)^{1/4}/n)$. \ 

Without loss of generality, let $Q=[q]$ and $R$ be a random subset of size $r$. \ 
For every function $f_0:[q]\to\bit$, we define the problem $\P_{q,r}$ as follows: \
In the yes case, the algorithm is given access to $f=f_0$, and in the no case, the algorithm is given access to $f(x):=f_0(x)\oplus\Id[x\in R]$ for a random subset $R$ of size $r$. \ 
Thus from the previous paragraph, our goal is to show that an algorithm solving $\P_{q,r}$ implies an algorithm solving $\P_{\eta,\ell}$ for $\eta=\lfloor q/r\rfloor$ and every $\ell\in\{1,2,\ldots,\eta\}$ using the same number of queries. \ 
First we can stretch the parameters $\eta,\ell$ by a factor $\alpha\geq 1$ to yield an instance $g$ in $\P_{\alpha\eta,\alpha\ell}$: \ 
Define $\tilde g:[\alpha\eta]\to\bit$, $\tilde g(x):=f(((x-1)\bmod\eta)+1)$, and the function $g(x):=\tilde g\circ\pi(x)$ for a random permutation $\pi$ on $[\alpha\eta]$. \ 
Second, we can pad a function with $\beta$ elements to yield a reduction from $\P_{\eta,\ell}$ to $\P_{\eta+\beta,\ell}$: \ 
Let $\tilde g:[\eta+\beta]\to\bit$, $\tilde g(x)=f(x)$ for $x\in[q]$ and $\tilde g(x)=0$ for $x>\eta$, and the function $g:=\tilde g\circ\pi$ for a random permutation $\pi$ on $[\eta+\beta]$. \ 
For problem $\P_{\eta,\ell}$, first we apply a stretch to yield $\P_{\eta r/\ell,r}$, followed by a padding with $q-\eta r/\ell$ elements to yield $\P_{q,r}$. \ 
\end{proof}

\alert{
When combining powerful quantum complexity classes with \emph{quantum} advice, one needs to be even more careful. \ For example, Raz \cite{Raz09} showed that $\mathsf{QIP/qpoly}=\mathsf{ALL}$ ($\mathsf{ALL}$ being the class of all languages), and Aaronson \cite{Aar18b} likewise showed that $\mathsf{PDQP/qpoly}=\mathsf{ALL}$, where $\mathsf{PDQP}$ is the generalization of $\BQP$ to allow multiple non-disturbing measurements \cite{ABFL16}. \ 
Our proof based on a strong direct product theorem can be used to further strengthen our result to $\LLQSV\notin\QMA/\qpoly$, thereby giving a new proof for $\QMA/\qpoly\neq\ALL$ relative to a random oracle. \ 
Previously, Aaronson proved that $\QMA/\qpoly\subseteq\PSPACE/\cpoly$ \cite{Aar06}. \

We note that the proof technique does not allow us to prove the problem is not in $\QIP[2]/\qpoly=\ALL$, because there is no generic rewinding procedure to restore the quantum advice for another use. \ 
}

In any case, the above suffices to establish the following.

\begin{theorem}
  Relative to a random oracle, $\LLQSV\notin\QCAM(c^n)/\mathsf{q}(c^n)$, for some constant $1<c<2$.
\end{theorem}

\section{A Device Performing the Ideal Measurement}\label{sec:ideal}

In this section, we consider a device which does the following: the device may share an arbitrary entanglement $\psi_{DE}$ with Eve before receiving a Haar-random circuit. \
After receiving $C\sim\Haar(N)$, the device $\A$ performs the ideal measurement, i.e., it runs the circuit $C$ on system $D$ of $\rho_{DE}$, and performs a standard basis measurement to output strings. \
On the other hand, the eavesdropper Eve may learn information about the device's output by applying an arbitrary quantum operation on the second system $E$, but it has no direct access to $\A$'s system and output. \
In the following discussion, we will also call such a device $\A$ a semi-honest device. \

\subsection{A single-round analysis}\label{sec:ideal-single}

In general, the first system is a Hilbert space $D\cong\mathbb C^{2^d}$ for $d\geq n$. \
Without loss of generality, we can only consider the case where $D\cong\mathbb C^N$, i.e., $d=n$. \
The reason is that for $d\geq n$, we can decompose $D={D'}\otimes{D''}$ such that $\H_{D'}\cong\mathbb C^N$ is the system the circuit $C$ acts on, and let $E'=D''\otimes E$. \
In this case, it suffices to find a lower bound of $H(D'|E')_\rho$ since
$H(D|E)_\rho\geq H(D'|E')_\rho$ and $H_{\min}(D|E)_\rho\geq H_{\min}(D'|E')_\rho$, where $\rho=\A_D^C\otimes\Id_E(\psi)$ is the classical-quantum state after the ideal measurement is performed when $C$ is chosen by the verifier. \
Thus in the following analysis, we assume $D\cong\mathbb C^N$ and $\A^C$ is the unitary channel of $C$ on $D$. \

\subsubsection{The Holevo Information}
By linearity of quantum operations, without loss of generality, we can consider the special case where the device and the eavesdropper share a pure state in Schmidt decomposition \
\begin{align}
  \ket{\psi} = \sum_{i\in\bit^n} \sqrt{\lambda_i} \ket{\phi_i}_D\ket{\psi_i}_E,
\end{align}
where each $\lambda_i\geq 0$ is real for every $i$ and $\{\ket{\phi_i}\}$ and $\{\ket{\psi_i}\}$ are twe sets of orthonormal vectors. \
We also denote $\psi:=\proj{\psi}$. \
\begin{lemma}\label{lem:holevo-schmidt}
  For any algorithm $\A$, let $\sigma=\A^C(\psi)$ be the classical-quantum state obtained by performing a projective measurement depending on $C$ on the first system $D$. \
  Then the Holevo information of $\sigma$, denoted $\chi(D:E)_\sigma=H(\lambda)$, where $H(\lambda)$ is the Shannon entropy of the distribution $\lambda$. \
\end{lemma}
\begin{proof}
We consider the unitary $U:=\sum_i \ket{\phi_i}\!\bra{i}$ and let $U_{ij}=\bra{i}U\ket{j}$ be the element of matrix $U$. \
Also let $\bar U = \sum_{ji} U_{ji}\ket{\psi_i}\!\bra{\psi_j}$, $\rho=\sum_i\lambda\proj{i}$ and $\bar\rho=\sum_i\lambda_i\proj{\psi_i}$. \
Since $\{\ket{\psi_i}\}$ is a set of orthonormal vectors, $\bar U$ is also a unitary acts on the space $\Span\{\ket{\psi_i}\}$. \
The entanglement $\ket{\psi}$ can be written in the standard basis: \
\begin{align}\nonumber
  \ket{\psi}
  &= \sum_i \sqrt{\lambda_i} U\ket{i}\ket{\psi_i} \\\nonumber
  &= \sum_{ij} \sqrt{\lambda_i} U_{ji}\ket{j}\ket{\psi_i} \\\nonumber
  &= \sum_{ij} \sqrt{\lambda_i} \ket{j} \proj{\psi_i}\bar U\ket{\psi_j} \\
  &= \sum_j \ket{j}_D\bar\rho^{1/2}\bar U\ket{\psi_j}_E.
\end{align}
We define the subnormalized pure state
\begin{align}
  \sigma_i := \bar\rho^{1/2}\bar U\proj{\psi_i} \bar U^\dag \bar\rho^{1/2},
\end{align}
and compute the Holevo information: 
\begin{align}\label{eq:holevo}\nonumber
  \chi(D:E)_{\A^C(\psi)}
  &= H(E)_{\sum_i\sigma_i} - \sum_i \alpha_i H(E)_{\sigma_i/\alpha_i} \\
  &= H(\lambda).
\end{align}
where $H(\lambda)$ is the Shannon entropy of $\lambda$ and $\alpha_i:=\tr(\sigma_i)$. \
The second equality in \eq{holevo} holds since 
\begin{align}
  \sum_i \sigma_i = \bar\rho,
\end{align}
and $H(E)_{\bar\rho}=H(\lambda)$. \
Moreover, since $\frac{1}{\alpha_i}\sigma_i$ is a normalized pure state, $H(E)_{\sigma_i/\alpha_i}=0$. \
Note that $\sum_i\sigma_i$ is independent of $\bar U$, so the Holevo information does not change if a different measurement is performed on the first system $D$. \
Therefore for every $C$, the Holevo information is $H(\lambda)$. \
\end{proof}

Next, we show that if the device solves $\frac{(2-\epsilon)N}{N+1}$-XHOG, then the Holevo information is at most $\epsilon n+1$. \
This observation can be used to estabilish a lower bound of the conditional von Neumann entropy. \

\subsubsection{High XHOG Score Implies Small Holevo Information}

Previously, we have shown the Holevo information can be explicitly computed if the Schmidt decomposition is known. \
In this section, we establish the connection between the Holevo information and the score of any semi-honest device. \

First, we introduce a useful technical lemma.
\begin{lemma}\label{lem:concentrated-entropy}
  Let $p$ be a distribution $p$ over finite set $\X$ of $N$ elements such that there exists $x\in\X$, $p(x)\geq 1-\epsilon$. \
Then the Shannon entropy $H(p)\leq \epsilon n+1$. %
\end{lemma}
\begin{proof}
  Let $p(x)=1-\gamma$ for $\gamma\leq\epsilon$ and $q$ be the conditional distribution on the event that $y\neq x$ i.e., $q(y)=p(x)/\gamma$ for $y\neq x$. \
  By definition, 
  \begin{align}\nonumber
    H(p)
    &=-p(x)\log p(x) - \sum_{y\neq x}p(y)\log p(y) \\\nonumber
    &= -(1-\gamma)\log(1-\gamma) - \gamma\sum_{y\neq x} q(y) (\log q(y)+\log\gamma) \\\nonumber
    &= -(1-\gamma)\log(1-\gamma)-\gamma\log\gamma  - \gamma\sum_{y\neq x}q(y)\log q(y) \\
    &= h(\gamma) + \gamma H(q),
    \end{align}
    where $h$ is the binary entropy function satisfying $h(\gamma)\leq 1$ for $\gamma\in[0,1]$. \
    Since $q$ is a distribution over $N-1$ elements, it Shannon entropy is upper bounded by $\log(N-1)\leq n$. \
\end{proof}

By \lem{concentrated-entropy}, to prove the Holevo information is small, it suffices to show that when the XHOG score is high, the distribution $\lambda$ must be concentrated on a single point.

\begin{lemma}\label{lem:holevo-semi}
  The Holevo information of the classical-quantum state obtained from any semi-honest device solves $\frac{(2-\epsilon)N}{N+1}$-$\XHOG$ is at most $\epsilon n+1$.
\end{lemma}
\begin{proof}
We analyze the XHOG score of a device on input a standard basis vector $\ket{i}$:
\begin{align}
  S_i := \Exp_{C\sim\Haar(N)}\left[\sum_z q_{C,i}(z)p_C(z)\right].
\end{align}
where $q_{C,i}(z)=|\bra{z}C\ket{i}|^2$. \
For $i=0$, $q_{C,0}=p_C$ the score is $\frac{2}{N+1}$. \
For $i\neq 0$, the distribution of $p_C(z)$ and $q_{C,i}(z)$ can be seen as two distinct elements of $P\sim\Dir(1^N)$. \
These facts imply
\begin{align} \nonumber
  S_i 
  &= N\cdot \Exp_{P\sim\Dir(1^N)}[P_0P_i] \\
  &= \frac{1+\delta_{i0}}{N+1}.
\end{align}
By definition, the XHOG score using $\ket{\psi}$ is \
\begin{align}\nonumber
  S &= \Exp_C\left[\sum_z p_C(z)\bra{\psi}(C^\dag \proj{z} C)\otimes\Id\ket{\psi}\right] \\
  &= \sum_i S_i \cdot \bra{\psi_i}\bar U^\dag \bar\rho \bar U\ket{\psi_i}.
\end{align}
Let $\bar\tau:=\bar U^\dag\bar\rho\bar U$ and $\bar\tau'$ be the state obtained by measuring $\bar\tau$ in basis $\{\ket{\psi_i}\}$. \
If $S\geq \frac{2-\epsilon}{N+1}$, \
\begin{align}\label{eq:score-tau}\nonumber
  \frac{2-\epsilon}{N+1}
  &\leq \sum_i S_i\bar\tau_{ii}' \\
  &= \frac{1+\bar\tau'_{00}}{N+1}.
\end{align}
This implies that $\bar\tau_{00}'\geq 1-\epsilon$. \
By \lem{holevo-schmidt}, the Holevo information can be upper bounded: \
\begin{align}
  \chi(D:E)_{\A^C(\rho)} 
  = H(E)_{\bar\tau} \leq H(E)_{\bar\tau'} \leq \epsilon n + 1. 
\end{align}
Since $\bar\tau'$ is obtained by performing a projective measurement $\{\ket{\psi_i}\}$ on $\tau$, by \lem{vn-data}, $H(E)_{\bar\tau}\leq H(E)_{\bar\tau'}$. \
The second inequality holds by \lem{concentrated-entropy}. \
Moreover, the inequality saturates when $\bar U$ is the identity matrix, i.e., $\{\ket{\psi_i}\}$ is the standard basis. \
\end{proof}
We note that the bound in \lem{holevo-semi} is nearly optimal: Consider the case where the entanglement shared by the device and Eve be the quantum state
\begin{align}
  \ket{\psi}_{DE} = (1-\epsilon)^{1/2} \ket{0,0} + \left(\frac{\epsilon}{N-1}\right)^{1/2}\sum_{x\neq 0} \ket{x,x}. \ 
\end{align}
The the Holevo information by \lem{holevo-schmidt} is $\epsilon\log(N-1)+h(\epsilon)$. \

\subsubsection{A Single-Round Analysis}

Now we show small Holevo information implies large conditional von Neumann entropy. \
We apply the idea from \cite{ADFRV18} and \cite{Arn18}. \
Let $D$ and $E$ denote the device's output random variable
and Eve's quantum register, repectively. \
By definition, the Holevo information
\begin{align}
  \chi(D : E)_{\A^C(\rho)}
  &= H(E)_{\A^C(\rho)} - H(E|D)_{\A^C(\rho)},
\end{align}
where $H$ denotes the von Neumann entropy. \
The following theorem shows a linear lower bound of the von Neumann entropy when the device solves $b$-XHOG for $b\approx(2-\epsilon)$. \

\begin{theorem}\label{thm:von-neumann}
  For any device $\A^C$ that given access to $C$ and performs the ideal measurement on the first system $D$ of any bipartite quantum state $\rho_{DE}$ 
and solves $b$-XHOG for $b\geq\frac{(2-\epsilon)N}{N+1}$, %
  \begin{align}
    \Pr_{C\sim\Haar(N)}\left[H(D|E)_{\A^C(\rho)}\geq (0.99-\epsilon)n\right] \geq 1-O\left(\frac{1}{N^{0.02}}\right).
  \end{align}
\end{theorem}
\begin{proof}
Let $\A^C$ be any device that performs an ideal measurement to yield $\psi^C_{DE}:=\A^C(\rho)$ and outputs a random variable $D$. \
We observe that
\begin{align}\nonumber
  H(D|E)_{\psi^C}
  &= H(DE)_{\psi^C} - H(E)_{\psi^C} \\\nonumber
  &= H(D)_{\psi^C} + H(E|D)_{\psi^C} - H(E)_{\psi^C} \\\label{eq:34}
  &= H(D)_{\psi^C} - \chi(D: E)_{\psi^C}.
\end{align}
The first term in the third line of \eq{34}, i.e., $H(D)_{\psi^C}$, denotes the von Neumann entropy of the device's output, when Eve's system $E$ is empty. \
To lower bound the first term in \eq{34}, recall that $q_{C,i}(z)=|\bra{z}C\ket{i}|^2=p_{CX^i}(z)$. \
From \eq{collision-concentration} and right translational invariance of the Haar measure, for every $i$, \
\begin{align}\nonumber
    \Pr_{C\sim\Haar(N)}\left[H(p_{CX^i}) \geq 0.99n+1\right]
    &\geq 
    \Pr_C\left[ S_C \leq N^{-0.99}/2\right] \\
    &\geq
    1-O\left(\frac{1}{N^{1.02}}\right).
\end{align}
By the union bound,
\begin{align}
    \Pr_{C\sim\Haar(N)}\left[\min_{i\in\bit^n} H(p_{CX^i})\geq 0.99n+1\right]
    \geq 1-O\left(\frac{1}{N^{0.02}}\right).
\end{align}
Now by the convexity of Shannon entropy and the definition that $q_C(z)=\Exp_{i\sim\bar\tau'}[q_{C,i}(z)]$, 
\begin{align}\nonumber
    \Pr_{C\sim\Haar(N)}\left[
    H(q_C) \geq 0.99n+1
    \right]
    &\geq 
    \Pr_{C\sim\Haar(N)}\left[\Exp_{i\sim\bar\tau'} \left[H(p_{CX^i})\right]\geq 0.99n+1\right] \\\nonumber
    &\geq 
    \Pr_{C\sim\Haar(N)}\left[\min_i H(p_{CX^i})\geq 0.99n+1\right] \\
    &\geq 
    1-O\left(\frac{1}{N^{0.02}}\right).
\end{align}
Thus since $H(D)_{\psi^C}=H(q_C)$,
\begin{align}\label{eq:von-neumann-c}
    \Pr_{C\sim\Haar(N)}\left[ H(D)_{\psi^C} \geq 0.99n+1\right] \geq 1-O\left(\frac{1}{N^{0.02}}\right).
\end{align}
In \lem{holevo-semi}, we have shown that for every $C$, $\chi(D:E)_{\psi^C}\leq \epsilon n+1$. \
Then by \eq{von-neumann-c}, 
\begin{align}
    \Pr_{C\sim\Haar(N)}\left[H(D|E)_{\psi^C}\geq (0.99-\epsilon)n\right] \geq 1-O\left(\frac{1}{N^{0.02}}\right).
\end{align}
\end{proof}

We denote the output state $\psi = \Exp_{C\sim\Haar(N)}\left[ \proj{C} \otimes \psi^C \right]$,
and 
\begin{align}
  H(D|CE)_\psi := \Exp_{C\sim\Haar(N)}\left[ H(D|E)_{\psi^C} \right].
\end{align}
The following corollary concludes our single-round analysis for this ideal setting.
\begin{corollary}\label{cor:ideal-vn}
  Let $\A$ be any device performing an ideal measurement solving $(1+\delta)$-$\XHOG$ and $\psi$ be its output.
  Then $H(D|CE)_{\psi} \geq (\delta-0.01) n -o(1)$.
\end{corollary}
\begin{proof}
  By definition and \thm{von-neumann}, $H(D|CE)_\psi \geq (\delta-0.01) n - O(N^{-0.02})$.
\end{proof}

\subsection{Entropy Accumulation}\label{sec:ideal-ea}

\begin{figure}
    \hrule\vspace{.5em}
  Input: security parameter $n$, the number of rounds $m$, the score parameter $\delta\in[0,1]$ and the fraction of test rounds $\gamma$.\\

  The protocol:

  \vspace{.5em}
  \begin{enumerate}
    \item The verifier samples $C\sim\Haar(N)$. 
      For $i=1,\ldots,m$, run the following steps:
  \begin{enumerate}
  \item The verifier sends $C$ to the device. (This step may be omitted.)
  \item The device returns a sample $z_i$.
  \end{enumerate} 
  
  \item Let $t=|\{i:T_i=1\}|$ be the number of test rounds. The verifier computes 
    \begin{align}
      s = \frac{1}{t}\sum_{i:T_i=1} p_{C}(z_i).
      \end{align}
      If $s\geq (1+\delta)/N$, then the verifier accepts and outputs $(z_1,\ldots,z_m)$ to the quantum-proof randomness extractor.
  \end{enumerate} 
  \hrule\vspace{1em}
  \caption{The entropy accumulation protocol for a device performing the ideal measurement.
}
  \label{fig:ideal}
  \end{figure}

We present our entropy accumulation protocol for any device which performs the ideal measurement in \fig{ideal}. \
    Note that since the von Neumann entropy lower bound holds for almost every $C$, the verifier may reuse the circuit in every round. \
\begin{theorem}
  Let $\A_1,\ldots,\A_m$ be $n^{O(1)}$-query sequential processes given access to the first system of a bipartite state $\rho_{DE}$ and outputting $z_1,\ldots,z_m$ solving $\LXEB_{1+\delta,m}$ with probability $p$. \
  Then with probability $1-2^{-\Omega(n)}$ over the choices of $C$,
  \begin{align}
    H_{\min}^{\epsilon}(Z^m|E)_{\A_m\circ\ldots\A_1(\rho)|\Omega} \geq n\left( (\delta-0.01) m - c\sqrt{m} -o(1)\right)
  \end{align}
  where $\Omega$ denotes the event that the output $z_1,\ldots,z_m$ solves $\LXEB_{1+\delta,m}$, i.e., \
\begin{align}
  \Omega :=\left\{ (z_1,\ldots,z_m): \frac{1}{m}\sum_{i=1}^m P(z) \geq \frac{1+\delta}{N} \right\}. 
\end{align}
The parameter $c:=4.01(1+\log(\frac{1}{p\epsilon}))^{1/2}$.
\end{theorem}
\begin{proof}
We apply the EAT shown in \sec{eat}.
In particular, we choose $G:=N\sum_z \proj{z} P(z)$.
By \cor{ideal-vn}, we choose $f(\delta)=(\delta-0.01) n-o(1)$.
This gives $\|\nabla f\|_\infty=n$.
Then we have $V = 2(\log(2d_Z+1)+\|\nabla f\|_\infty) = 2(\log(2N+1)+n) = 4n+O(1/N)\leq 4.01n$.
Now we have 
\begin{align}
  H_{\min}^\epsilon(Z^m|E)_{\A_m\circ\ldots\circ\A_1(\rho)|\Omega}
  &\geq n \left( (\delta-0.01) m - o\left(1/n\right) -4.01\sqrt{m}\sqrt{\log\frac{2}{\Pr_\sigma[\Omega]^2\epsilon^2}}\right).
\end{align}
\end{proof}

As a side note, the same analysis in \sec{ideal-single} and \sec{ideal-ea} also applies to a Fourier sampling circuit, i.e., $C$ describes the unitary transformations in \eq{fourier} under a random $n$-bit Boolean functions $f$. \ 
In this case, an ideal devices solves $b$-$\XHOG$ for $b\approx 3$ \cite{Kre21}, and any device solving $(1+2\delta)$-$\XHOG$ has a lower bound $\Omega(\delta n)$ on the conditional von Neumann entropy. \
Applying the entropy accumulation theorem, an $m$-round protocol also accumulates $\Omega(\delta m n)$ conditional min-entropy.

The analysis in \sec{ideal} is meant to provide an intuition on why LXEB can be used to generate randomness in a simplified setting. \
In particular, for entropy accumulation, the channels $\A_1,\ldots,\A_m$ are the ideal unitary channel of $C$, so this setting is far from full device-independence. \
In the next section, we provide an analysis for a fully general device where the channels $\A_1,\ldots,\A_m$ are not necessarily ideal. \

\section{A Fully General Device}\label{sec:general}

In this section, we give an analysis for a general device. \
Before we prove our main result, first we provide some intuition. \
Let $\rho_{DE}$ be an entangled state shared between the device and Eve, and $\A^C(\rho_{DE})$ be the quantum algorithm that the device performs given access to the first subsystem $D$ and the circuit $C$. \
Without loss of generality, $\A^C(\rho_{DE})$ outputs a classical-quantum state in the form $\sum_z p(z)\proj{z}_{D'}\otimes \xi_E(z)$ (this also implies that $\A^C$ is a quantum channel of type $D\to D'$ for Hilbert spaces $D$ and $D'$) for normalized quantum states $\xi_E(z)$. \
The classical part, i.e., the random variable $z$ is sent to the verifier, and the quantum information $\xi_E(z)$ can be used as the input to the next round of interaction. \
In an $m$-round protocol, the device and the verifier repeat the message exchanges $m$ times, and the verifier outputs a decision bit indicating accept or abort.
If the verifier accepts, then the output of the device is fed into a quantum-proof randomness extractor. \

We consider any device given oracle access to $C$, and show that if the device passes $\LXEB_{b,k}$, the output must has a min-entropy lower bound. \
First, applying the rotational invariance of a Haar random unitary, we formally define an input model in which the device is given access to: the verifier samples a circuit $C\sim\Haar(N)$ and two unitaries $V,W$. \
The first unitary $V$ is a random phase unitary of the form $\sum_z e^{i\theta_z}\proj{z}$, where each $\theta_z$ is independently sampled from the uniform distribution over $[0,2\pi)$.
The second unitary $W$ is a Haar random unitary on the subspace $\{\ket{x}:x\in\bit^n\backslash\{0^n\}\}$.
The device is then given oracle access to $C'=VCW$, and by definition, it is clear that $p_C=p_{C'}$, where $p_C(z):=|\bra{z}C\ket{0}|^2$ is the density of $p_C$. \
Also, $C$ and $C'$ are identically distributed (according to the Haar measure). \
We will say an efficient device is given oracle access to $C$ if it can apply $C'$ on any subset of its system $D$ of the same size. \

The symmetrization allows us to analyze a von Neumann entropy lower bound. \
In particular, we show that the conditional von Neumann entropy on $C$ is equal to the von Neumann entropy on $p_C$ in this query model. \
Since $C$ is distributed over an infinite set, we define the von Neumann entropy $H(Z|CE)_\psi:= \Exp_C\left[H(Z|E)_{\psi^C}\right]$, where $\psi^C$ is the classical-quantum state output by the device $\A$ when $\A$ is given access to $C$. \
In the following discussion, we use capitalized $P_C$ to emphasize that the distribution is a vector of complex-valued random variables over the probability simplex. \
In particular, for $C\sim\Haar(N)$, $P_C$ is distributed according the the Dirichlet distribution $\Dir(1^N)$ (see \sec{haar-intro} for details). \
\begin{theorem}\label{thm:symmetry}
  For Haar random $C$ and every device $\A$ which on input the first system of a bipartite state $\rho_{DE}$ and outputs a classical-quantum state $\psi$ given oracle access to $C$.
  Then $H(Z|CE)_\psi=H(Z|P_CE)_{\psi}$, where $P_C(z):=|\bra{z}C\ket{0^n}|^2$ is distributed according to the Dirichlet distribution $\Dir(1^N)$.
\end{theorem}
\begin{proof}
  For every distribution $P$, we let $S_P$ be the set of unitaries such that $P_C=P$.
  For each $P$, let $C_P$ denote a representative in the set $S_P$. 
  Then for $C\in S_P$, there exists a diagonal unitary $V'$ and a unitary $W'$ on the subspace $\mathrm{span}\{\ket{x}:x\neq 0^n\}$ such that $C=V' C_P W'$.
  This implies that one can simulate a query to $C\in S_P$ using one query to $C_P$. 
  More precisely, for $C\in S_P$, we denote $\sigma^P = \psi^C=\psi^{C_P}$.
  This implies that
  $H(Z|E)_{\psi^C} = H(Z|E)_{\sigma^P}$.
  By the definition of conditional von Neumann entropy and the fact that $p_C$ is distributed according the distribution $\Dir(1^N)$,
  $H(Z|CE)_{\psi} = \Exp_{P\sim\Dir(1^N)}\left[ H(Z|E)_{\sigma^p} \right] = H(Z|PE)_\psi$.
\end{proof}

Our analysis proceeds in the following steps. \
First, we show that for every $T$-query $\A$ and $\rho_{DE}$, there is another quantum algorithm $\B$ which on input the first system $D$ of $\rho_{DE}$ and $k$ i.i.d. random variables $z_1,\ldots,z_k$ sampled according to $P_C$ for Haar-random $C$ such that 
\begin{align}
  \left\| \Exp_C\A^C(\rho_{DE}) - \Exp_{C,z_1,\ldots,z_k\sim P_C}\B(\rho_{DE},z_1,\ldots,z_k) \right\|_{\tr} \leq 2^{-\Omega(n)},
\end{align}
for $k=T^2\cdot 2^{-\Omega(n)}$.
Thus without loss of generality, we may consider the behavior of $\B$.

Next, we show that $\B$'s score can be easily calculated using the probability that the output $z\in\{z_1,\ldots,z_k\}$.
More specifically, we show that $\B$'s score is negligibly close to $\frac{1}{N+k}\cdot \left(1 + \Pr[z\in\{z_1,\ldots,z_k\}]\right)$. \
Thus, when $\B$'s score is $\frac{1+\delta}{N+k}$,
potentially the device and the eavesdropper can coordinate in a way such that the conditional von Neumann entropy using $\rho_{DE}$ with probability negligibly close to $\delta$;
otherwise, since the eavesdropper has no information about $z_1,\ldots,z_k$, and each $z_i$ is sampled from $P_C$, the conditional von Neumann entropy can be bounded by the min-entropy of $P_C$. \
Then since with overwhelming probability over $C$, $P_C$'s min-entropy is $n-\log n-O(1)$, $\A$'s output must have high conditional von Neumann entropy when it wins the prototol for a ``typical'' $C$. \

\subsection{Simulation of a Haar Random Unitary Given Sample Access}

In this section, we show that any $T$-query algorithm for $T=2^{O(n)}$ outputs a classical-quantum state whose conditional von Neumann entropy is $\delta n - o(n)$, provided it solves the $(1+\delta)$-XHOG problem. \
We show that one can simulate an oracle-access algorithm using a sample-access one. \
For this, we rely on the ideas from Ambainis, Rosmanis and Unruh \cite{ARU14} and from Kretschmer \cite{Kre21}. \

\subsubsection{From a Circuit Oracle to a State-Preparation Oracle}

For state $\ket{\psi}$, let $C^\psi$ be a random unitary such that $C^\psi\ket{0}=\ket{\psi}$ and $C^\psi$ is Haar random for the subspace orthogonal to $\ket{0}$, and $\O^\psi$ be a reflection about 
$\ket{\psi_\bot} = \frac{\ket{\psi}-\ket{\bot}}{\sqrt{2}}$, i.e., it sends $\ket{\psi}\mapsto\ket{\bot}$ and $\ket{\bot}\mapsto\ket{\psi}$ and acts trivially for states orthogonal to states in $\mathrm{span}\{\ket{\psi},\ket{\bot}\}$. \

The first step is to show one can approximate every $T$-query algorithm given access to $C^\psi$ using an $O(T)$-query algorithm given access to $\O^\psi$ for every $\ket{\psi}$. \
We prove the following theorem, improving the constant factor of \cite[Theorem~19]{Kre21}. \
\begin{theorem}[{cf. \cite[Theorem~19]{Kre21}}]\label{thm:kre21-thm19}
  For every quantum state $\ket{\psi}$, every $T$-query algorithm $\A^{C^\psi}$ can be approximated by a $(2T)$-query quantum algorithm $\B^{\O^\psi}$ such that 
  \begin{align}
    \left\|\B^{\O^\psi}-\Exp_{C^\psi}[\A^{C^\psi}]\right\|_\diamond\leq \frac{4T}{2^{n/2}}.
  \end{align}
\end{theorem}
\begin{proof}
  It suffices to give a simulation of $C^\psi$ using oracle access to $\O^\psi$. \
  The crucial idea from Kretschmer \cite{Kre21} is that if we have access to $\O^{\psi^\bot}$ for any state $\ket{\psi^\bot}$ orthogonal to $\ket{\psi}$, and can prepare $\ket{\psi^\bot}$ with any unitary $V^{\psi^\bot}\ket{0}=\ket{\psi^\bot}$, then
  \begin{align}
    \ket{0^n} 
    &\xmapsto{V^{\psi^\bot}} \ket{\psi^\bot} 
      \xmapsto{\O^{\psi}} \ket{\psi^\bot}
      \xmapsto{\O^{\psi^\bot}}\ket{\bot} 
      \xmapsto{\O^{\psi}} \ket{\psi}.
  \end{align}
  Then we consider another unitary $W=\proj{0}\oplus C'$ where $C'$ is a Haar random unitary on the Hilbert space $\Span\{\ket{x}:x\neq 0\}$.
  The sequence of unitaries 
  \begin{align}
    \ket{0^n} 
    &\xmapsto{W} \ket{0} 
      \xmapsto{V^{\psi^\bot}} \ket{\psi^\bot} 
      \xmapsto{\O^{\psi}} \ket{\psi^\bot}
      \xmapsto{\O^{\psi^\bot}}\ket{\bot} 
      \xmapsto{\O^{\psi}} \ket{\psi}, \\
    \ket{x}
    &\xmapsto{\O^{\psi}\O^{\psi^\bot}\O^{\psi}V^{\psi^\bot}W} \ket{\psi_x}, \qquad x\neq 0,
  \end{align}
  satisfying $\ket{\psi_x}$ is a Haar random unitary over spaces spanned by states orthogonal to $\ket{\psi}$ (and of course $\ket{\bot}$). \
  Setting $C^\psi:=\O^\psi\O^{\psi^\bot}\O^{\psi}V^{\psi^\bot}W$ satisfies all the requirements as desired. \

  However, since we do not know $\ket\psi$, we have no access to $\ket{\psi^\bot}$, $V^{\psi^\bot}$ and $\O^{\psi^\bot}$, thus the above maps may not be implemented (without using a large number of queries to $\O^\psi$). \
  Instead, if we sample a Haar random state $\ket{\varphi}$, $\Pr_{\ket{\varphi}}[|\langle\varphi|\psi\rangle|\geq \gamma]\leq e^{-n\gamma^2}$; thus with high probability they are nearly orthogonal. \
  Given the observation, we can sample a haar random state $\ket{\varphi}$, and use $\ket{\varphi}$ in place of $\ket{\psi^\bot}$ for our simulation of $C^\psi$. \
  Since the state $\ket{\varphi}$ is sampled uniformly instead of a fixed quantum state, 
  we must consider the quantum channel
  \begin{align}\label{eq:89}
    \Phi^\psi(\rho) := \Exp_{\ket{\varphi}}\left[\O^\psi\O^\varphi\O^\psi V^\varphi\rho(\O^\psi\O^\varphi\O^\psi V^\varphi)^\dag\right].
  \end{align}
  and show that $\Phi^\psi$ is close to $\C^\psi(\rho):=C^\psi\rho(C^{\psi})^\dag$ in diamond norm.

  Let $\V^\varphi$ be the unitary channel of $V^\varphi$.
  To show they are close, by triangle inequality, 
  \begin{align}\label{eq:90}
    \|\Phi^\psi-\C^\psi\|_\diamond
    \leq \left\|\Exp_{\ket{\varphi}} \V^\varphi - \V^{\psi^\bot}\right\|_\diamond
    + \left\| \Exp_{\ket{\varphi}}\O^\varphi - \O^{\psi^\bot}\right\|_\diamond.
  \end{align}
  Since all state $\ket{\psi^\bot}$ behaves equally well for the simulation purpose (i.e., every $\ket{\psi^\bot}$ orthogonal to $\ket{\psi}$ can be used simulates $C^\psi$ exactly), in \eq{90} our choice of $\ket{\psi^\bot}$ can be made to actually depend on $\ket{\varphi}$.
  More explicitly, we define $\ket{\psi^\bot}\in\Span\{\ket{\psi},\ket{\varphi}\}$ to be the unique orthogonal state to $\ket{\psi}$ (up to a phase).
  In this case, the diamond norm $\|\V^\varphi-\V^{\psi^\bot}\|_\diamond\leq 2|\langle\psi|\varphi\rangle|$.
  Since $\O^{\varphi}\O^{\psi^\bot}$ is a rotation by angle $2\alpha$ for $\alpha=\arccos(\cos^2(\theta/2))$, the eigenvalues are $e^{i2\alpha},e^{-i2\alpha},1$, and
  \begin{align}\label{eq:91}\nonumber
    \|\O^{\varphi}-\O^{\psi^\bot}\|_\diamond
    &\leq 2\sin \alpha \\\nonumber
    &= 2\sqrt{1-\cos^4(\theta/2)} \\\nonumber
    &= 2\sqrt{1-\left(\frac{1+\sqrt{1-|\epsilon|^2}}{2}\right)^2} \\
    &\leq 2|\epsilon|,
  \end{align}
  where $\epsilon=\langle\varphi|\psi\rangle$. \
  The last equality in \eq{91} holds since $\sqrt{1-x^2}\leq 1$ for $x\in[0,1]$. \
  By \eq{90} and a hybrid argument, \
  \begin{align}\label{eq:92}\nonumber
    \|\Phi^\psi-\C^\psi\|_{\diamond}
    &\leq 4\Exp_{\ket{\varphi}}[|\langle\psi|\varphi\rangle|] \\\nonumber
    &\leq 4\Exp_{\ket{\varphi}}[|\bra\varphi\psi\rangle|^2]^{1/2} \\
    &= \frac{4}{2^{n/2}}.
    \end{align}
  Now let $\A^{C^\psi}$ be a quantum algorithm given access to $C^\psi$ and $\B^{\O^\psi}=\A^{\Phi^\psi}$.
  By triangle inequality and \eq{92}, $\|\A^{C^\psi}-\B^{\O^\psi}\|_\diamond\leq \frac{4T}{2^{n/2}}$. \
  Also from \eq{89}, $\B^{\O^\psi}$ makes $2T$ queries to $\O^\psi$. \
\end{proof}

In the above analysis, note that the definition of $\ket{\bot}$ is not unique: every state $\ket{\bot}$ orthogonal to $\Span\{\ket{z}:z\in\bit^n\}$ can be used. 
In particular, we can replace $\ket{\bot}$ with $e^{i\theta}\ket{\bot}$, and the analysis will still go through.
Let $\O_\theta^\psi$ be a map $e^{i\theta}\ket{\bot}\mapsto\ket{\psi}$, $\ket{\psi}\mapsto e^{i\theta}\ket{\bot}$ and acts as the identity for states orthogonal to $\ket{\psi}$ and $\ket{\bot}$.
The observation leads to the following corollary.
\begin{corollary}\label{cor:AtoB}
  Let $\ket{\bot}$ be a fixed state orthogonal to $\Span\{\ket{z}:z\in\bit^n\}$.
  For $\theta\in[0,2\pi)$, every quantum state $\ket{\psi}$ and every $T$-query algorithm $\A^{C^\psi}$, there exists a $(2T)$-query algorithm $\B$ such that 
  \begin{align}
    \left\|\B^{\O_\theta^\psi}-\Exp_{C^\psi}[\A^{C^\psi}] \right\|_\diamond \leq \frac{4T}{2^{n/2}}.
  \end{align}
\end{corollary}

\subsubsection{From Canonical State Preparation Oracles to Resource States}

In this section, we apply the idea from Ambainis, Rosmanis and Unruh \cite{ARU14} to show that given resource states that depend on the Haar random state, one can approximate any algorithm given access to the canonical state preparation oracle. \ 
\begin{theorem}[{\cite[Theorem~3]{ARU14}}, paraphrased]\label{thm:aru14}
  Let $\ket{\psi}$ be a quantum state and $\O^\psi$ be a reflection about $\frac{1}{\sqrt{2}}(\ket{\psi}-\ket{\bot})$. \
  Let $\O$ be an oracle and $\rho$ be a quantum state. \
  Let $\ket{R^\psi}=\ket{\psi^\ell}\otimes\ket{\alpha^\psi_1}\otimes\ldots\otimes\ket{\alpha^\psi_m}$,
  where $\ket{\alpha_i^\psi}=\cos(\frac{i\pi}{2k})\ket{\psi}+\sin(\frac{i\pi}{2m})\ket{\bot}$.
  For every quantum state $\rho$ and algorithm $\B$ that on input $\rho$ and makes $T$-query to $\O^\psi$, there exists an algorithm $\G$ that on input $\rho,\ket{R}$ makes the same number of queries to $\O$ as $\B$ such that 
  \begin{align}\label{eq:85}
    \|\B^{\O^\psi,\O}(\rho)- \G^\O(\rho,\ket{R^\psi}) \|_\tr \leq O\left(T\left(\frac{1}{\sqrt{m}}+\frac{1}{\sqrt{\ell}}\right)\right).
  \end{align}
  Moreover, \eq{85} holds when $\ket{\psi}$, $\O$ and $\rho$ are not independent. \
\end{theorem}

For a detailed proof, see \cite{ARU14}. \
Here we provide some intuition. \
Let $S$ be the left cyclic shift operator and $\ket{\tilde R^\psi}=\ket{\alpha_1^\psi}\otimes\ldots\otimes\ket{\alpha_m^\psi}$. \  
One application of $S$ on $\ket{\bot}\ket{\tilde R^\psi}$ yields $S\ket{\bot}\ket{\tilde R^\psi}=\ket{\tilde R^\psi}\ket{\bot}\approx \ket{\psi}\ket{\tilde R^\psi}$ since the fidelity
\begin{align}
  |\bra{\psi,\tilde R^\psi}S\ket{\bot,\tilde R^\psi}|^2 
  = \cos\left(\frac{\pi}{2k}\right)^k \geq \left( 1-\frac{\pi^2}{4k^2}\right)^k \geq 1-\frac{\pi^2}{4k}
\end{align}
implies that the trace distance is $O(k^{-1/2})$. \ 
Furthermore, with an ideal resource state where $k\to\infty$, one can give rise to a reflection about $\frac{1}{\sqrt 2}(\ket{\psi}-\ket{\bot})$ using control-$S$, if the reflection $R$ about $\ket{\psi}$ can be implemented. \
Since $\ket{\psi}$ is unknown, it is not known how to implement $R$ exactly. \
However, with $\ell$ copies of $\ket{\psi}$, we can approximate $R$ to diamond distance $O(\ell^{-1/2})$: \
Let $V$ be the space of $(\ell+1)$-partite states invariant under permutations and $M$ be the projection onto $V$.
For $\ket{\Phi}=\ket{\phi}\ket{T}$ where $\ket{T}=\ket{\psi}^{\otimes \ell}$, if $\ket{\phi}=\ket{\psi}$, $\bra{\Phi}M\ket{\Phi}=1$; for $\ket{\phi}$ orthogonal to $\ket{\psi}$, $\bra{\Phi}M\ket{\Phi}\leq O(1/\ell)$. \
Thus a reflection $\Id-2M$ approximates $R$ to diamond distance $O(1/\sqrt{\ell})$, and after the measurement, the state is disturbed to trace distance $O(\ell^{-1/2})$. \
Finally, since there are $T$ queries, by a hybrid argument, the overall diamond distance is $O(T(m^{-1/2}+\ell^{-1/2}))$. \

For every state $\ket{\psi}$, let the state $\ket{R_\theta^\psi}$ be the same as $\ket{R^\psi}$ except that $\ket{\bot}$ is replaced with $e^{i\theta}\ket{\bot}$, i.e., 
\begin{align}\label{eq:96}
  \ket{R^\psi_\theta} := \ket{\psi^\ell} \otimes \bigotimes_{i=1}^k \left(
  \cos\left(\frac{i\pi}{2m}\right)\ket{\psi} + \sin\left(\frac{i\pi}{2m}\right)e^{i\theta}\ket{\bot}\right)
\end{align}
Since \thm{aru14} holds if every occurrence of $\ket{\bot}$ is replaced by $e^{i\theta}\ket{\bot}$, we have the following corollary:
\begin{corollary}\label{cor:BtoG}
  Let $\ket{\psi}$ be a quantum state and $\O_\theta^\psi$ be a reflection about $\frac{1}{\sqrt{2}}(\ket{\psi}-e^{i\theta}\ket{\bot})$.
  Let $\O$ be an oracle and $\rho$ be a quantum state.
  Let $\ket{R_\theta^\psi}$ be defined as in \eq{96}.
  For $\theta\in[0,2\pi)$, quantum state $\rho$ and every algorithm $\B$ that on input $\rho$ and makes $T$-query to $\O_\theta^\psi$, there exists an algorithm $\G$ that on input $\rho,\ket{R^\psi_\theta}$ makes the same number of queries to $\O$ as $\B$ such that 
  \begin{align}\label{eq:97}
    \|\B^{\O_\theta^\psi,\O}(\rho)-\G^\O(\rho,\ket{R_\theta^\psi})\|_\tr \leq O\left(T\left(\frac{1}{\sqrt{m}}+\frac{1}{\sqrt{\ell}}\right)\right).
  \end{align}
  Moreover, \eq{97} holds when $\ket{\psi}$, $\O$ and $\rho$ are not independent.
\end{corollary}

Combining \cor{AtoB} and \cor{BtoG}, the following corollary holds.
\begin{corollary}\label{cor:AtoG}
  For quantum state $\ket{\psi}$ %
  and every quantum algorithm $\A$ making $T$ queries to $C^\psi$, there exists a quantum algorithm $\G$ given access to $\ket{R_\theta^\psi}$ for uniform $\theta\in[0,2\pi)$ such that 
  \begin{align}
    \left\|\Exp_{C^\psi}\left[\A^{C^\psi}\right] - \G\left(\Exp_\theta[\proj{R_\theta^\psi}]\right)\right\|_\diamond \leq O(Tk^{-1/2}) + \frac{4T}{N^{1/2}}.
  \end{align}
\end{corollary}
\begin{proof}
  By \cor{AtoB}, for every $\theta\in[0,2\pi)$ and $T$-query algorithm $\A^{C^\psi}$, there exists $\B^{\O_\theta^\psi}$ such that their diamond distance is at most $4TN^{-1/2}$. \
  By \cor{BtoG}, for every $\theta\in[0,2\pi)$ and $T$-query algorithm $\B^{\O^\psi_\theta}$, there exists $\G(\ket{R^\psi_\theta})$ such that their diamond distance is at most $O(Tm^{-1/2})$ (we set $m=\ell=k/2$ to simplify the expression). \
  
  By triangle inequality,
  \begin{align}\nonumber
    \left\|\Exp_{C^\psi}\left[\A^{C^\psi}\right] - \Exp_\theta[\G(\proj{R_\theta^\psi})]\right\|_\diamond
    &\leq 
    \Exp_\theta\left\| 
    \Exp_{C^\psi}\left[\A^{C^\psi}\right] - \B^{\O^\psi_\theta}
    \right\|_\diamond 
    +
    \Exp_\theta\left\| 
    \B^{\O^\psi_\theta} - \G(\proj{R^\theta})
    \right\|_\diamond \\
    &\leq
      O(Tk^{-1/2}) + \frac{4T}{N^{1/2}}.
  \end{align}
  By linearity of quantum operations, $\Exp_\theta\left[\G\left(\proj{R^\psi_\theta}\right)\right] = \G\left(\Exp_\theta\left[\proj{R_\theta^\psi}\right]\right)$.
\end{proof}

\subsubsection{From Resource States to Samples}

Now we go from resource states to samples. 
For random $\ket{\psi}$, we write the state 
\begin{align}
\ket{\psi}=(\sqrt{P_0}e^{i\Theta_0},\ldots, \sqrt{P_{N-1}}e^{i\Theta_{N-1}}), 
\end{align}
where $P=(P_0,\ldots,P_{N-1})\sim\Dir(1^N)$ and each component $\Theta_i$ of $\Theta=(\Theta_0,\ldots,\Theta_{N-1})$ are independent random phases, i.e., each $\Theta_i$ is sampled according to the uniform distribution over $[0,2\pi)$. \
For every distribution $P$, we denote $\ket{P}:= (\sqrt{P_0},\ldots,\sqrt{P_{N-1}})$ and $\ket{\psi^P}=W\ket{P}$ for a random diagonal phase matrix $W=\diag(W_0,\ldots,W_{N-1}):=e^{i\diag(\Theta)}$. \
For every $\ket{\psi}$, let $\ket{R_\theta^{P,W}}$ be the associated resource state. \
In matrix form,
\begin{align}
  \ket{R_\theta^{P,W}} 
  &= \ket{\psi^{P,W}}^{\otimes\ell} \otimes \bigotimes_{j=1}^m \ket{\alpha_{i,\theta}^{P,W}},
\end{align}
where $\ket{\alpha_{i,\theta}^{P,W}}:= \cos(\frac{\pi i}{2k})\ket{\psi^{P,W}}+\sin(\frac{\pi i}{2k})e^{i\theta}\ket{\bot}$.

Let $\tilde W=W+e^{i\theta}\proj{\bot}$. %
The state $\ket{R_\theta^{P,W}}$ can be written as
\begin{align}
  \ket{R^{P,\tilde W}} := \ket{R^{P,W}_\theta}= (\tilde W\ket{P})^{\otimes \ell} \otimes \bigotimes_{j=1}^k (\tilde W\ket{P_i})
\end{align}
We will use the following lemma from \cite{Kre21}.
\begin{lemma}[{\cite[Lemma~15]{Kre21}}, paraphrased]\label{lem:resource-to-sample}
  There is an algorithm which prepares 
  \begin{align}
    \sigma^P := \Exp_{\tilde W}\left[\proj{R^{P,W}}\right]
  \end{align}
  by measuring $(k+\ell)$ copies of $\ket{\psi}$ in the standard basis.
\end{lemma}

Let $C^P$ be the following random circuit: Let $V\ket{0}=(\sqrt{P_0},\ldots,\sqrt{P_{N-1}})$, $C'$ be a Haar random matrix on $\Span\{\ket{z}:z\in\{1,\ldots,N-1\}\}$ and $C^P=WVC'$ for random diagonal phase matrix $W$.
We then prove the following theorem.
\begin{theorem}\label{thm:query-to-sample}
  For every distribution $P$, let $C^P$ be a random circuit sampled from the above process.
  For every algorithm $\A$ making $T$ queries to $C^P$, there exists a quantum algorithm $\F$ given access to $k$ samples drawn from $P$ such that 
  \begin{align}
    \left\|\bar\A^P-\bar\F^P\right\|_\diamond
    \leq O(Tk^{-1/2}) + \frac{4T}{N^{1/2}},
  \end{align}
  where $\bar\A^P:=\Exp_{C^P}[\A^{C^P}]$ and $\bar\F^P:=\Exp_{z_1,\ldots,z_k\sim P}\F(z_1,\ldots,z_k)$. \
\end{theorem}
\begin{proof}
  Let $\ket{\psi^P}=\tilde W\ket{P}$ for random phase matrix $\tilde W$ on $\Span\{\ket{0},\ldots,\ket{N-1},\ket{\bot}\}$. \
  The channel
  \begin{align}
    \Exp_{C^P}\left[\A^{C^P}\right] = \Exp_{\tilde W,C^\psi}\left[\A^{C^\psi}\right].
  \end{align}
  Let the quantum process in \lem{resource-to-sample} be $\Phi$, the process
  \begin{align} \nonumber
    \Exp_{\tilde W}\left[\G(\proj{R^{P,\tilde W}})\right]
    &= \G\left(\Exp_{\tilde W}\proj{R^{P,\tilde W}}\right) \\
    &= \G(\sigma^P) = \Exp_{z_1,\ldots,z_k\sim P}\G\circ\Phi(z_1,\ldots,z_k).
  \end{align}
  Now let $\F=\G\circ\Phi$.
  By triangle inequality and \cor{AtoG}, 
  \begin{align}\nonumber
    \left\|\Exp_{C^P}\left[\A^{C^P}\right] - \Exp_{z_1,\ldots,z_m\sim P}\F(z_1,\ldots,z_k)\right\|_\diamond
    &\leq 
    \Exp_{W}
    \left\|\Exp_{C^\psi}\left[\A^{C^\psi}\right] - \G\left(\Exp_\theta\proj{R_\theta^\psi}]\right)\right\|_\diamond \\
    &\leq O(Tk^{-1/2}) + \frac{4T}{N^{1/2}}.
  \end{align}
\end{proof}

\subsection{A Single-Round Analysis}\label{sec:single-general}

In this section, we prove our main result in \sec{general}. 
We show that with probability $1-N^{-\Omega(1)}$ over the choice of $P_C$ for Haar random $C$,
the conditional von Neumann entropy of any $T$-query device's output on Eve's information is at least $\Omega(\delta n)$, provided that $T=2^{O(n)}$ and the device solves $b$-XHOG for $b\approx 1+\delta$.

First, we consider a simplified device $\F$ which is only given sample access to $P_C$ and solves $b$-XHOG. \
We show this game is equivalent to
the following protocol: \ 
The verifier samples a distribution $P\sim\Dir(1^N)$, $z_1,\ldots,z_k\sim P$.
The verifier sends $z_1,\ldots,z_k$ to the device. \
Without loss of generality, we may assume that $P$ is revealed to Eve but not all the samples. \ 
The device is challenged to return a string $z$ and the verifier accepts if $z\in\{z_1,\ldots,z_k\}$. \
Recall that if $C$ is Haar random, $P_C$ is distributed according to $\Dir(1^N)$. \
A detailed description of the protocol is given in \fig{simplified}. \

The simplified protocol is equivalent to solving $\XHOG$ in the following sense: If the device solves $(2-\epsilon')$-$\XHOG$, then the device is accepted in the protocol described in \fig{simplified} with probability at least $1-\epsilon'$ for $\epsilon'=\epsilon+O(k^3/N)$. \
Then, we show that for every device in \fig{simplified} wins the protocol with probability $1-\epsilon'$, the conditional von Neumann entropy is at least $(1-\epsilon')n-o(n)$. \
This implies that the protocol in \fig{simplified} also certifies the conditional von Neumann entropy. \

  \begin{figure}
    \hrule\vspace{.5em}
  Input: security parameter $n$ and number of samples $k$.

  The protocol:
  \begin{enumerate}
  \item Eve and the device $\F$ share an arbitrary entangled state $\rho_{DE}$.
  \item The verifier samples a distribution $P\sim\Dir(1^N)$ (where $N=2^n$) and samples $z_1,\ldots,z_k\sim P$ which is sent to the device $\F$ (but not to Eve). Moreover, $P$ may be revealed to Eve (but not to the device).
  \item The device sends a string $z$.
  \item The verifier accepts if $z\in\{z_1,\ldots,z_k\}$.
  \end{enumerate} \hrule\vspace{1em}
  \caption{A simplified protocol.}
  \label{fig:simplified}
  \end{figure}

By \thm{query-to-sample}, any $T$-query algorithm $\A$ can be well approximated by an algorithm $\F$ that is given only sample access to $P_C$ to diamond distance $\delta=O(Tk^{-1/2})$. \
Furthermore, we will also show that if $\A$ solves $(2-\epsilon)$-XHOG, then $\F$ solves $b'$-XHOG for $b'=2-\epsilon-\delta$. \
Combining these results completes our single-round analysis. \ 

\subsubsection{A Simplified Device}

In this section, we show that for any algorithm $\F$ given $k$ samples drawn from $P$, the only way that $\F$ has high score is to output one of the given samples. \
Thus the proof system in \fig{simplified} is equivalent to solving $\XHOG$. \

First, we prove a technical lemma which will be useful later.

\begin{lemma}\label{lem:freq-dist}
  For $P\sim\Dir(1^N)$ and $z_1,\ldots,z_k\sim P$, let $m=(m_0,\ldots,m_{N-1})$ be the frequency vector with $m_z=|\{i\in[k]:z_i=z\}|$. \
  Then $m$ is distributed according to $\Phi(N,k)$, the uniform distribution over frequency vectors that has $N$ elements summing to $k$. \
\end{lemma}
\begin{proof}
  Recall that the probability density function (pdf) of $\Dir(1^N)$ is
  $f(p) = \Gamma(N)$,
  where $p=(p_0,\ldots,p_{N-1})$ is any element in the probability simplex. \
  Given a probability distribution $p$, the probability density that $k$ samples form a frequency vector $m$ is \
  \begin{align}
    f(m|p) = p_0^{m_0}\ldots p_{N-1}^{m_{N-1}}\frac{\Gamma(k+1)}{\Gamma(m_0+1)\ldots \Gamma(m_{N-1}+1)}.
  \end{align}
  Since the posterior disrbitution is $\Dir(1^N+m)$, the probability density \
  \begin{align}
    f(p| m) = p_0^{m_0}\ldots p_{N-1}^{m_{N-1}} \frac{\Gamma(N+k)}{\Gamma(m_0+1)\ldots\Gamma(m_{N-1}+1)}.
  \end{align}
  Thus the probability density of $m$ is
  \begin{align}
    f(m) = \frac{f(m|p) f(p)}{f(p|m)} = \frac{\Gamma(k+1)\Gamma(N)}{\Gamma(N+k)} = \binom{N+k-1}{k}^{-1}.
  \end{align}
  This means that $m$ is distributed according to the uniform distribution over frequency vectors that has $N$ elements summing to $k$. \
\end{proof}

The following theorem says that any algorithm $\F$ given $k$ independent samples drawn according to $P$ and solves $\frac{(2-\epsilon)N}{N+k}$-XHOG, $\F$ must output a given sample with probability at least $1-\epsilon-O(k^3/N)$. \

\begin{theorem}\label{thm:hog-simple}
  For every quantum channel $\F$, let $\bar\F^P(\rho):=\Exp_{z_1,\ldots,z_k\sim P}[\F(\rho,z_1,\ldots,z_k)]$. \
  If 
  \begin{align}
    \Exp_{P\sim\Dir(1^N)}\Exp_{z\sim\bar\F^P(\rho)}\left[ P_z \right] \geq \frac{2-\epsilon}{N+k},
  \end{align}
  then 
  \begin{align}
    \Pr_{P\sim\Dir(1^N),z_1,\ldots,z_k\sim P}[\F(\rho,z_1,\ldots,z_k)\in\{z_1,\ldots,z_k\}] \geq 1-\epsilon - O\left(\frac{k^3}{N}\right).
  \end{align}
\end{theorem}
\begin{proof}
  For each tuple of samples $(z_1,\ldots,z_k)$, let $m=(m_0,\ldots,m_{N-1})$ be the frequency vector. \ 
  Since there are $k$ samples, $\|m\|_1=k$. \
  Let $\Phi(N,k)$ be the uniform distribution over possible frequency vectors. \

  Sampling $m\sim\Phi(N,k)$ can be done with the following process: sample $P\sim\Dir(1^N)$ and $z_1,\ldots,z_k\sim P$; output the frequency vector $m$. \
  By \lem{freq-dist}, the frequency vector has infinity norm 1, i.e., $\max_z m_z=1$, with probability
  \begin{align}\nonumber\label{eq:prob-collision}
    \Pr_{m\sim\Phi(N,k)}[\|m\|_\infty=1] 
    &= \frac{\binom{N}{k}}{\binom{N+k-1}{k}} \\\nonumber
    &= \frac{N!(N-1)!}{(N-k)!(N+k-1)!} \\\nonumber
    &= \frac{(N-1)\ldots (N-k+1)}{(N+k-1)\ldots (N+1)} \\\nonumber
    &\geq \left(1-\frac{k}{N}\right)^{k-1} \\
    &= 1-\frac{O(k^2)}{N}.
  \end{align}
  The first equality holds by \lem{freq-dist} and the fact that the number of frequency vector that has norm 1 is $\binom{N}{k}$. \
  The rest follows by direct calculation. \
  
  For every $\F$ that learns $(z_1,\ldots,z_k)$, the posterior distribution $P|m\sim\Dir(1^N+m)$. \
  That is, seeing these samples, the distribution of $P$ to $\F$ is distributed according to $\Dir(1^N+m)$. \
  Thus the expectation 
  \begin{align}\nonumber
    \Exp_{P\sim\Dir(1^N)}[P_z|m] 
    &= \Exp_{P\sim\Dir(m+1^N)}[P_z] \\
    &= \frac{m_z+1}{N+k}.
  \end{align}

  Without loss of generality, the output of any algorithm $\F(z_1,\ldots,z_k,\rho)$ can be described with a distribution $Q(m)$ that only depends on the frequency vector $m$. \
  For each $m$, the score of the algorithm is \
  \begin{align}
    \Exp_{z\sim Q(m)}\Exp_{P\sim\Dir(1^N+m)}[P_z]
    &= \Exp_{z\sim Q(m)} \left[\frac{m_z+1}{N+k}\right].
  \end{align}
  For each $m$ such that $\|m\|_\infty=1$, 
  \begin{align}
    \Exp_{z\sim Q(m)}\Exp_{P\sim\Dir(1^N+m)}[P_z]
    = \frac{1}{N+k} \left(1 + \Pr_{z\sim Q(m)}[m_z>0]\right).
  \end{align}
  Let $\gamma$ be the probability that $\|m\|_\infty>1$ for $m\sim\Phi(N,k)$.
  From \eq{prob-collision}, $\gamma=O(k^2/N)$. \
  The score of $\F$ can be calculated with the expectation
  \begin{align}\nonumber
    &\Exp_{m\sim\Phi(N,k)}\Exp_{z\sim Q(m)}\Exp_{P\sim\Dir(1^N+m)}[P_z] \\\nonumber
    &\qquad\leq (1-\gamma) \frac{1}{N+k}\left(1+\Pr_{m\sim\Phi(N,k),z\sim Q(m)}[m_z>0]\right) + \gamma \frac{k+1}{N+k} \\
    &\qquad\leq \frac{1}{N+k} + \frac{1}{N+k}\Pr_{m\sim\Phi(N,k),z\sim Q(m)}[m_z>0] + \frac{O(k^3)}{N(N+k)}.
  \end{align}
  Thus if the score is at least $\frac{2-\epsilon}{N+k}$,
  \begin{align}
    \Pr_{m\sim\Phi(N,k), z\sim Q(m)}[m_z>0] \geq 1-\epsilon-\frac{O(k^3)}{N}.
  \end{align}
\end{proof}
For every $P$, if $\F$ outputs one of the given samples in the protocol in \fig{simplified}, then it is not difficult to see the von Neumann entropy of $\F$'s output conditioned on Eve's side information is at least $H_{\min}(P)$ up to an $O(\log k)$ additive loss. \
We show that the entropy lower bound scales linearly in the probability that $\F$ outputs one of the given samples. \

\begin{theorem}\label{thm:vn-bound}
  For distribution $P$ over $\bit^n$, let $\F$ be any algorithm given access to $z_1,\ldots,z_k\sim P$ and the first subsystem of any bipartite state $\rho_{DE}$ and $\bar\F(\rho)=\Exp_{z_1,\ldots,z_k\sim P}[\F(\rho,z_1,\ldots,z_k)]$.
  If 
  \begin{align}
    \Pr_{z_1,\ldots,z_k\sim P}\left[\F(\rho,z_1,\ldots,z_k) \in \{z_1,\ldots,z_k\} \right] = 1-\delta,
  \end{align}
  then 
  \begin{align}
    H(Z|E)_{\bar\F(\rho)} \geq (1-\delta) (H_{\min}(P) - 2\log k) - 2.
  \end{align}
\end{theorem}
\begin{proof}
  We give a lower bound on the von Neumann entropy $H(Z|E)_{\bar\F^P(\rho)}$ for random variable $Z\in\{Z_1,\ldots,Z_k\}$ and $Z_1,\ldots,Z_k\sim P$. \
  Here we use uppercase $Z_1,\ldots,Z_k$ to denote the registers or the random variables describing these samples, and lowercase $z_1,\ldots,z_k$ to denote a particular event. \
  We define the classical-quantum state 
  \begin{align}
    \psi_{Z_1\ldots Z_k ZE} = \sum_{z_1,\ldots z_k} P(z_1)\ldots P(z_k) \proj{z_1,\ldots,z_k}_{Z_1\ldots Z_k} \otimes \F(\rho,z_1,\ldots,z_k)_{ZE}.
  \end{align}
  By definition, $\bar\F^P(\rho)=\tr_{Z_1\ldots Z_k}(\psi)$, and thus $H(Z|E)_{\bar\F^P(\rho)}=H(Z|E)_{\psi}$. \

  Applying \lem{triangle} with $A=Z_1\ldots Z_k$, $B=Z$ and $C=E$, 
  \begin{align}
    H(Z|E)_{\psi} \geq H(Z_1\ldots Z_k|E)_{\psi} - H(Z_1\ldots Z_k|Z)_{\psi}.
  \end{align}
  Since the register $Z$ is disjoint from $Z_1\ldots Z_kE$, 
  \begin{align}
    \tr_{Z}(\psi) = \sum_{z_1,\ldots,z_k} P(z_1)\ldots P(z_k)\proj{z_1,\ldots,z_k}_{Z_1\ldots Z_k} \otimes\rho_E,
  \end{align}
  which is product state. \
  Thus we have $H(Z_1\ldots Z_k|E)_{\psi}=H(Z_1\ldots Z_k)_{\psi}$, and therefore
  \begin{align}\label{eq:vn-ZE}\nonumber
    H(Z|E)_{\psi} 
    &\geq H(Z_1\ldots Z_k)_\psi - H(Z_1\ldots Z_k|Z)_\psi \\
    &= I(Z:Z_1,\ldots,Z_k)_\psi.
  \end{align}
  This means that it suffices to bound the mutual information between $Z$ and $Z_1,\ldots,Z_k$. \
  Now we give a lower bound on $H(Z)_{\psi}$: \
  Recall that $\psi$ is a quantum state of the form %
  \begin{align}
    \psi = \sum_{z_1,\ldots,z_k} P(z_1)\ldots P(z_k) \proj{z_1,\ldots,z_k} \otimes \sum_z Q(z|z_1,\ldots,z_k) \proj{z},
  \end{align}
  for some conditional distribution $Q(\cdot | z_1,\ldots,z_k)$ which may depend on $P$. \
  We decompose each distribution
  \begin{align}
    Q(z|z_1,\ldots,z_k) = Q_1(z | z_1,\ldots,z_k) + Q_0(z | z_1,\ldots,z_k)
  \end{align}
  into two subnormalized distributions $Q_1(\cdot |z_1,\ldots,z_k)$ and $Q_0(|z_1,\ldots,z_k)$ such that the support $\mathrm{supp}(Q_1(\cdot |z_1,\ldots,z_k)\subseteq\{z_1,\ldots,z_k\}$ and 
  $\mathrm{supp}(Q_0(\cdot | z_1,\ldots,z_k))\cap\{z_1,\ldots,z_k\}=\emptyset$.
  Let $\delta_{z_1,\ldots,z_k}=\sum_z Q_0(z|z_1,\ldots,z_k)$.
  By definition, $\Exp_{z_1,\ldots,z_k\sim P}[\delta_{z_1,\ldots,z_k}]=\delta$. \
  Furthermore, for $b\in\bit$, let the normalized distribution \
  \begin{align}
    \bar Q_b(z|z_1,\ldots,z_k) = \frac{1}{p_{b,z_1,\ldots,z_k}}Q_b(z|z_1,\ldots,z_k),
  \end{align}
  where $p_{1,z_1,\ldots,z_k}=1-\delta_{z_1,\ldots,z_k}$ and $p_{0,z_1,\ldots,z_k}=\delta_{z_1,\ldots,z_k}$. \
  Now let $(p_0,p_1)=(\delta,1-\delta)$. \
  For $b\in\bit$, define
  \begin{align}
    \psi_b := \frac{1}{p_b}\sum_{z_1,\ldots,z_k} p_{b,z_1,\ldots,z_k} P(z_1)\ldots P(z_k)\proj{z_1,\ldots,z_k} \otimes \sum_z \bar Q_b(z|z_1,\ldots,z_k)\proj{z}.
  \end{align}
  We have $\psi=p_0\psi_0+p_1\psi_1$. \
  By \lem{mutual} and non-negativity of quantum mutual information,
  \begin{align}\label{eq:mutual}\nonumber
    I(Z:Z_1\ldots Z_k)_\psi 
    &\geq p_0 I(Z:Z_1\ldots Z_k)_{\psi_0} + p_1 I(Z:Z_1\ldots Z_k)_{\psi_1} - h(\delta) \\ 
    &\geq (1-\delta) I(Z:Z_1\ldots Z_k)_{\psi_1} -1.
  \end{align}
  It suffices to give a lower bound on $I(Z:Z_1\ldots Z_k)_{\psi_1}$.

  Let $\sigma_1:=\tr_{Z_1\ldots Z_k}(\psi_1)$. \
  By definition,
  \begin{align}\nonumber
    \sigma_1 
    &=
    \frac{1}{1-\delta} 
    \sum_{z_1,\ldots,z_k} (1-\delta_{z_1,\ldots,z_k}) P(z_1)\ldots P(z_k)
    \sum_z \bar Q_1(z|z_1,\ldots,z_k)\proj{z} \\\nonumber
    &\leq
    \frac{1}{1-\delta} 
    \sum_{z_1,\ldots,z_k} (1-\delta_{z_1,\ldots,z_k}) P(z_1)\ldots P(z_k)
    \sum_z \Id[z\in\{z_1,\ldots,z_k\}]\proj{z} \\
    &\leq
    \frac{1}{1-\delta} 
    \sum_{z_1,\ldots,z_k} (1-\delta_{z_1,\ldots,z_k}) P(z_1)\ldots P(z_k)
    \sum_z \sum_{i=1}^k \Id[z=z_i]\proj{z}.
  \end{align}
  Thus for every $z\in\bit^n$,
  \begin{align}
    \bra{z}\sigma_1\ket{z}
    &\leq \sum_i \frac{1-\delta^{(i)}_z}{1-\delta} P(z),
  \end{align}
  where $\delta_z^{(i)}:= \Exp_{z_1,\ldots,z_{i-1},z_{i+1}\ldots,z_k\sim P}[\delta_{z_1,\ldots,z_k}]$. \
  By definition, $\delta_z^{(i)}\in[0,1]$, and thus $\bra{z}\sigma_1\ket{z}\leq \frac{k}{1-\delta}\max_z P(z)$. \
  This implies that 
  \begin{align}\label{eq:vn-Z}
    H(Z)_{\psi_1} \geq H_{\min}(Z)_{\sigma_1} \geq H_{\min}(P) - \log k + \log(1-\delta).
  \end{align}
  Next we consider the quantity $H(Z|Z_1\ldots Z_k)_{\psi_1}$: \ 
  By definition,
  \begin{align}
    H(Z|Z_1\ldots Z_k)_{\psi_1} = \frac{1}{1-\delta}\sum_{z_1,\ldots,z_k} (1-\delta_{1,z_1,\ldots,z_k})P(z_1)\ldots P(z_k) H(\bar Q_1(\cdot | z_1,\ldots,z_k))
  \end{align}
  Since $\bar Q_1(\cdot|z_1,\ldots,z_k)$ has support in $\{z_1,\ldots,z_k\}$, for every $z_1,\ldots,z_k$,
  \begin{align}
    S(\bar Q_1(\cdot|z_1,\ldots,z_k)) \leq \log k.
  \end{align}
  This gives 
  \begin{align}\label{eq:mutual-2}
    H(Z|Z_1\ldots Z_k)_{\psi_1}\leq \log k.
    \end{align}
  Combining \eq{vn-ZE}, \eq{mutual}, \eq{vn-Z} and \eq{mutual-2}, 
  \begin{align}\nonumber
    H(Z|E)_\psi 
    &\geq (1-\delta)I(Z:Z_1\ldots Z_k)_{\psi_1} - 1 \\\nonumber
    &\geq (1-\delta) \left(H_{\min}(P) - 2\log k\right) + (1-\delta)\log(1-\delta) - 1 \\
    &\geq (1-\delta) \left(H_{\min}(P) - 2\log k\right) -2.
  \end{align}
  The last inequality holds since $-(1-\delta)\log(1-\delta)\leq h(\delta)\leq 1$.
\end{proof}
\thm{vn-bound} says that the von Neumann entropy depends on the probability that the device's output $z\in\{z_1,\ldots,z_k\}$. \
For $P\sim\Dir(1^N)$, by \lem{haar-min}, with overwhelming probability, $S_{\min}(P)\geq n-\log n-O(1)$. \
We now use the notation
\begin{align}
  H(Z|PE)_\psi := \Exp_{P\sim\Dir(1^N)}\left[ H(Z|E)_{\bar\F^P(\rho)}\right],
\end{align}
where $\psi$ is the joint output classical-quantum state, defined by
\begin{align}\label{eq:psi}
  \psi = \Exp_{P\sim\Dir(1^N)}\left[\proj{P} \otimes \bar\F^P(\rho)\right],
\end{align}
even though $P$ is not sampled from a finite set. \ 
We then show the following theorem. \
\begin{corollary}\label{cor:simple-vn}  For $P\sim\Dir(1^N)$, let $\F$ be any algorithm given access to $z_1,\ldots,z_k\sim P$ and the first subsystem of any bipartite state $\rho_{DE}$, and $\bar\F^P(\rho)=\Exp_{z_1,\ldots,z_k\sim P}[\F(\rho,z_1,\ldots,z_k)]$.
  If 
  \begin{align}
    \Pr_{P\sim\Dir(1^N), z_1,\ldots,z_k\sim P}[\F(\rho,z_1,\ldots,z_k)\in\{z_1,\ldots,z_k\}] \geq 1-\epsilon,
  \end{align}
  then
  \begin{align}
    H(Z|PE)_\psi
    \geq (1-\epsilon)n-\log n-2\log k - O(1),
  \end{align}
  where $\psi$ is the state defined in \eq{psi}.
\end{corollary}
\begin{proof}
  The corollary is a direct consequence of \thm{vn-bound}.
  Let the error probability for $P$ be $\epsilon_P$.
  By \thm{vn-bound}, 
  \begin{align}\nonumber
    H(Z|E)_{\bar\F^P(\rho)} 
    &\geq (1-\epsilon_P) H_{\min}(P) - 2\log k \\
    &\geq H_{\min}(P) - \epsilon_P n -2\log k.
  \end{align}
  By \lem{haar-min-avg}, 
  \begin{align}\nonumber
    H(Z|PE)_\psi 
    &\geq  \Exp_{P\sim\Dir(1^N)}[H_{\min}(P)] - \Exp_P[\epsilon_P] n -2\log k \\
    &\geq (1-\epsilon)n-\log n-2\log k -O(1).
  \end{align}
\end{proof}
Combining \thm{hog-simple} and \cor{simple-vn}, we conclude that any $\bar\F$ solving $(1+\delta)$-XHOG has conditional von Neumann entropy $n-o(n)$.
\begin{corollary}\label{cor:vn-bound-F}
  Let $\F$ be any device given access to $k$ samples from $P=P_C$ for $C\sim\Haar(N)$. \
  If $\F$ solves $(1+\delta)$-$\XHOG$,
  then 
  \begin{align}
      H(Z|PE)_\psi 
      \geq \left(\delta-O(k^3/N)\right) n -\log n -2\log k -O(1).
  \end{align}
\end{corollary}

\subsubsection{The Analysis of a General Device}

Given the analysis of the protocol in \fig{simplified} and the closeness of a general device and the simplified device, we show that any device must exhibit a von Neumann entropy lower bound if it solves $\XHOG$. \

\begin{theorem}\label{thm:close-score}
For every $T$-query device $\A$ that solves $b$-$\XHOG$, there exists a device $\F$ given access to $k$ samples such that $\F$ solves $b'$-XHOG for $b'\geq b-O(nTk^{-1/2})$.
\end{theorem}
\begin{proof}
For every $\A^P$, by \thm{query-to-sample}, there exists $\F^P$ such that 
\begin{align}\nonumber
  \left|\Exp_{z\sim\bar\A^P(\rho)}[P_{z}] - \Exp_{z\sim\bar\F^P(\rho)}[P_{z}]\right|
  &= \left|\sum_{z} P_z (p_\A(z) - p_\F(z))\right| \\\nonumber
  &\leq \max_z P_z \sum_z |p_A(z)-p_F(z)| \\\nonumber
  &\leq \max_z P_z \cdot \|\bar\A^P(\rho)-\bar\F^P(\rho)\|_{\tr} \\
  &\leq \max_z P_z \cdot (O(Tk^{-1/2}) + O(TN^{-1/2})).
\end{align}
Taking the expectation over $P\sim\Dir(1^N)$, the upper bound is at most $O(nT(k^{-1/2}+N^{-1/2})/N)$ by \lem{haar-min-avg}. \ 
This implies that if $\A$ solves $b$-XHOG, then $\F$ solves $(b-O(nTk^{-1/2}+nTN^{-1/2}))$-XHOG. \
\end{proof}

\thm{close-score} shows that the score of any device $\A$ and its simplified version $\F$ must be close if $\F$ is given sufficiently many samples.
Given this fact, we prove the following theorem.
\begin{theorem}\label{thm:single}
  Let $\A$ be any algorithm making $T$ queries to a Haar random unitary $C$. %
  If $\A$ solves $(1+\delta)$-XHOG, then there exists a state $\psi$ which is $O(Tk^{-1/2})$-close to $\A$'s output classical-quantum state such that
  \begin{align}
    H(Z|CE)_\psi \geq \left(\delta-O(k^3/N)-O(nTk^{-1/2}+nTN^{-1/2})\right) n - \log n - 2\log k - O(1).
  \end{align}
 \end{theorem}
 \begin{proof}
   By \thm{query-to-sample}, for every $T$-query $\A$, there exists $\F$ given access to $k$ samples (drawn according to $P\sim\Dir(1^N)$) such that 
   $\|\bar\A^P-\bar\F^P\|_\diamond \leq O\left(Tk^{-1/2}\right)$.
   Also by \thm{close-score}, $\F$ solves $b$-XHOG for $b=1+\delta-O(nTk^{-1/2})$. \
   By \cor{vn-bound-F}, 
   \begin{align}
     H(Z|PE)_\psi \geq \left(\delta - O(k^3/N) - O(nTk^{-1/2}+nTN^{-1/2}) \right) n - \log n - 2 \log k - O(1),
   \end{align}
   where $\psi_{ZPE}$ is $\F$'s output state defined in \eq{psi} and is $O(Tk^{-1/2})$-close to $\A$'s output. \ 
   Finally, by \thm{symmetry}, we conclude the proof. \ 
 \end{proof}

 By \thm{single}, for constant $\delta$, every algorithm solving $(1+\delta)$-XHOG outputs a sample of conditional von Neumann entropy $o(n)$ must make $N^{\Omega(\delta)}$ queries. \ 
 For devices making $\poly(n)$ queries, \thm{single} implies the following corollary.
 \begin{corollary}\label{cor:single}
   For $T=\poly(n)$, $\delta=\Omega(1)$, and $\eta\in (0,1]$, 
   every algorithm $T$-query algorithm which solves $(1+\delta)$-$\XHOG$ must output a sample $Z$ satisfying 
   \begin{align}
     H(Z|CE)_\psi \geq (1-\eta)\delta n - O(\log n),
   \end{align}
   where $\psi$ is a quantum state $N^{-\Omega(\delta\eta)}$-close to $\A$'s output state.
 \end{corollary}
 \begin{proof}
   By \thm{single}, choosing $k=T^2N^{\delta\eta}$, $Tk^{-1/2}=N^{-\delta\eta/2}$ and $k^3/N=n^{O(1)}\cdot N^{-1+3\delta\eta}$. \
   This implies the von Neumann entropy has lower bound $(1-\eta)\delta n-O(\log n)$.
 \end{proof}

\subsection{Entropy Accumulation}\label{sec:general-ea}

  \begin{figure}
    \hrule\vspace{.5em}
  Input: security parameter $n$, the number of rounds $m$, the score parameter $\delta\in[0,1]$, the fraction of circuit updates $\gamma=O((\log n)/m)$, and the fraction of test rounds $\eta=O((n^2\log n)/m)$.\\

  The protocol:

  \vspace{.5em}
  \begin{enumerate}
    \item For $i=1,\ldots,m$, run the following steps:
  \begin{enumerate}
  \item The verifier samples $T_i\sim\Bernoulli(\gamma)$,
  and $F_i\sim\Bernoulli(\eta)$.
    If $T_{i-1}=1$ or $i=1$, the device chooses a fresh circuit $C_i\sim\Haar(N)$.
    Otherwise, if $T_{i-1}=0$ and $i>1$, then the device sets $C_i=C_{i-1}$ to be the circuit used in the previous round.
    The verifier sends $C_i$ to the device.

  \item The prover returns a sample $z_i$.
  \end{enumerate} 
  
  \item 
  Let the number of epoches, i.e., the set of consecutive rounds $i$ such that the same circuit $C_i=C$ is used, be $t$.
  For each epoch $E_j$, let $t_j=|\{i\in E_j:F_i=1\}|$ denote the number of test rounds in this epoch. 
  The verifier computes 
    \begin{align}
      s_j = \frac{1}{t_j}\sum_{i\in E_j:T_i=1} p_{C_i}(z_i).
      \end{align}
      If $\frac{1}{t}\sum_{j=1}^t\delta[s_j\geq (1+\delta)/N]\geq 0.99$, then the verifier accepts and outputs $(z_1,\ldots,z_m)$ to the quantum-proof randomness extractor.
  \end{enumerate} 
  \hrule\vspace{1em}
  \caption{The entropy accumulation protocol.}
  \label{fig:full}
  \end{figure}

We present our entropy accumulation protocol in \fig{full}. \
By the entropy accumulation theorem shown in \sec{eat}, we prove the following lower bound on the conditional min-entropy. \ 
\begin{theorem}
  Let $\A_1,\ldots,\A_m$ be $n^{O(1)}$-query sequential processes given access to the first system of a bipartite state $\rho_{DE}$ and outputting $z_1,\ldots,z_m$ solving $\LXEB_{1+\delta,k}$ with probability $p$.
  Then
  \begin{align}
    H_{\min}^{\epsilon+m\epsilon'}(Z^m|C^mT^mE)_{\A_m\circ\ldots\A_1(\rho)|\Omega} \geq n\left( 0.99\delta m - c\sqrt{m} -o(1)\right)
  \end{align}
  where $\epsilon'=2^{-\Omega(n)}$ and $\Omega$ denotes the event of non-aborting.
The parameter $c:=3.99(1+\log(\frac{1}{p\epsilon}))^{1/2}$.
\end{theorem}
\begin{proof}
We apply the entropy accumulation theorem shown in \sec{eat}. \
In particular, we choose $G:=N\sum_z \proj{z} P(z)$. \
By \cor{single}, we choose $f(\delta)=(1-\eta)\delta n -O(\log n)$ for $\eta=0.01$ for concreteness. \
This gives $\|\nabla f\|_\infty=0.99n$. \
Then we have 
\begin{align}\nonumber
V &= 2(\log(2d_Z+1)+\|\nabla f\|_\infty) \\\nonumber
&= 2(\log(2N+1)+0.99n) \\\nonumber
&= 3.98n+O(1) \\
&\leq 3.99n.
\end{align}
Now we have 
\begin{align}
  H_{\min}^{\epsilon+m\epsilon'}(Z^m|C^mT^m E)_{\A_m\circ\ldots\circ\A_1(\rho)|\Omega}
  &\geq n \left( 0.99\delta m - O\left(\frac{\log n}{n}\right) -3.99\sqrt{m}\sqrt{\log\frac{2}{\Pr_\sigma[\Omega]^2\epsilon^2}}\right).
\end{align}
\end{proof}

\section{Pseudorandomness and Statistical Zero-Knowledge}\label{sec:szk}

To produce a net gain in randomness, we must derandomize the challenge circuits in our protocols \fig{llqsv} and \fig{full}. \
However, existing quantum-secure pseudorandom function was not known to be secure for certified randomness since the security only holds against quantum polynomial-time adversaries. \
In \sec{qszk-prf}, we define a stronger security notion for pseudorandom functions sufficient for certified randomness, and show there exists a construction relative to a random oracle. \
Then in \sec{expansion}, with the pseudorandom function, polynomial expansion can be achieved. \

In this section, we show that a pseudorandom function against quantum statistical zero-knowledge protocols is sufficient for certified randomness. \
The same argument also applies to pseudorandom unitaries. \

By \thm{mw18} with an extension to oracle distributions, two oracle distributions $\D_0,\D_1$ are said to be $\QSZK$-distinguishable if there exists a pair of algorithms $\A,\B$ such that if $\Exp_{F\sim \D_0}\|\A^F-\B^F\|_{\tr}$ and $\Exp_{F\sim\D_1}\|\A^F-\B^F\|_{\tr}$ are far apart.
\begin{definition}[$\QSZK$-distinguishability]\label{dfn:qszk-distinguishable}
  Two oracle distributions are said to be $\QSZK$-distinguishable if there exist a pair of algorithms $\A,\B$ and a polynomial $p$ such that 
  \begin{align}\label{eq:qszk-distinguishable}
    \left| \Exp_{F\sim\D_0}\|\A^\F-\B^F\|_{\tr} - \Exp_{F\sim\D_0}\|\A^\F-\B^F\|_{\tr}\right| > \frac{1}{p(n)}.
  \end{align}
  If no such algorithms exist, then the distributions are said to be $\mathrm{QSZK}$-indistinguishable.
\end{definition}

\dfn{qszk-distinguishable} strengthens the definition against $\BQP$ adversaries by Zhandry \cite{Zha12}: \
Two oracle distributions $\D_0,\D_1$ are said to be distinguishable if there exists an algorithm $\A$ which outputs one with probability far apart, i.e.,
\begin{align}\label{eq:standard-distinguishable}
  \left|\Pr_{F\sim\D_0}[\A^F=1] - \Pr_{F\sim\D_1}[\A^F=1]\right| > \frac{1}{\poly(n)}.
\end{align}
For clarity, in this paper, we will say $\D_0,\D_1$ are standard-distinguishable if there exists $\A$ such that \eq{standard-distinguishable} holds. \
To see why \dfn{qszk-distinguishable} is a stronger definition, we prove the following theorem.
\begin{theorem}\label{thm:qszk-standard}
  If two oracle distributions $\D_0,\D_1$ are standard-distinguishable, then they are $\QSZK$-distinguishable.
\end{theorem}
\begin{proof}
  Since $\D_0,\D_1$ are standard-distinguishable, there exists $\A$ which outputs a single bit such that \eq{standard-distinguishable} holds. \ 
  Then consider a trivial algorithm $\B$ which outputs $0$ with probability 1. \ 
  This implies that 
  \begin{align}
    \Exp_{F\sim\D_b}\|\A^F-\B^F\|_{\tr} = \Pr_{F\sim\D_b}[\A^F=1],
  \end{align}
  and \eq{qszk-distinguishable} holds.
\end{proof}

\subsection{Pseudorandom Functions}\label{sec:qszk-prf}

It is known that there is a construction of a pseudorandom function standard-indistinguishable from a random function. \
As we do not know if the converse of \thm{qszk-standard} holds, we propose the following assumption. \
\begin{assumption}[Pseudorandom function assumption]\label{asm:prfa}
  Let $\kappa\in\mathbb N$ be the security parameter and $\ell,m$ be polynomially bounded functions. \
  There exists a keyed function $F:\bit^{\ell(\kappa)}\times\bit^{m(\kappa)}\to\bit$ such that the following conditions hold. \
  \begin{itemize}
      \item $F(k,x)$ can be computed in polynomial time for every $k\in\bit^{\ell(\kappa)}$ and $x\in\bit^{m(\kappa)}$. \
      
      \item For every pair of quantum algorithms $\A,\B$, it holds that \
      \begin{align}
      \left|
      \Exp_{k\sim\K_{\ell(\kappa)}}\|\A^{F_k}-\B^{F_k}\|_{\tr}
      - \Exp_{F\sim\F_{m(\kappa)}} \|\A^{F}-\B^{F}\|_{\tr}
      \right| \leq \negl(\kappa)
    \end{align}
    for sufficiently large $\kappa$, where $\F_{m}$ is the uniform distribution over $m$-bit Boolean functions and $\K_\ell$ is the uniform distribution over $\bit^{\ell}$. \
  \end{itemize}
\end{assumption}

In this section, we prove that \asm{prfa} holds relative to a random oracle: \ 
For $\kappa=n$, $m(n)=n$, let $\O:\bit^{\ell+n}\to\bit$ be a random function and $\O_k(x):=\O(k,x)$ for $k\in\bit^\ell$ and $x\in\bit^n$. \
The key length $\ell$ is to be determined later. \ 
We show that for every pair of $\poly(n)$-query algorithms $\A$ and $\B$, $\O,\O_k$ for uniform $k\in\bit^\ell$ is indistinguishable from $\O,H$ for random $H:\bit^n\to\bit$. \ 

Recall that for a random function $H:\bit^n\to\bit$ and $x\in\bit^n$, $H(x)$ is a fair coin independent from $H(x')$ for every $x'\neq x$. \
We will also say a distribution $\D$ over $n$-bit Boolean functions is a $\epsilon$-biased random function if for $F\sim\D$, the random variables in $\{F(x):x\in\bit^n\}$ are independent, and for each $x\in\bit^n$, the marginal distribution of $F$ on $x$ satisfies $|\Pr_{F\sim\D}[F(x)=1]-1/2|\leq\epsilon$. \  
First we prove the following lemma which states that for the uniform distribution $\U_\ell$ over $\bit^\ell$ and every $x\in\bit^n$, $\Exp_{k\sim\U_\ell}\O_k(x)$ is weakly biased under uniform $k$. \ 
\begin{lemma}\label{lem:qszk-prf-1}
  For integer $m$, let $\F_m$ be the uniform distribution over functions $\bit^m\to\bit$. \ 
  Then 
  \begin{align}
    \Pr_{\O\sim\F_{\ell+n}}\left[\forall x\in\bit^n,\left|\Exp_{k\sim\U_\ell}\O_k(x) - 1/2\right|\leq \epsilon\right] \geq 1-2\cdot 2^Ne^{-L\epsilon^2}. \ 
  \end{align}
\end{lemma}
\begin{proof}
  By Hoeffding's inequality, for every $x$, 
  \begin{align}
    \Pr_{\O\sim\F_{\ell+n}}\left[\left|\Exp_{k\sim\U_\ell}\O_k(x) - 1/2\right|\leq \epsilon\right] \geq 1-2e^{-L\epsilon^2}. \ 
  \end{align}
  By a union bound, we conclude the proof. \ 
\end{proof}
Fix a function $\O$ that satisfies this likely event. \ 
Since for each $k\in\bit^\ell$, $\{\O_k(x):x\in\bit^n\}$ are independent, $\O_k$ under uniform $k$ is an $\epsilon$-biased random function. \
The following lemma shows that no pair of $T$-query algorithms can distinguish a slightly biased random function from a random function with advantage $O(T\sqrt\epsilon)$ for the difference in biases $\epsilon$. \  
\begin{lemma}\label{lem:qszk-prf-2}
  Let $\A$ and $\B$ be two $T$-query algorithms given access to an oralce sampled from either $\D_0$ or $\D_1$ such that for every $x\in\bit^\ell$, $|\Pr_{F\sim\D_0}[F(x)=1]-\Pr_{F\sim\D_1}[F(x)=1]|\leq \epsilon$. \ 
  Then
  \begin{align}
    \left|\Exp_{F\sim\D_0}\|\A^F-B^F\|_\tr - \Exp_{F\sim\D_1}\|\A^F-\B^F\|_{\tr}\right| \leq 16T\sqrt\epsilon. \ 
  \end{align}
\end{lemma}
\begin{proof}
  For $b\in\bit$ and $x\in\bit^n$, let $p_b(x):=\Pr_{F\sim\D_b}[F(x)=1]$. \
  By definition, $\epsilon(x):=p_0(x)-p_1(x)$ satisfies $|\epsilon(x)|\leq\epsilon$. \ 
  Then we consider the following hybrid distribution $\R$ defined as follows: \ 
  For each $x$, let
  \begin{align}
    \Pr_{F\sim\R}[F(x)=1]
    = \frac{\min\{p_0(x),p_1(x)\}}{1-|\epsilon(x)|}. \ 
  \end{align}
  Also for $b\in\bit$, define the following difference distributions $\T_b$
  \begin{align}
    \Pr_{F\sim\T_b}[F(x)=1] = \left\{
    \begin{array}{ll}
      1 & \text{if } p_b(x) > p_{1-b}(x) \\ 
      0 & \text{if } p_b(x) < p_{1-b}(x) \\
      1/2 & \text{otherwise.}
    \end{array}
    \right.
  \end{align}
  For every $x\in\bit^n$, 
  \begin{align}
    \Pr_{F\sim\D_b}[F(x)=1] = (1-|\epsilon(x)|) \Pr_{F\sim\R}[F(x)=1] + |\epsilon(x)|\Pr_{F\sim\T_b}[F(x)=1]. 
  \end{align}
  That is, sampling $F\sim\D_b$ can be done using the following process: \ 
  For each $x\in\bit^n$, first sample $X\in\bit$ according to the distribution $\Pr[X=1]=\Pr_{F\sim\R}[F(x)=1]$. \ 
  Then with probability $1-|\epsilon(x)|$, outputs $X$; with probability $|\epsilon(x)|$, output $1$ with probability $\Pr_{F\sim\T_b}[F(x)=1]$. \ %
  Thus $\D_b$ can be viewed as the distribution $\R$ with modification on each $x$ with probability at most $|\epsilon(x)|$. \ 

  Thus the rest of the proof is devoted to showing that for $b\in\bit$, 
  \begin{align}
    \left|\Exp_{F\sim\D_b}\|\A^F-B^F\|_\tr - \Exp_{F\sim\R}\|\A^F-\B^F\|_{\tr}\right| \leq 8T\sqrt\epsilon. \ 
  \end{align} 
  For every $T$-query algorithm $\A$ given oracle access to $F\sim\D_b$ described by unitaries $U_0,\ldots,U_T$, we use the above sampling process to get $F=F_0\oplus\Delta$ where $F_0\sim\R$ and $\Delta$ is a biased random function where for each $x\in\bit^n$, $\Delta(x)=1$ with probability at most $|\epsilon(x)|$. \ 
  For $b\in\bit$, let the distribution of $\Delta$ be $\P_b$. \ 
  For every $F_0,\Delta:\bit^n\to\bit$, consider the following hybrids:
  \begin{align}\nonumber
    \ket{\phi_T^\Delta} &= U_T \O(F_0) U_{T-1} \O(F_0)\ldots \O(F_0) U_1 \O(F_0) U_0\ket{0}, \\\nonumber
    \ket{\phi_{T-1}^\Delta} &= U_T \O(F) U_{T-1} \O(F_0)\ldots \O(F_0) U_1 \O(F_0) U_0\ket{0}, \\\nonumber
    \vdots \\\nonumber
    \ket{\phi_{0}^\Delta} &= U_T \O(F) U_{T-1} \O(F)\ldots \O(F) U_1 \O(F_0) U_0\ket{0}, \\
    \ket{\phi_{0}^\Delta} &= U_T \O(F) U_{T-1} \O(F)\ldots \O(F) U_1 \O(F) U_0\ket{0}.
  \end{align}
  That is, the state $\ket{\phi_t^{x}}$ is obtained by making $t$ queries to $F_0$ followed by $T-t$ queries to $F$.
  Also let $\ket{\psi_t^0}=U_t \O(F_0)\ldots U_1 \O(F_0) U_0\ket{0}$, which is independent of $\Delta$.

  For every function $\Delta:\bit^n\to\bit$, let $P_\Delta$ denote the projector onto the subspace $\{x:\Delta(x)=1\}$. \ 
  By triangle inequality, 
  \begin{align}\nonumber
    \|\ket{\phi_0^\Delta} - \ket{\phi_T^\Delta}\| 
    &\leq \sum_{t=0}^{T-1} \|\O(F)\ket{\psi_{t}^0}-\O(F_0)\ket{\psi_t^0}\| \\
    &= 2\sum_{t=0}^{T-1} \|P_\Delta\ket{\psi_t^0}\|. 
  \end{align}
  Since $\Delta(x)=1$ with probability at most $\epsilon$, 
  $\Exp_\Delta P_\Delta\leq \epsilon\Id$, and 
  \begin{align}\label{eq:qszk-prf-2}\nonumber
    \Exp_{\Delta\sim\P_b,F_0\sim\R}\|\A^{F_0\oplus\Delta}-\A^{F_0}\|_{\tr}
    &\leq 2\Exp_{\Delta\sim\P_b,F_0\sim\R}\|\ket{\phi_0^\Delta}-\ket{\phi_T^\Delta}\| \\\nonumber
    &\leq 4T\Exp_{F_0\sim\R}\|Q\|_{\op}^{1/2} \\
    &\leq 4T\sqrt{\epsilon}. \ 
  \end{align}
  where $Q=\Exp_{\Delta} P_\Delta$. 
  Since the same bound holds for $\B$, 
  \begin{align}\nonumber
    &\left|\Exp_{F\sim\D_b}\|\A^F-\B^F\|_\tr - \Exp_{F\sim\R}\|\A^F-\B^F\|_{\tr}\right| \\ \nonumber
    &\qquad=\left|\Exp_{F\sim\R}\left(\Exp_{\Delta\sim\P_b}\|\A^{F\oplus\Delta}-\B^{F\oplus\Delta}\|_\tr - \|\A^F-\B^F\|_{\tr}\right)\right| \\ \nonumber
    &\qquad\leq\Exp_{F\sim\R}\left(\Exp_{\Delta\sim\P_b}\|\A^{F\oplus\Delta}-\A^{F}\|_\tr + \Exp_{\Delta\sim\P_b}\|\B^{F\oplus\Delta}-\B^F\|_{\tr}\right)\\
    &\qquad \leq 8T\sqrt{\epsilon}. \ 
  \end{align}
  The first inequality holds by triangle inequality, and the second holds from \eq{qszk-prf-2}.
\end{proof}

Combining \lem{qszk-prf-1} and \lem{qszk-prf-2}, we prove the following theorem. \ 
\begin{theorem}
  For $\ell\geq 2n$, $L=2^\ell$, and every pair of $T$-query algorithms $\A,\B$, 
  \begin{align}
    \Pr_{\O\sim\F_{\ell+n}}\left[\left|\Exp_{k\sim\U_\ell}\|\A^{\O,\O_k}-\B^{\O,\O_k}\|_\tr - \Exp_{H\sim\F_n}\|\A^{\O,H}-\B^{\O,H}\|_{\tr}\right| > 16TN^{-1/4}\right] \leq 2^{-\Omega(N)}, \ 
  \end{align}
  where $\U_\ell$ is the uniform distribution over $\bit^\ell$ and $\F_m$ is the uniform distribution over $\bit^m\to\bit$. \ 
\end{theorem}
\begin{proof}
  For $\ell\geq 2n$, $L=2^\ell$, and $\epsilon=N^{-1/2}$, $L\geq N^2$ and $2\cdot 2^Ne^{-L\epsilon^2}=2\cdot 2^Ne^{-L/N}\leq 2^{-\Omega(N)}$. \ 
  By \lem{qszk-prf-2}, 
  \begin{align}\nonumber
    &\Pr_{\O\sim\F_{\ell+n}}\left[\left|\Exp_{k\sim\U_\ell}\|\A^{\O,\O_k}-\B^{\O,\O_k}\|_\tr - \Exp_{H\sim\F_n}\|\A^{\O,H}-\B^{\O,H}\|_{\tr}\right| \leq 16TN^{-1/4}\right] \\\nonumber
    &\qquad\geq
    \Pr_{\O\sim\F_{\ell+n}}\left[\forall x\in\bit^n,~ \left|\Exp_k\O_k(x)-1/2\right| \leq N^{-1/2}\right] \\\nonumber
    &\qquad\geq 1-2\cdot 2^Ne^{-L/N} \\
    &\qquad\geq 1 - 2^{-\Omega(N)}.
  \end{align}
\end{proof}
By the Borel-Cantelli Lemma, every pair of $\poly(n)$-query algorithms breaks the conditions in \asm{prfa} infintely often with probability 0. \  
\begin{corollary}
  \asm{prfa} holds relative to a random oracle with probability 1. \ 
\end{corollary}

\subsection{Pseudorandom Unitaries}

Similarly, we assume that there exists a pseudorandom unitary and give a construction relative to an oracle. \ 
The construction is the same as the one given by Kretschmer \cite{Kre21a}, and we strengthen the hardness to $\QSZK$. \ 
\begin{assumption}[Pseudorandom unitary assumption]\label{asm:qszk-prua}
  Let $\kappa\in\mathbb N$ be the security parameter, $\ell,m$ be polynomially bounded functions. \ 
  There exists a family of keyed unitaries $\{U_k\in\mathbb U(2^{m(\kappa)}):k\in\bit^{\ell(\kappa)}\}$ such that the following conditions hold.
  \begin{itemize}
      \item There exists a polynomial-time quantum algorithm $G$ that implements $U_k$ on input $k\in\bit^{\ell(\kappa)}$, i.e., on input $k$ and quantum state $\ket{\psi}\in\mathbb S(2^{m(\kappa)})$, $G(k,\ket{\psi})=U_k\ket{\psi}$.
      
      \item For every pair of quantum algorithms $\A,\B$ that makes $\poly(\kappa)$ queries, it holds that 
      \begin{align}
      \left|
      \Exp_{k\sim\K_{\ell(\kappa)}}\|\A^{U_k}-\B^{U_k}\|_{\tr}
      - \Exp_{U\sim\Haar(2^{m(\kappa)})} \|\A^{U}-\B^{U}\|_{\tr}
      \right| \leq \negl(\kappa).
    \end{align}
  \end{itemize}
\end{assumption}

We give a construction relative to the following oracle: Let $\kappa=n$, $N=2^n$, $L=2^\ell$, and $\O=\sum_{k\in\bit^\ell}\proj{k}\otimes \O_k$ where $\O_k\in\mathbb U(N)$ is a Haar random unitary for every $k\in\bit^\ell$. \  
Also let $\mu=\Haar(N)^{L}$ denote the measure of $\O$. \ 
We show that every pair of algorithms $\A,\B$ distinguishes $\O,\O_k$ under uniform $k$ from $\O,U$ for Haar random $U$ must make exponentially many queries in $\ell$. \
This implies that for $\ell=\Omega(n)$, the construction satisfies the conditions in \asm{qszk-prua}. 

Let $(P_1=P,P_0=\Id-P)$ be any binary measurement and $f(U):=\tr(\A^UP)$ denote the probability that measuring $\A^U$'s output state yields an outcome 1. \ 
We apply the following lemma by Kretschmer \cite{Kre21a}. \ 
\begin{lemma}[{\cite[Lemma~18]{Kre21a}, paraphrased}]\label{lem:qszk-lipschitz}
  Let $\A^U$ be any $T$-query algorithm and $f(U)=\tr(\A^UP)$ for $0\leq P\leq\Id$. \
  Then $f(U)$ is a $2T$-Lipschitz function in the Frobenius norm. \ 
\end{lemma}

We skip the proof of \lem{qszk-lipschitz} and only sketch the ideas. \
For a detailed proof, see \cite[Lemma~18]{Kre21a}. \ 
For unitaries $U,V$ satisfying $\|U-V\|_F=d$, their unitary channels have diamond distance at most $2d$. \ 
Thus by triangle inequality, $\|\A^U-\A^V\|_\diamond\leq 2Td = 2T \|U-V\|_F$ and thus $f(U)$ is $2T$-Lipschitz. \ 

For our purpose, we show that that for every $k$ and $U$ and pairs of $T$-query algorithms $\A$ and $\B$, the function
\begin{align}\label{eq:qszk-target}
  f(\O,k,U) := \|\A^{\O,\O_k}-\B^{\O,\O_k}\|_{\tr} - \|\A^{\O,U}-\B^{\O,U}\|_{\tr}
\end{align}
is a $2T$-Lipschitz function of $\O$. \ 
\begin{lemma}\label{lem:qszk-lipschitz-2}
  For every pair of $T$-query algorithms $\A,\B$, $k\in[L]$, and $U\in\mathbb U(N)$, $f(\O,k,U)$ as defined in \eq{qszk-target} is $8T$-Lipschitz. \ 
\end{lemma}
\begin{proof}
  It suffices to show that for each $k$, $\|\A^{\O,\O_k}-\B^{\O,\O_k}\|_{\tr}$ is a Lipschitz function of $\O$, and a similar argument applies to the second term in \eq{qszk-target}. \ 
  To see why, we observe that 
  \begin{align}\label{eq:qszk-lipschitz-3}
    \|\A^{\O,\O_k}-\B^{\O,\O_k}\|_{\tr} = \max_{0\leq P\leq\Id} \left(\tr(\A^{\O,\O_k}P) - \tr(\B^{\O,\O_k}P) \right).
  \end{align}
  By \lem{qszk-lipschitz}, for every $P$, $\tr(\A^{\O,\O_k}P) - \tr(\B^{\O,\O_k}P)$ is $4T$-Lipschitz, and so is the left side of \eq{qszk-lipschitz-3}. \ 
  Since for every $k\in\bit^\ell$ and $U\in\mathbb U(N)$, $f(\O,k,U)$ is a summation of two $4T$-Lipschitz functions, it is $8T$-Lipschitz. \ 
\end{proof}
\lem{qszk-lipschitz-2} implies that the advantage we aim to upper bound, i.e., 
\begin{align}\nonumber
  \adv(\O) 
  :=& \Exp_{k\sim\U_\ell,U\sim\Haar(N)}f(\O,k,U) \\
  =& \Exp_{k\sim\U_\ell}\|\A^{\O,\O_k}-\B^{\O,\O_k}\|_{\tr} - \Exp_{U\sim\Haar(N)} \|\A^{\O,U}-\B^{\O,U}\|_{\tr}
\end{align}
is $8T$-Lipschitz. \ 

Next, we prove a search lower bound, formally stated in \lem{qszk-search}, for showing that the average $\Exp_\O\adv(\O)$ is bounded with overwhelming probability over choices of $\O$. \ 
We denote $F_0$ the all-zero $\ell$-bit Boolean function and $F_x$ denote the function whose only 1-preimage is $x$, i.e., $F_x(y)=\Id[x=y]$. \
The following lemma will be useful later. \ 
\begin{lemma}\label{lem:bbbv}
  Let $\A$ be a $T$-query algorithm, $\D$ be the uniform distribution over $\{F_x:x\in\bit^\ell\}$, and $F_0$ be the zero function. \
  Then 
  \begin{align}
    \Exp_{F\sim\D}\|\A^F-\A^{F_0}\|_{\tr} \leq 4TL^{-1/2}.
  \end{align}
\end{lemma}
\begin{proof}
  Let the query algorithm be described with $T+1$ unitaries $U_{T},\ldots,U_0$, i.e., \ 
  \begin{align}
    \A^F\ket{0} = U_T \O(F)\ldots U_1 \O(F) U_0\ket{0}.
  \end{align}
  where $\O(F):=\sum_z (-1)^{F(z)}\proj{z}$ is the oracle. \ 
  Let $F_x$ denote the function whose only one-preimage is $x$. \ 
  
  Consider the sequence of states
  \begin{align}\nonumber
    \ket{\phi_T^x} &= U_T \O(F_0) U_{T-1} \O(F_0)\ldots \O(F_0) U_1 \O(F_0) U_0\ket{0}, \\\nonumber
    \ket{\phi_{T-1}^x} &= U_T \O(F_x) U_{T-1} \O(F_0)\ldots \O(F_0) U_1 \O(F_0) U_0\ket{0}, \\\nonumber
    \vdots \\\nonumber
    \ket{\phi_{0}^x} &= U_T \O(F_x) U_{T-1} \O(F_x)\ldots \O(F_x) U_1 \O(F_0) U_0\ket{0}, \\
    \ket{\phi_{0}^x} &= U_T \O(F_x) U_{T-1} \O(F_x)\ldots \O(F_x) U_1 \O(F_0) U_0\ket{0}.
  \end{align}
  That is, the state $\ket{\phi_t^{x}}$ is obtained by making $t$ queries to $F_0$ followed by $T-t$ queries to $F_x$.
  Also let $\ket{\psi_t^0}=U_t \O(F_0)\ldots U_1 \O(F_0) U_0\ket{0}$, which is independent of $x$.

  Let $P_x=\proj{x}$ denote the projection onto $\ket{x}$.
  By triangle inequality,
  \begin{align}\nonumber
    \|\ket{\phi_T^x}-\ket{\phi_0^x}\| 
    &\leq \sum_{t=0}^{T-1} \|\O(F_x)\ket{\psi_{t}^0}-\O(F_0)\ket{\psi_t^0}\| \\
    &\leq 2\sum_{t=0}^{T-1}\|P_x\ket{\psi_t^0}\|.
  \end{align}
  Note that 
  \begin{align}\nonumber
    \Exp_{F\sim\D}\|\A^F - \A^{F_0}\|_{\tr}
    &\leq 4\sum_{t=0}^{T-1} \Exp_{x\sim\U_n}\|P_x\ket{\psi_t^0}\| \\\nonumber
    &\leq 4\sum_{t=0}^{T-1} \left(\Exp_{x}\bra{\psi_t^0}P_x\ket{\psi_t^0}\right)^{1/2} \\
    &\leq 4T\cdot L^{-1/2}.
  \end{align}
  The first inequality holds by the fact that $\|\proj{\psi}-\proj{\phi}\|_{\tr}\leq 2\|\ket{\psi}-\ket{\phi}\|$ for two normalized states $\ket{\psi}$ and $\ket{\phi}$.
  The second inequality holds by Cauchy-Schwarz inequality, and the third holds since $\Exp_x P_x = \Id/L$.
\end{proof}

Next, we use \lem{bbbv} to prove the following search lower bound. \  %
\begin{lemma}\label{lem:qszk-search}
  Let $\A,\B$ be two $T$-query algorithms, $\D$ be the uniform distribution over $\{F_x:x\in\bit^\ell\}$, and $F_0$ be the zero function. \
  Then
  \begin{align}
    \left|\Exp_{F\sim\D}\|\A^F-\B^F\|_{\tr} - \|\A^{F_0}-\B^{F_0}\|_{\tr}\right| \leq 8TL^{-1/2}.
  \end{align}
\end{lemma}
\begin{proof}
  By \lem{bbbv},  
  \begin{align}
    \Exp_{F\sim\D}\|\A^F-\A^{F_0}\|_{\tr} \leq 4TL^{-1/2}, 
  \end{align}
  and a similar statement holds for $\B$.
  Then for each $F=F_x$ for $x\in\bit^\ell$,
  \begin{align}\nonumber
    \|\A^F-\B^F\|_{\tr} - \|\A^{F_0}-\B^{F_0}\|_{\tr}
    &\leq \|\A^F-\B^F-\A^{F_0}+\B^{F_0}\|_{\tr} \\
    &\leq \|\A^F-\A^{F_0}\|_{\tr} + \|\B^F-\B^{F_0}\|_{\tr}.
  \end{align}
  Since the above inequality is symmetric with repect to the exchange of $F$ and $F_0$, we also have 
  \begin{align}
    \|\A^{F_0}-\B^{F_0}\|_{\tr} - \|\A^{F}-\B^{F}\|_{\tr}
    &\leq \|\A^F-\A^{F_0}\|_{\tr} + \|\B^F-\B^{F_0}\|_{\tr}.
  \end{align}
  Thus we have 
  \begin{align}
    \Exp_{F\sim\D} \left[\left|\|\A^F-\B^F\|_{\tr} - \|\A^{F_0}-\B^{F_0}\|_{\tr}\right|\right] \leq 8TL^{-1/2}.
  \end{align}
  The rest of the proof follows from the triangle inequality. \ 
\end{proof}

With \lem{qszk-search}, we are ready to prove that the average of $\adv(\O)$ is exponentially small in $\ell$. \ 
Then we apply a concentration inequality to show $|\adv(\O)|$ is sharply concentrated around $|\Exp_{\O\sim\mu}\adv(\O)|$. \ 
\begin{lemma}
  For every pair of $T$-query algorithms $\A,\B$, $|\Exp_{\O\sim\mu}\adv(\O)|\leq 8TL^{-1/2}$.
\end{lemma}
\begin{proof}
  We prove the lemma by contrapositive.
  Assume that there exist two algorithms $\A,\B$
  \begin{align}
    \Exp_{\O\sim\mu}\adv(\O)
    = \Exp_{\O\sim\mu,k\sim\U_\ell}\|\A^{\O,\O_k}-\B^{\O,\O_k}\|_{\tr} - \Exp_{\O\sim\mu,U\sim\Haar(N)} \|\A^{\O,U}-\B^{\O,U}\|_{\tr} > 8TL^{-1/2}. \ 
  \end{align}
  Then we construct another pair of algorithms $\tilde\A,\tilde\B$ which break the bound in \lem{qszk-search}. \ 
  Given access to $F$ which is either the zero function or $F_x$ for uniform $x\in\bit^{\ell}$, $\tilde\A$ samples $U_0,\ldots,U_L\sim\Haar(N)$, and sets $V=U_0$ and $\O=\sum_{k\in[L]}\proj{k}\otimes V_k$, where $V_k=U_k$ if $F(k)=0$ and $V_k=U_0$ if $F(k)=1$. \ 
  The algorithm $\tilde\A$ runs $\A^{\O,V}$, and each query to $\O$ takes one query to $F$. \
  If $F=0$, then $\tilde\A^F$ outputs $\Exp_{\O\sim\mu,U\sim\Haar(N)}\proj{\O,U}\otimes\A^{\O,U}$ for independent $\O,U$;
  if $F=F_k$, then $\tilde\A^F$ outputs $\Exp_{\O\sim\mu,k\sim\U_\ell}\proj{\O,k}\otimes\A^{\O,\O_k}$. \ 
  The algorithm $\tilde\B$ runs exactly the same algorithm except that $\A$ is replaced with $\B$. \ 
  With $\tilde\A,\tilde\B$, we get a bound that violates \lem{qszk-search}. \

  Now we have shown that $\Exp_{\O\sim\mu}\adv(\O)\leq 8TL^{-1/2}$. \ 
  The same argument can be used to prove $-\Exp_{\O\sim\mu}\adv(\O)\leq 8TL^{-1/2}$, and thus we conclude the proof. \ 
\end{proof}
We have shown that $\adv(\O)$ is a $8T$-Lipschitz function and the absolute value of the average $|\Exp_{\O\sim\mu}\adv(\O)|$ is $O(TL^{-1/2})$. \ 
Now we apply the following concentration inequality to show that with overwhelming probability over choices of $\O$, $|\adv(\O)|\leq O(TL^{-1/2})$. \ 
\begin{theorem}[{\cite[Theorem~11]{Kre21a} and \cite[Theorem~5.17]{Mec19}}]\label{thm:qszk-concentration}
  For $N_1,\ldots,N_k\in\mathbb N$, let $X=\mathbb U(N_1)\oplus\ldots\oplus\mathbb U(N_k)$. \ 
  Let $\mu=\Haar(N_1)\times\ldots\Haar(N_k)$ be the product of the Haar measure on $X$. \ 
  Suppose that $f:X\to\mathbb R$ is $K$-Lipschitz in the Frobenius norm. \ 
  Then for every $t>0$, 
  \begin{align}
    \Pr_{\O\sim\mu}\left[f(\O) \geq \Exp_{\P\sim\mu}f(\P) + t \right] \leq \exp\left(-\frac{(N-2)t^2}{24K^2}\right),
  \end{align}
  where $N=\min\{N_1,\ldots,N_k\}$.
\end{theorem}

Since $\O=\sum_{k\in\bit^\ell}\proj{k}\otimes\O_k$ where $\O_k\sim\Haar(N)$ for each $k\in\bit^\ell$, $\adv(\O)$ is a function of $\{\O_k\}_{k\in\bit^\ell}$, and thus the concentration inequality can be applied. \ 
\begin{theorem}
  For $\ell=\lceil n-\log n\rceil$ and $L=2^\ell$, $\Pr_{\O\sim\mu}[|\adv(\O)| \geq 16TL^{-1/2}]\leq 2^{-\Omega(n)}$. \ 
\end{theorem}
\begin{proof}
  By \thm{qszk-concentration}, 
  \begin{align}
    \Pr_{\O\sim\mu}[\adv(\O) \geq 16TL^{-1/2}] \leq \exp\left(-\frac{(N-2)}{1536L}\right).
  \end{align}
  For $\ell=\lceil n-\log n\rceil$, the upper bound is $2^{-\Omega(n)}$. \ 
  The same argument can be used to show the event $-\adv(\O)\geq 16TL^{-1/2}$ occurs with probability $2^{-\Omega(n)}$. \ 
  Applying a union bound, we conclude the proof. \ 
\end{proof}

Applying the Borel-Cantelli Lemma, for $\ell=\lceil n-\log n\rceil$, the event that every pair of $\poly(n)$-query algorithms breaks the conditions in \asm{qszk-prua} occurs infinitely often with probability 0. \ 
\begin{corollary}
  There exists an oracle relative to which \asm{qszk-prua} holds with probability 1. \ 
\end{corollary}

\subsection{Randomness Expansion}\label{sec:expansion}

Under the assumptions introduced in \sec{qszk-prf}, we can reduce the input entropy to the device for any certified randomness protocol, and meanwhile, the output remains close to uniform. \ 
To see why, we recall the definition of quantum-proof randomness extractor~\cite{KR11,DPVR12}. \ 
\begin{definition}[{Quantum-proof strong extractor~\cite[Definition~3.2]{DPVR12}}]\label{dfn:quantum-extractor}
  A function $\Ext:\bit^n\times\bit^d\to\bit^m$ is a quantum-proof $(k,\epsilon)$-stronger extractor with uniform seed if for all state $\rho_{XE}$ classical on $X$ with $H_{\min}(X|E)\geq k$, and for uniform $Y$, 
  \begin{align}
    \|\rho_{\Ext(X,Y)YE} - \rho_{U_m}\otimes\rho_Y\otimes\rho_E\|_{\tr}\leq \epsilon,
  \end{align}
  where $\rho_{U_m}$ is the uniform distribution over $\bit^m$.
\end{definition}

If alternatively we only have the guarantee that $H_{\min}^{\delta}(X|E)\geq k$, for sufficiently small but nonzero $\delta$, the extractor still outputs a distribution close enough to the uniform distribution conditioned on the side information $E$ by triangle inequality. \ 
\begin{lemma}[{\cite[Lemma~3.5]{DPVR12}}]
  If $\Ext:\bit^n\times\bit^d\to\bit^m$ is a quantum-proof $(k,\epsilon)$-strong extractor, then for any $\rho_{XE}$ and $\delta>0$ with $H_{\min}^\delta(X|E)_\rho\geq k$,
  \begin{align}
    \|\rho_{\Ext(X,Y)YE} - \rho_{U_m}\otimes\rho_Y\otimes\rho_E\|_{\tr}\leq \epsilon + 2\delta.
  \end{align}
\end{lemma}

If $E$ is restricted to a classical state, then we say an extractor that satisfies \dfn{quantum-extractor} a classical-proof $(k,\epsilon)$-strong extractor. \ 
Furthermore, if \dfn{quantum-extractor} holds in the special case that $E$ is empty, we call such an extractior a $(k,\epsilon)$-stronger extractor. \ 

In the same paper, De, Portmann, Vidick, and Renner \cite{DPVR12} showed that Trevisan's extractor \cite{Tre01} is a quantum-proof strong extractor. \ 
We take relevant definitions from \cite{DPVR12}. \
\begin{definition}[{Weak design~\cite[Definition~4.1]{DPVR12}}]\label{dfn:weak-design}
  For integer $d$, the family of sets $S_1,\ldots,S_m\subset[d]$ is a weak $(t,r)$-design if 
  \begin{enumerate}
  \item For every $i\in[m]$, $|S_i|=t$.
  \item For every $i\in[m]$, $\sum_{j=1}^{i-1}2^{|S_i\cap S_j|}\leq rm$.
  \end{enumerate}
\end{definition}
\begin{definition}[{Trevisan's extractor \cite[Definition~4.2]{DPVR12}\cite{Tre01}}]\label{dfn:trevisan}
  For a one-bit extractor $C:\bit^n\times\bit^t\to\bit$, which uses a (not necessarily uniform) seed of length $t$, and for a weak $(t,r)$-design, $S_1,\ldots,S_m\subset[d]$, an $m$-bit extractor $\Ext_C:\bit^n\times\bit^d\to\bit^m$ is defined as
  \begin{align}
    \Ext_C(x,y) := C(x,y_{S_1})\ldots C(x,y_{S_m}).
  \end{align}
\end{definition}

The integer $d$ in \dfn{weak-design} is the length of the seed of the extractor $\Ext_C$ and depends on $t$, the size of the seed of the 1-bit extractor $C$. \ 
The size of the seed will always be $d=\poly(\log n)$ if the error $\epsilon=\poly(1/n)$ \cite{DPVR12}. \ 
In the same paper, it is also proved that that a single-bit strong extractor is also a quantum-proof strong extractor. \ 
\begin{theorem}[{\cite[Theorem~4.7]{DPVR12}}]
  Let $C:\bit^n\to\bit^t\to\bit$ be a $(k,\epsilon)$-strong extractor with an $s$-bit seed---i.e., the seed needs at least $s$ bits of min-entropy---and $S_1,\ldots,S_m\subset[d]$ a weak $(t,r)$-design.
  Then the extractor defined in \dfn{trevisan} is a quantum proof $(k+rm+\log(1/\epsilon),6m\sqrt{\epsilon})$-strong extractor for any seed with min-entropy $d-(t-s-\log\frac{1}{3\sqrt\epsilon})$.
\end{theorem}

Furthermore, there exists a $(k,\epsilon)$-strong extractor by Raz, Reingold and Vadhan \cite{RRV99}, and the input seed has length $\poly\log(n)$.
\begin{theorem}[{\cite[Proposition~5.2]{DPVR12}}]
  For any $\epsilon>0$ and $n\in\mathbb N$, there exists a $(k,\epsilon)$-strong extractor with uniform seed $\Ext_{n,\epsilon}:\bit^n\times\bit^t\to\bit$ with $t=O(\log(n/\epsilon))$ and $k=3\log(1/\epsilon)$. 
\end{theorem}

Combining the fact that there exists a quantum-proof strong extractor and \asm{prfa}, we prove the following theorem.
\begin{theorem}\label{thm:qszk-main}
  Assume that \asm{llha} and \asm{prfa} holds. \ 
  There exists a randomness expansion protocol for a device which solve $\LXEB_{b,k}$ for $99\%$ of the given circuits and achieves polynomial expansion for $b=1.02$ and $k=O(n^2)$, where $n$ is the number of qubits each circuit acts on. \
\end{theorem}
\begin{proof}
In our $m$-round randomness expansion protocol in \fig{llqsv} using a random circuit, we change the circuits at least $\gamma m=O(\log n)$ times, and each can be replaced with pseudorandom circuits that can be computed using $\ell=\poly(n)$ bits as input (used for generating the key of the pseudorandom function). \ 
The protocol generates a distribution $(\epsilon+2\epsilon'+\negl(n))$-close to an $\Omega(m)$-bit uniform distribution for devices solving $\LXEB_{1+\delta,k}$ for constant $\delta$ and $k=O(n^2)$, using an input of $O(\ell\log n)$ uniformly random bits. \ 
Setting $m=\ell^c$ for $c\geq 2$, the protocol accumulates net entropy.
\end{proof}
Applying the same reasoning as in the proof of \thm{qszk-main}, we conclude that \asm{qszk-prua} implies that there exists a randomness expansion protocol. \ 
\begin{theorem}
  Assume that \asm{qszk-prua} holds. \ 
  There exists a randomness expansion protocol for a device that is given oracle access to a pseudorandom unitary and solves $\LXEB_{b,k}$ for $99\%$ of the given circuits and achieves polynomial expansion for $b=1.02$ and $k=O(n^2)$, where $n$ is the number of qubits each circuit acts on. \ 
\end{theorem}

\section*{Acknowledgements}

We thank Fernando Brand{\~a}o, Alex Halderman, William Kretschmer, John Martinis, Carl Miller, Ron Peled, Ren\'{e} Peralta, Or Sattath, Thomas Vidick, and David Zuckerman for helpful discussions.

\bibliographystyle{plain}
\bibliography{refs}

\end{document}